\documentclass[english]{article}
\usepackage[T1]{fontenc}
\usepackage[utf8]{inputenc}
\usepackage{colortbl}
\usepackage{array}
\usepackage{float}
\usepackage{booktabs}
\usepackage{multirow}
\usepackage{amsmath}
\usepackage{amsthm}
\usepackage{amssymb}
\usepackage{rotating}
\usepackage[authoryear]{natbib}

\makeatletter

\providecommand{\tabularnewline}{\\}
\floatstyle{ruled}
\newfloat{algorithm}{tbp}{loa}
\providecommand{\algorithmname}{Algorithm}
\floatname{algorithm}{\protect\algorithmname}

\usepackage{fullpage}
\usepackage[utf8]{inputenc} 
\usepackage[T1]{fontenc}    
\usepackage{hyperref}       
\usepackage{url}            
\usepackage{booktabs}       
\usepackage{xcolor}         
\usepackage{algorithmic} 

\makeatother

\theoremstyle{plain}
\newtheorem{thm}{\protect\theoremname}
\newtheorem{lem}[thm]{\protect\lemmaname}
\theoremstyle{definition}
\newtheorem{defn}[thm]{\protect\definitionname}
\theoremstyle{remark}
\newtheorem{claim}[thm]{\protect\claimname}
\theoremstyle{plain}
\newtheorem{cor}[thm]{\protect\corollaryname}
\theoremstyle{remark}
\newtheorem{rem}[thm]{\protect\remarkname}
\usepackage{babel}
\providecommand{\claimname}{Claim}
\providecommand{\corollaryname}{Corollary}
\providecommand{\definitionname}{Definition}
\providecommand{\lemmaname}{Lemma}
\providecommand{\remarkname}{Remark}
\providecommand{\theoremname}{Theorem}

\begin{document}
\global\long\def\threshold{\mathrm{Threshold}}%
\global\long\def\N{{\cal N}}%
\global\long\def\R{\mathbb{R}}%

\global\long\def\diag{\mathrm{diag}}%
\global\long\def\rank{\mathrm{rank}}%
\global\long\def\range{\mathrm{range}}%

\global\long\def\cov{\mathrm{Cov}}%
\global\long\def\E{\mathbb{E}}%
\global\long\def\poly{\mathrm{poly}}%
\global\long\def\sign{\mathrm{sign}}%
\global\long\def\Head{\mathrm{Head}}%
\global\long\def\Tail{\mathrm{Tail}}%

\global\long\def\Lap{\operatorname{Lap}}%

\global\long\def\sym{\operatorname{sym}}%

\global\long\def\Clip{\operatorname{Clip}}%

\global\long\def\tr{\operatorname{tr}}%

\global\long\def\gclip{\mathcal{G}_{\mathrm{clip}}}%

\global\long\def\gmart{\mathcal{G}_{\mathrm{mart}}}%

\global\long\def\ginit{\mathcal{G}_{\mathrm{init}}}%

\global\long\def\calT{\mathcal{T}}%

\global\long\def\calF{\mathcal{F}}%

\global\long\def\calN{\mathcal{N}}%

\allowdisplaybreaks
\title{Differentially Private Preconditioned Power Iteration Method for Principal Component Analysis that Adapts to the Coherence
}
\author{
  Duy V. Le\thanks{Department of Computer Science, Boston University, \texttt{duyle@bu.edu}.} \and
  Ta Duy Nguyen\thanks{\texttt{taduy@bu.edu}.} \and
  Alina Ene\thanks{Department of Computer Science, Boston University, \texttt{aene@bu.edu}. This work was supported in part by a grant of access to OpenAI models through the ChatGPT for Academic Researchers program.} \and
  Huy L. Nguyen\thanks{Khoury College of Computer Sciences, Northeastern University, \texttt{hu.nguyen@northeastern.edu}. This work was supported in part by NSF CCF 2311649 and a gift from Apple Inc.}
}
\maketitle
\begin{abstract}
In Principal Component Analysis (PCA), we are given as input a matrix
$A\in\mathbb{R}^{n\times d}$ and an integer $k\leq d$, and the goal
is to find the top $k$ right-singular vectors of $A$. We consider
the setting where each row of $A$ is the private data of an individual,
and we design an $\left(\epsilon,\delta\right)$-differentially private
algorithm whose utility guarantee improves with the coherence of the
input matrix. Prior works were either restricted to the case $k=1$
or ensured significantly weaker privacy protections. Our algorithm
is based on a preconditioned variant of the power iteration method
that is designed to significantly reduce the noise needed to ensure
privacy.

\end{abstract}

\section{Introduction}

Principal component analysis (PCA) is a fundamental tool for dimensionality
reduction, data visualization, and low-rank approximation \citep{blum2020foundations}.
Given a data matrix $A\in\mathbb{R}^{n\times d}$, PCA finds a low-dimensional
subspace that captures as much of the variation in the data as possible.
If $A=U\Sigma V^{\top}$ is the singular value decomposition of $A$,
the optimal rank-$k$ subspace is spanned by the top $k$ right singular
vectors $V_{k}=[v_{1}\ \cdots\ v_{k}]$. It is also the rank $k$
subspace that maximizes the total squared lengths of the projections
of the data.

When the rows of $A$ contain information about individuals, releasing
a subspace computed from $A$ can reveal information about a particular
row. This motivates differentially private algorithms for PCA. We
consider row-level privacy, where each row represents the data of
one individual. We assume that every row has $\ell_{2}$ norm at most
$1$, and two matrices are adjacent if one row is replaced. Normalizing
the lengths of the rows is natural because we would like to measure
the correlation with the entire dataset and no single data point should
have an outsize contribution. The goal is to estimate $V_{k}$ while
protecting the privacy of every individual row. In this model, \citet{DBLP:conf/stoc/DworkTT014}
give a private rank-$k$ PCA algorithm whose subspace error scales
with $\widetilde{O}\left(\sqrt{d}\right)$. Thus, the error can be
large when the number of dimensions $d$ is large. This motivates
utility guarantees that exploit additional structure of the input
matrix. One such structural parameter is coherence, which is related
to the largest absolute entries of the left singular vectors $U$.
Intuitively, coherence captures the difficulty in privatizing the
singular vectors because it measures the maximum correlation between
the singular vectors and individual data rows. The more the singular
vectors reveal about one row, the more perturbation we need to protect
privacy. On the other hand, for many natural datasets, such as when
the data is close to i.i.d. data generated from a generative model,
the coherence can be shown to be small. Previous work has shown that
coherence-dependent guarantees can substantially improve on worst-case
bounds on low-coherence data \citep{DBLP:conf/stoc/HardtR12,DBLP:conf/stoc/HardtR13,DBLP:conf/nips/HardtP14,DBLP:conf/colt/BalcanDWY16,NicolasSMMC25,dOrsiN26,NguyenEN26}.

However, most earlier coherence-dependent guarantees were obtained
under privacy models different from the row-level $\ell_{2}$ model
considered in our paper. For example, \citet{DBLP:conf/stoc/HardtR13}
consider entry-level adjacency, under which neighboring matrices differ
in a single entry by at most $1$. Subsequent rank-$k$ results use
structured adjacency models that, when restricted to a one-row change,
impose a $\ell_{1}$-norm restriction on the changed row \citep{NicolasSMMC25,dOrsiN26}.
These models do not cover all neighboring datasets in our setting,
since a row $a$ satisfying $\|a\|_{2}\leq1$ can have $\|a\|_{1}$
as large as $\sqrt{d}$. \citet{NguyenEN26} recently gave the first
coherence-adaptive guarantee under row-level $\ell_{2}$ privacy.
However, their algorithm and analysis are specific to estimating a
single leading singular vector, leaving the analogous rank-$k$ problem
open. 

In this work, we give a coherence-adaptive algorithm for estimating
the top-$k$ right singular subspace under row-level $\ell_{2}$ differential
privacy. Our result extends the coherence-adaptive utility guarantee
of \citet{NguyenEN26} from the leading singular vector to the top-$k$
singular subspace while retaining the same privacy model. Our algorithm
is based on a preconditioned and clipped variant of block power iteration.
It does not require the coherence or the ratio $\sigma^{2}_{k+1}/\sigma^{2}_{k}$
as input, and its utility improves when the input has low coherence.
The prior work is also suboptimal when the ratio $\sigma^{2}_{2}/\sigma^{2}_{1}$
is very close to 1 whereas our algorithm remains near optimal for
all ratios $\sigma^{2}_{2}/\sigma^{2}_{1}\in[1/(d+1),1)$. 

\subsection{Our contributions}

\begin{table*}
\caption{Comparison of deterministic rank-$k$ PCA guarantees. The error is
$\|(I-\widehat{X}\widehat{X}^{\top})V_{k}\|^{2}_{2}=\sin^{2}\left(\theta_{\max}\right)$
where $\theta_{\max}$ is the largest principal angle between the
column spaces of $\widehat{X}$ and $V_{k}$. Poly-logarithmic factors
in $n$, $d$, $\frac{1}{\epsilon}$, $\frac{1}{\delta}$ are suppressed.
We let $\Delta_{k}:=\sigma^{2}_{k}-\sigma^{2}_{k+1}.$}\label{tab:result}

\begin{centering}
\begin{tabular*}{1\textwidth}{@{\extracolsep{\fill}}>{\raggedright}m{0.007\textwidth}>{\raggedright}m{0.08\textwidth}>{\raggedright}m{0.075\textwidth}>{\raggedright}m{0.36\textwidth}>{\raggedright}m{0.3\textwidth}}
\toprule 
 & Algorithm & Privacy model & $\sin^{2}$-Error & Condition\tabularnewline
\midrule
\midrule 
\multirow{4}{0.007\textwidth}[-0.8cm]{\begin{turn}{90}
\raggedright{}Rank $k$
\end{turn}} & \citet{DBLP:conf/nips/HardtP14} & Entry & $\widetilde{O}\left(\frac{k^{3}\sigma^{4}_{k}\max\left\{ n\|U\|^{2}_{\infty},d\|V\|^{2}_{\infty}\right\} }{\epsilon^{2}\Delta^{3}_{k}}\right)$ & None\tabularnewline
\cmidrule{2-5}
 & \citet{dOrsiN26} & See note\footnotemark[1] & $\widetilde{O}\left(\frac{k\sigma^{2}_{k}\max\left\{ \frac{n}{k}\|U_{k}U^{\top}_{k}\|_{\infty},\frac{d}{k}\|V_{k}V^{\top}_{k}\|_{\infty}\right\} }{\epsilon^{2}\Delta^{2}_{k}}\right)$\footnotemark[2] & None\tabularnewline
\cmidrule{2-5}
 & \citet{DBLP:conf/stoc/DworkTT014} & Row-$\ell_{2}$ & $\widetilde{O}\left(\frac{d}{\epsilon^{2}\Delta^{2}_{k}}\right)$ & $\Delta_{k}=\widetilde{\Omega}\left(\frac{\sqrt{d}}{\epsilon}\right)$\tabularnewline
\cmidrule{2-5}
 & Our algorithm & Row-$\ell_{2}$ & $\widetilde{O}\left(\frac{dk\sigma^{2}_{k}\|U\|^{2}_{\infty}+1}{\epsilon^{2}\Delta^{2}_{k}}\right)$ & $\Delta_{k}=\widetilde{\Omega}\Big(\frac{\sqrt{d}}{\epsilon}+\left(\frac{\sigma_{k}\sqrt{k}}{\epsilon}\right)^{2/3}+\left(\frac{k\sigma^{2}_{k}\left\Vert U\right\Vert _{\infty}}{\epsilon}\right)^{2/3}\Big)$\tabularnewline
\midrule 
\multirow{2}{0.007\textwidth}[0.3cm]{\begin{turn}{90}
Rank 1
\end{turn}} & \citet{NguyenEN26} & Row-$\ell_{2}$ & $\widetilde{O}\left(\frac{\left(\min\{4n,d\sigma^{2}_{1}\}+\frac{\sigma^{4}_{1}}{\Delta_{1}}+d\Delta_{1}\right)\|U_{1}\|^{2}_{\infty}}{\epsilon^{2}\Delta^{2}_{1}}\right)$\footnotemark[3] & $\Delta_{1}=\widetilde{\Omega}\Big(\left(\frac{\sigma^{2}_{1}\left\Vert U_{1}\right\Vert _{\infty}}{\epsilon}\right)^{2/3}+\left(\frac{\sigma_{1}}{\epsilon}\right)^{2/3}\Big)$\footnotemark[3]\tabularnewline
\cmidrule{2-5}
 & Our algorithm & Row-$\ell_{2}$ & $\widetilde{O}\left(\frac{d\sigma^{2}_{1}\|U\|^{2}_{\infty}+1}{\epsilon^{2}\Delta^{2}_{1}}\right)$ & $\Delta_{1}=\widetilde{\Omega}\Big(\frac{\sqrt{d}}{\epsilon}+\left(\frac{\sigma^{2}_{1}\left\Vert U\right\Vert _{\infty}}{\epsilon}\right)^{2/3}+\left(\frac{\sigma_{1}}{\epsilon}\right)^{2/3}\Big)$\tabularnewline
\bottomrule
\end{tabular*}
\par\end{centering}
\end{table*}

\footnotetext[1]{When a single row changes, they require the $\ell_1$ norm of the change to be bounded.
}

\footnotetext[2]{We can bound $\max\left\{ \frac{n}{k}\|U_{k}U_{k}^{\top}\|_{\infty}, \frac{d}{k}\|V_{k}V_{k}^{\top}\|_{\infty} \right\} \le \max\left\{ n\|U\|^{2}_{\infty},d\|V\|^{2}_{\infty}\right\}$ to compare with other bounds.
}

\footnotetext[3]{$\|U_{1}\|_{\infty}$ is the maximum magnitude of the coordinates of the first singular vector in $U$.}

Let $\Delta_{k}:=\sigma^{2}_{k}-\sigma^{2}_{k+1}$ and let $\left\Vert (I-\widehat{X}\widehat{X}^{\top})V_{k}\right\Vert _{2}$
denote the error of an output $\widehat{X}\in\mathbb{R}^{d\times k}$
with orthonormal columns. Let $a^{\top}_{i}$ denote the $i$th row
of $A$, and recall that $\|a_{i}\|_{2}\leq1$. We write $U_{k}=U_{:,1:k}$,
$\Sigma_{k}=\diag\left(\sigma_{1},\ldots,\sigma_{k}\right)$, $\left\Vert U\right\Vert _{\infty}=\max_{i,j}\left|U_{ij}\right|$.
For the $i$th standard basis vector $e_{i}$, we have $\left\Vert e^{\top}_{i}U_{k}\right\Vert _{2}\leq\sqrt{k}\left\Vert U\right\Vert _{\infty}$.
Our main result is the following:
\begin{thm}
\label{thm:main-informal} Let $\epsilon,\delta,\beta\in(0,1)$, and
suppose that every row of $A$ has $\ell_{2}$ norm at most $1$.
There is an $(\epsilon,\delta)$-differentially private algorithm
that, under the condition 
\[
\Delta_{k}\geq\widetilde{\Omega}\left(\frac{\sqrt{d}}{\epsilon}+\left(\frac{\sigma_{k}\sqrt{k}}{\epsilon}\right)^{2/3}+\left(\frac{k\sigma^{2}_{k}\left\Vert U\right\Vert _{\infty}}{\epsilon}\right)^{2/3}\right),
\]
outputs $\widehat{X}\in\mathbb{R}^{d\times k}$ satisfying, with probability
at least $1-\beta$, 
\[
\left\Vert (I-\widehat{X}\widehat{X}^{\top})V_{k}\right\Vert ^{2}_{2}\leq\left(\frac{\beta\delta\epsilon}{n}\right)^{2}+\widetilde{O}\left(\frac{dk\sigma^{2}_{k}\left\Vert U\right\Vert ^{2}_{\infty}+1}{\epsilon^{2}\Delta^{2}_{k}}\right).
\]
The algorithm does not require $\left\Vert U\right\Vert _{\infty}$,
$\Delta_{k}$, or $\sigma^{2}_{k+1}/\sigma^{2}_{k}$ as input.
\end{thm}

\subsection{Comparison with prior works}

We additionally write $\|V\|_{\infty}:=\max_{i,j}|V_{ij}|$. We summarize
our results in comparison with prior work in Table \ref{tab:result}.

\textbf{Comparison with the row-$\ell_{2}$ baseline.} The most direct
general rank-$k$ comparison is the Analyze Gauss algorithm of \citet{DBLP:conf/stoc/DworkTT014},
since it provides privacy under the same row-$\ell_{2}$ model. Our
result replaces the worst-case dependence on $d$ by a dependence
on the coherence of the input. Subject to the $\Delta_{k}$ condition
in Theorem $\ref{thm:main-informal}$, our bound improves on Analyze
Gauss when $k\sigma^{2}_{k}\|U\|^{2}_{\infty}$ is small. Analyze
Gauss does not require low coherence and can give the stronger guarantee
outside this regime. 

\textbf{Comparison under other privacy models.} The results of \citet{DBLP:conf/nips/HardtP14,dOrsiN26}
provide coherence-dependent guarantees for general rank $k$, but
they use privacy models different from ours. \citet{DBLP:conf/nips/HardtP14}
consider entry-level adjacency. Any replacement between two rows of
$\ell_{2}$ norm at most $1$ can be implemented using at most $2d$
entry changes by passing through the zero row. Group privacy therefore
converts their guarantee to our row-level model by setting $\epsilon_{\mathrm{ent}}=\Theta\left(\epsilon/d\right)$
and adjusting $\delta_{\mathrm{ent}}$ accordingly. So \citet{DBLP:conf/nips/HardtP14}
have a converted $\sin^{2}$ bound of
\[
\widetilde{O}\left(\frac{k^{3}\sigma^{4}_{k}\max\left\{ n\|U\|^{2}_{\infty},d\|V\|^{2}_{\infty}\right\} }{\left(\epsilon/d\right)^{2}\Delta^{3}_{k}}\right)=\widetilde{O}\left(\frac{d^{2}k^{3}\sigma^{4}_{k}\max\left\{ n\|U\|^{2}_{\infty},d\|V\|^{2}_{\infty}\right\} }{\epsilon^{2}\Delta^{3}_{k}}\right).
\]
The utility guarantee can be improved by a factor of $d$ by adapting
their analysis.

For the model considered by \citet{dOrsiN26}, the replacement of
row $a$ by $a'$ will result in a change of $O\left(\sqrt{d}\right)$,
leading to a $\sin^{2}$ error of
\[
\widetilde{O}\left(\frac{dk\sigma^{2}_{k}\max\left\{ \frac{n}{k}\|U_{k}U^{\top}_{k}\|_{\infty},\frac{d}{k}\|V_{k}V^{\top}_{k}\|_{\infty}\right\} }{\epsilon^{2}\Delta^{2}_{k}}\right)=\widetilde{O}\left(\frac{dk\sigma^{2}_{k}\max\left\{ n\|U\|^{2}_{\infty},d\|V\|^{2}_{\infty}\right\} }{\epsilon^{2}\Delta^{2}_{k}}\right).
\]
The above bound can be a factor $n$ larger than ours depending on
the comparison between $\|U_{k}U^{\top}_{k}\|_{\infty}$ and $\|U\|^{2}_{\infty}$.
Note that while this bound does not require any condition, it is only
meaningful if the fraction inside $\widetilde{O}$ is less than 1,
which implies a lower bound on $\Delta_{k}$. 

\textbf{Comparison for $k=1$.} The result of \citet{NguyenEN26}
uses the same row-$\ell_{2}$ privacy model as ours and is the most
relevant comparison when $k=1$.  The coherence parameter considered
in this work is $\max_{i}\left|U_{i,1}\right|$, which can be smaller
than our coherence parameter $\max_{i,j}\left|U_{i,j}\right|$. Our
theorem additionally requires $\Delta_{1}=\widetilde{\Omega}(\sqrt{d}/\epsilon)$
which can be more stringent than their condition. Assuming that the
conditions of both results hold and ignoring logarithmic factors,
one setting where our bound improves upon theirs is $\kappa:=\frac{\Delta_{1}}{\sigma^{2}_{1}}\ll\frac{\sigma^{2}_{1}\|U_{1}\|^{2}_{\infty}}{d\sigma^{2}_{1}\|U\|^{2}_{\infty}+1}$.
When $\|U\|_{\infty}=\Theta(\|U_{1}\|_{\infty})$, this condition
is equivalent to $\kappa\ll\min\left\{ \frac{1}{d},\,\sigma^{2}_{1}\|U_{1}\|^{2}_{\infty}\right\} .$
In many natural settings (e.g. all rows have length $\Theta(1)$),
we have $\sigma^{2}_{1}=\Omega(n/d),\|U_{1}\|^{2}_{\infty}\ge1/n$
so the condition simplifies to $\kappa\ll1/d$.

\subsection{Our approach and techniques}

Let $a^{\top}_{i}$ denote the $i$th row of $A$, and recall that
$\|a_{i}\|_{2}\leq1$. We write $U_{k}=U_{:,1:k}$, $\Sigma_{k}=\diag\left(\sigma_{1},\ldots,\sigma_{k}\right)$,
$\left\Vert U\right\Vert _{\infty}=\max_{i,j}\left|U_{ij}\right|$.
For the $i$th standard basis vector $e_{i}$, we have $\left\Vert e^{\top}_{i}U_{k}\right\Vert _{2}\leq\sqrt{k}\left\Vert U\right\Vert _{\infty}$.
Thus, when $\left\Vert U\right\Vert _{\infty}$ is small, every row
of $U_{k}$ has small norm.

A standard noisy power method \citep{DBLP:conf/nips/HardtP14} maintains
a matrix $X^{(t)}\in\mathbb{R}^{d\times k}$ with orthonormal columns
and performs an update of the form 
\begin{align*}
Z^{(t+1)} & =A^{\top}AX^{(t)}+\text{Gaussian noise},\\
X^{(t+1)}R^{(t+1)} & \text{be QR-decomposition of }Z^{(t+1)}
\end{align*}
Since $A^{\top}AX^{(t)}=\sum^{n}_{i=1}a_{i}\left(a^{\top}_{i}X^{(t)}\right)$,
the contribution of row $a_{i}$ to the update has Frobenius norm
$\left\Vert a_{i}\left(a^{\top}_{i}X^{(t)}\right)\right\Vert _{F}=\left\Vert a_{i}\right\Vert _{2}\left\Vert a^{\top}_{i}X^{(t)}\right\Vert _{2}$. 

A direct implementation under row-level privacy would therefore calibrate
the noise using bounds on these row contributions. The difficulty
is that these contributions are not controlled by coherence alone.
To see the difficulty, suppose that $X^{(t)}$ has already reached
the target $V_{k}$, so that $X^{(t)}=V_{k}$. The contribution of
row $a_{i}$ then has norm 
\[
\left\Vert a_{i}\left(a^{\top}_{i}V_{k}\right)\right\Vert _{F}=\left\Vert a_{i}\left(e^{\top}_{i}AV_{k}\right)\right\Vert _{F}=\left\Vert a_{i}e^{\top}_{i}U_{k}\Sigma_{k}\right\Vert _{F}=\left\Vert a_{i}\right\Vert _{2}\left\Vert e^{\top}_{i}U_{k}\Sigma_{k}\right\Vert _{2}
\]
Low coherence controls $\left\Vert e^{\top}_{i}U_{k}\right\Vert _{2}\leq\sqrt{k}\left\Vert U\right\Vert _{\infty}$,
but the direct update contains 
\[
\left\Vert e^{\top}_{i}U_{k}\Sigma_{k}\right\Vert _{2}=\left(\sum^{k}_{j=1}\sigma^{2}_{j}U^{2}_{ij}\right)^{1/2}\leq\sigma_{1}\left\Vert e^{\top}_{i}U_{k}\right\Vert _{2}\leq\sigma_{1}\sqrt{k}\left\Vert U\right\Vert _{\infty}
\]
This gives $\left\Vert a_{i}\left(a^{\top}_{i}V_{k}\right)\right\Vert _{F}\leq\sigma_{1}\sqrt{k}\|U\|_{\infty}$.
Thus, the Gaussian noise adjusted to the above upper bound inherits
a dependence on $\sigma_{1}$. On the other hand, in the normal power
method, the subspace error contracts at rate $\sigma^{2}_{k+1}/\sigma^{2}_{k}$.
This creates the main difficulty in analyzing the above power method:
the bound on Gaussian noise can depend on $\sigma_{1}$, whereas the
recurrence for the subspace error depends on $\sigma_{k}$ and $\sigma_{k+1}$.
This is not an issue for the prior work for $k=1$ but it is a major
issue for our setting.

\textbf{Preconditioner.} Our main idea is to remove the $\Sigma_{k}$
term on the contribution of row $a_{i}$, using the preconditioner
method. Our new update has this form:
\begin{align*}
H^{(t)} & =X^{(t)\top}A^{\top}AX^{(t)}+\text{symmetric Gaussian noise},\\
\widetilde{H}^{(t)} & =H^{(t)}\text{ with every eigenvalue below \ensuremath{\gamma} replaced by \ensuremath{\gamma}},\\
B^{(t)} & =\left(\widetilde{H}^{(t)}\right)^{-1/2},\\
Z^{(t+1)} & =A^{\top}AX^{(t)}B^{(t)}+\text{Gaussian noise},\\
X^{(t+1)}R^{(t+1)} & \text{be QR-decomposition of }Z^{(t+1)}
\end{align*}
where the second step is to ensure that the inverse of $\widetilde{H}^{(t)}$
exists and $B^{(t)}$ is defined. The parameter $\ensuremath{\gamma}$
is set roughly the same as the magnitude of the symmetric Gaussian
noise so that we remove the artifact from the noise and the components
smaller than the noise.

To explain the role of the preconditioner $B^{(t)}$, first suppose
that the Gaussian noise is absent and $\widetilde{H}^{(t)}=H^{(t)}=X^{(t)\top}A^{\top}AX^{(t)}$.
Provided that $X^{(t)\top}A^{\top}AX^{(t)}$ is positive definite,
we have $B^{(t)}=\left(X^{(t)\top}A^{\top}AX^{(t)}\right)^{-1/2}$,
and the columns of $AX^{(t)}B^{(t)}$ are orthonormal. Moreover, since
$B^{(t)}$ is invertible, $\mathrm{range}\left(A^{\top}AX^{(t)}B^{(t)}\right)=\mathrm{range}\left(A^{\top}AX^{(t)}\right)$.
Thus, in this idealized iteration, the preconditioner does not change
the next subspace produced by the block power method. Instead, it
changes the basis in which $a^{\top}_{i}X^{(t)}B^{(t)}$ is measured.

At $X^{(t)}=V_{k}$, we have $B^{(t)}=\left(V^{\top}_{k}A^{\top}AV_{k}\right)^{-1/2}=\Sigma^{-1}_{k}$
and therefore $a^{\top}_{i}V_{k}B^{(t)}=e^{\top}_{i}U_{k}$. The contribution
of row $a_{i}$ consequently has norm 
\[
\left\Vert a_{i}\left(a^{\top}_{i}V_{k}B^{(t)}\right)\right\Vert _{F}=\left\Vert a_{i}\right\Vert _{2}\left\Vert e^{\top}_{i}U_{k}\right\Vert _{2}
\]
which is controlled by the coherence and no longer contains $\sigma_{1},\ldots,\sigma_{k}$.
When the iterate is close to $V_{k}$ and $B^{(t)}$ is computed from
the noisy $k\times k$ matrix, this identity does not hold exactly.
Our analysis instead shows that the part of $a^{\top}_{i}X^{(t)}B^{(t)}$
from $V_{k}$ remains bounded by a constant multiple of a coherence-related
quantity. The remaining part of $a^{\top}_{i}X^{(t)}B^{(t)}$ from
$V_{>k}$ is controlled through a separate row-wise analysis. This
allows the final error bound to depend on $\sigma_{k}$, rather than
on all of the leading singular values.

\textbf{Clipping.} The preceding bounds hold with high probability
when the iterates remain close to $V_{k}$. Differential privacy,
however, requires a sensitivity bound that holds for every input.
For a row vector $y\in\mathbb{R}^{1\times k}$ and a radius $R>0$,
define $\mathrm{Clip}_{R}\left(y\right)=y\min\left\{ 1,\frac{R}{\|y\|_{2}}\right\} ,$with
$\mathrm{Clip}_{R}\left(0\right)=0$. Our update has the form: 
\begin{align*}
Z^{(t+1)} & =\sum^{n}_{i=1}a_{i}\mathrm{Clip}_{R}\left(a^{\top}_{i}X^{(t)}B^{(t)}\right)+\text{Gaussian noise}.
\end{align*}
Since $\left\Vert a_{i}\right\Vert _{2}\leq1$, $\left\Vert a_{i}\mathrm{Clip}_{R}\left(a^{\top}_{i}X^{(t)}B^{(t)}\right)\right\Vert _{F}\leq R$.
Under replacement adjacency, the sensitivity of the sum is therefore
at most $2R$, and the Gaussian noise can be adjust to this sensitivity.

Clipping supplies a worst-case sensitivity bound but is never active
with high probability. In particular, under the assumptions of our
main result, our analysis gives a radius $R=\widetilde{O}\left(\sqrt{k}\left\Vert U\right\Vert _{\infty}\right)$
for which $\left\Vert a^{\top}_{i}X^{(t)}B^{(t)}\right\Vert _{2}\leq R$
for every row and every iteration, with high probability.

\textbf{Private initialization.} The preconditioner and row-wise bounds
described above apply only when the current iterate has sufficient
overlap with the target subspace. We obtain such an initialization
using the Gaussian mechanism of \citet{DBLP:conf/stoc/DworkTT014}.
In particular, we add symmetric Gaussian noise to $A^{\top}A$ and
let $X^{(0)}$ contain the top $k$ eigenvectors of the resulting
matrix. Under an initial condition, this initialization satisfies
$\left\Vert \left(I-X^{(0)}X^{(0)\top}\right)V_{k}\right\Vert _{2}\leq1/4$.

\textbf{Private selection of $R$ and $T$.} The appropriate choices
of $R$ and $T$ depend on private properties of the input: $R$ depends
on the coherence, while $T$ depends on $\sigma^{2}_{k+1}/\sigma^{2}_{k}$.
We consider public powers-of-two grids for $R$ and $T$ and use the
private selection method of \citet{PapernotSteinke2022} to choose
among the resulting outputs. The construction and analysis of this
procedure are given in the algorithm section. Thus, our algorithm
does not require the coherence or $\sigma^{2}_{k+1}/\sigma^{2}_{k}$
as input.

\textbf{Analysis technique.} The quality of the solution $X^{(t)}$
is evaluated via the ratio between its projection on the top $k$
subspace $C=V^{\top}_{k}X^{(t)}$ and its projection on the bottom
$d-k$ subspace $D=V^{\top}_{>k}X^{(t)}$. Let $E^{(t)}=D^{(t)}\left(C^{(t)}\right)^{-1}$.
We track $\left\Vert E^{(t)}\right\Vert _{2}$ over the iterations,
ideally showing that it decreases exponentially. If no clipping happens,
we prove the recurrence $E^{(t+1)}=\left(\Sigma^{2}_{>k}E^{(t)}\Sigma^{-2}_{k}+noise\right)\left(I+noise\right)^{-1}$
where the noise terms come from the Gaussian noise in the update for
$Z^{(t+1)}$. Without noise, we can see that $\left\Vert \Sigma^{2}_{>k}E^{(t)}\Sigma^{-2}_{k}\right\Vert _{2}\le\frac{\sigma^{2}_{k+1}}{\sigma^{2}_{k}}\left\Vert E^{(t)}\right\Vert _{2}=(1-\kappa_{k})\left\Vert E^{(t)}\right\Vert _{2}$,
indeed decreasing exponentially with iterations. Naively bounding
the noises per iteration would create an error term linear in the
number of iterations, which is too large. In different iterations,
we add independent noises so intuitively, these noises are nearly
orthogonal to each others and have a lot of cancellation. To exploit
the cancellation, we need to perform a global analysis involving all
iterations together. However, the dependence on the noise in each
iteration is complex with a term in the numerator and a term in the
denominator. Using $(I+F)^{-1}=I-F(I+F)^{-1}$ for any noise matrix
$F$, the noise contribution is split into two parts: one part being
added up linearly across all iterations, which is amenable to a martingale
analysis, and one complicated part, which is bounded separately per
iteration. The second part is small so even though the bound might
be suboptimal, it is sufficiently small to not matter. This analysis
is possible only if we never clip. First we show that clipping does
not happen for the initial solution using Gaussian concentration.
Then we use induction to show that clipping does not happen for the
rest of the iterations. 

\section{Preliminaries}

Let $A\in\R^{n\times d}$ be the input data matrix whose rows $a_{i}$
are datapoints in $\R^{d}$. We assume that $n\geq d$ and each row
$a_{i}$ of $A$ has length $\left\Vert a\right\Vert _{2}\leq1$.
As noted by \citet{DBLP:conf/stoc/DworkTT014}, the latter assumption
is necessary for achieving non-trivial utility guarantees and it is
a natural restriction for PCA, since otherwise the solution could
be dominated by a single datapoint with very large norm and it could
fail to capture the underlying structure shared by most of the datapoints.

\textbf{Notation}. Let $A=U\Sigma V^{\top}$ be the thin singular
value decomposition (SVD) of $A$, where $U\in\R^{n\times d}$ is
the matrix whose columns are the left-singular vectors $u_{1},\dots,u_{d}$,
$V\in\R^{d\times d}$ is the matrix whose columns are the right-singular
vectors $v_{1},\dots,v_{d}$, and $\Sigma=\mathrm{diag}\left(\sigma_{1},\dots,\sigma_{d}\right)\in\R^{d\times d}$
where $\sigma_{1}\geq\sigma_{2}\geq\dots\geq\sigma_{d}$ are the singular
values.  Let $U_{k}\in\R^{n\times k}$ and $V_{k}\in\R^{d\times k}$
be the matrices whose columns are the top-$k$ left-singular and right-singular
vectors, respectively. Let $\Sigma_{k}=\diag\left(\sigma_{1},\dots,\sigma_{k}\right)$.
We assume that there is a positive eigengap $\Delta_{k}:=\sigma^{2}_{k}-\sigma^{2}_{k+1}>0$. 

We let $\left\Vert U\right\Vert _{\infty}:=\max_{i,j}\left|U_{ij}\right|$,
and we refer to $\left\Vert U\right\Vert _{\infty}$ as the \textbf{coherence}
of the input matrix $A$. Note that $U_{ij}=\left\langle a_{i},v_{j}\right\rangle /\sigma_{j}$
and thus $\left\Vert U\right\Vert _{\infty}$ measures the maximum
correlation between the datapoints (the rows of $A$) and the right-singular
vectors. The coherence value is not known to the algorithm. 

We use the $\widetilde{O}/\widetilde{\Omega}$ notation to suppress
factors that are poly-logarithmic in $n$, $d$, $1/\epsilon$, $1/\delta$,
and $1/\beta$, where $n$ and $d$ are the input dimensions, $\epsilon$
and $\delta$ are the privacy parameters, and $\beta$ is the failure
probability.

\textbf{Clipping. }For a row vector $y\in\R^{1\times k}$ and a radius
$R>0$, we define the clipping operation:
\[
\Clip_{R}(y):=y\cdot\min\left\{ 1,\frac{R}{\left\Vert y\right\Vert _{2}}\right\} 
\]
with the convention that $\Clip_{R}(0)=0$.

\textbf{Differential privacy. }We consider the privacy model where
each row of the input matrix $A$ is the private data of an user.
We consider the row-replacement adjacency model where two input matrices
$A$ and $A'$ are neighboring inputs if $A$ and $A'$ differ in
exactly one row. A randomized algorithm ${\cal M}$ is $(\epsilon,\delta)$-differentially
private ($(\epsilon,\delta)$-DP) if, for every pair of neighboring
inputs $A,A'$ and every measurable subset of possible outcomes ${\cal O}$,
we have
\begin{align*}
\Pr\left[{\cal M}(A)\in{\cal O}\right] & \le e^{\epsilon}\Pr\left[{\cal M}(A')\in{\cal O}\right]+\delta.
\end{align*}
Our algorithm and analysis use the Gaussian mechanism and advanced
composition, and we give the definitions and theorems in Section \ref{sec:utility-theorems}.

\section{Our Algorithm}

\begin{algorithm}
\caption{Algorithm for constructing an initial solution}

\label{alg:Algorithm-initialize}

\begin{algorithmic}[1]

\STATE \textbf{Input}: Matrix $A\in\R^{n\times d}$, privacy parameters
$\epsilon_{\mathrm{init}},\delta_{\mathrm{init}}$

\STATE Let $\Delta_{\mathrm{init}}=2\sqrt{2\log\left(2/\delta_{\mathrm{init}}\right)}/\epsilon_{\mathrm{init}}$.

\STATE {\small{} Let $\widehat{M}=A^{\top}A+\Delta_{\mathrm{init}}\cdot\sym\left(G^{(\mathrm{init})}\right)$
where $G^{(\mathrm{init})}\sim\N(0,1)^{d\times d}$ and $\sym\left(G\right)=\frac{1}{2}\left(G+G^{\top}\right)$}{\small\par}

\STATE {\small Let $X^{(0)}$ be a $d\times k$ matrix whose columns
are the top $k$ eigenvectors of $\widehat{M}$}{\small\par}

\STATE Return $X^{(0)}$

\end{algorithmic}
\end{algorithm}

\begin{algorithm}
\caption{Private Preconditioned Power Method}

\label{alg:Algorithm-power-method}

\begin{algorithmic}[1]

\STATE \textbf{Input}: Matrix $A\in\R^{n\times d}$ with rows $a_{i}$,
initial solution $X^{(0)},$number of iterations $T$, $k\leq d$,
per-iteration privacy parameters $\epsilon_{\mathrm{iter}},\delta_{\mathrm{iter}}$,
failure probability $\beta$, clipping radius $R$

\STATE {\small Let $\Delta=2\sqrt{2\log\left(2/\delta_{\mathrm{iter}}\right)}/\epsilon_{\mathrm{iter}}$,
$\gamma=4\Delta\left(2\sqrt{k}+\sqrt{2\log(8T/\beta)}\right)$}{\small\par}

\STATE \textbf{for} $t=0\dots T-1$

\STATE $\qquad${\small Let $H^{(t)}=X^{(t)\top}A^{\top}AX^{(t)}+\Delta\cdot\sym\left(G^{(t)}\right)$
where $G^{(t)}\sim\N\left(0,1\right)^{k\times k}$ and $\sym\left(G\right)=\frac{1}{2}\left(G+G^{\top}\right)$}{\small\par}

\STATE $\qquad${\small Let $H^{(t)}=Q^{(t)}\Lambda^{(t)}Q^{(t)\top}$,
$\widetilde{H}^{(t)}=Q^{(t)}\diag\left(\max\left\{ \Lambda^{(t)}_{11},\gamma\right\} ,\ldots,\max\left\{ \Lambda^{(t)}_{kk},\gamma\right\} \right)Q^{(t)\top}$,
$B^{(t)}=\left(\widetilde{H}^{(t)}\right)^{-1/2}$}{\small\par}

\STATE $\qquad${\small Let $Z^{(t+1)}=\sum^{n}_{i=1}a_{i}\Clip_{R}\left(a^{\top}_{i}X^{(t)}B^{(t)}\right)+\Delta R\cdot\widetilde{G}^{(t)}$
where $\widetilde{G}^{(t)}\sim\N(0,1)^{d\times k}$}{\small\par}

\STATE $\qquad$Let $X^{(t+1)}R^{(t+1)}=Z^{(t+1)}$ be the QR-decomposition
of $Z^{(t+1)}$

\STATE \textbf{end for}

\STATE \textbf{return} $X^{(T)}$

\end{algorithmic}
\end{algorithm}

\begin{algorithm}
\caption{Main Algorithm}

\label{alg:Algorithm-main}

\begin{algorithmic}[1]

\STATE \textbf{Input}: Matrix $A\in\R^{n\times d}$, privacy parameters
$\epsilon,\delta$, failure probability $\beta$

\STATE Set $\epsilon_{\mathrm{init}}=\frac{\epsilon}{2}$ and $\delta_{\mathrm{init}}=\frac{\delta}{2}$

\STATE Let $X^{(0)}$ be the solution returned by Algorithm \ref{alg:Algorithm-initialize}
when run with $\epsilon_{\mathrm{init}}=\frac{\epsilon}{2}$ and $\delta_{\mathrm{init}}=\frac{\delta}{2}$

\STATE Set $\epsilon_{\mathrm{iter}}$ and $\delta_{\mathrm{iter}}$
so that Algorithm \ref{alg:Algorithm-power-method} is $\left(\frac{\epsilon}{2},\frac{\delta}{2}\right)$-DP

\STATE Let $X^{(T)}$ be the solution returned by Algorithm \ref{alg:Algorithm-power-method}
when run with $\epsilon_{\mathrm{iter}}$ and $\delta_{\mathrm{iter}}$,
and failure probability $\frac{\beta}{2}$

\STATE Return $X^{(T)}$

\end{algorithmic}
\end{algorithm}

In this section, we give an overview of our algorithm and analysis.
Our main algorithm is shown in Algorithm \ref{alg:Algorithm-main}.
The algorithm uses the algorithm of \citet{DBLP:conf/stoc/DworkTT014}
to construct an initial solution (Algorithm \ref{alg:Algorithm-initialize}).
Our main optimization algorithm is shown in Algorithm \ref{alg:Algorithm-power-method}.
This algorithm takes as input the initial solution constructed by
Algorithm \ref{alg:Algorithm-initialize} together with a clipping
radius $R$ and an iteration count $T$, and it iteratively constructs
a solution using a private variant of a preconditioned power method
as discussed in the introduction. In Section \ref{sec:private-radius-selection}
in the Appendix, we show how to privately select the hyperparameters
$R$ and $T$ using the private hyperparameter-selection method of
Papernot and Steinke \citet{PapernotSteinke2022}.

We now give an overview of our analysis. Due to space constraints,
we defer the complete analysis to Section \ref{sec:analysis} in the
Appendix.

\textbf{Privacy analysis. }Our algorithm uses the Gaussian mechanism
to privatize the updates, and we show that the algorithm is $\left(\epsilon,\delta\right)$-DP
via standard applications of the Gaussian mechanism and the advanced
composition theorem.

\textbf{Utility analysis. }We measure the utility of our algorithm
via the error $\left\Vert \left(I-XX^{\top}\right)V_{k}\right\Vert _{2}$,
where $X$ is a solution with orthonormal columns and $V_{k}$ is
the matrix whose columns are the top-$k$ right-singular vectors.
We recall that this error is the sine of the maximum principal angle
between the column spaces of $X$ and $V_{k}$. We let $s_{t}:=\left\Vert \left(I-X^{(t)}X^{(t)\top}\right)V_{k}\right\Vert _{2}$
be the error of solution $X^{(t)}$ constructed in iteration $t-1$
of our algorithm. 

We analyze the error $s_{t}$ via the following quantity. Let $C^{(t)}:=V^{\top}_{k}X^{(t)}$,
$D^{(t)}:=V^{\top}_{>k}X^{(t)}$, and $E^{(t)}:=D^{(t)}\left(C^{(t)}\right)^{-1}$
provided that $C^{(t)}$ is invertible. The initial solution $X^{(0)}$
constructed by Algorithm \ref{alg:Algorithm-initialize} satisfies
$s_{0}\leq1/4$, which guarantees that $C^{(0)}$ is invertible. Given
this property, we show by induction that the algorithm maintains the
invariant that $s_{t}<1$, which ensures that $C^{(t)}$ is invertible
(Lemma \ref{lem:s-to-C-invertible}). We analyze the error $s_{t}$
via the spectral norm $m_{t}:=\left\Vert E^{(t)}\right\Vert _{2}$
of $E^{(t)}$, which upper bounds the error (i.e., we have $s_{t}\leq m_{t}$). 

In an ideal update with no clipping or privacy noise, we have $E^{(t+1)}=\Sigma^{2}_{>k}E^{(t)}\Sigma^{-2}_{k}$
and thus
\[
\left\Vert E^{(t+1)}\right\Vert _{2}\leq\left\Vert \Sigma^{2}_{>k}\right\Vert _{2}\left\Vert E^{(t)}\right\Vert _{2}\left\Vert \Sigma^{-2}_{k}\right\Vert _{2}=\frac{\sigma^{2}_{k+1}}{\sigma^{2}_{k}}\left\Vert E^{(t)}\right\Vert _{2}=\left(1-\kappa_{k}\right)\left\Vert E^{(t)}\right\Vert _{2},
\]
where $\kappa_{k}=\frac{\sigma^{2}_{k}-\sigma^{2}_{k+1}}{\sigma^{2}_{k}}$.
Thus, in this ideal case, the error decreases exponentially with the
eigengap. For our private update, we show that, if the clipping is
inactive in iteration $t$ (i.e., $\left\Vert a^{\top}_{i}X^{(t)}B^{(t)}\right\Vert _{2}\leq R$
for all $i\in[n]$), then the following modified recurrence holds:
\begin{equation}
E^{(t+1)}=\left(\Sigma^{2}_{>k}E^{(t)}\Sigma^{-2}_{k}+L^{(t)}\right)\left(I+F^{(t)}\right)^{-1},\label{eq:main-recurrence-body}
\end{equation}
where $L^{(t)}$ and $F^{(t)}$ are matrices arising from the Gaussian
noise added, defined as follows: $F^{(t)}:=\Delta RV^{\top}_{k}\widetilde{G}^{(t)}\left(K^{(t)}\right)^{-1}$
and $L^{(t)}:=\Delta RV^{\top}_{>k}\widetilde{G}^{(t)}\left(K^{(t)}\right)^{-1}$,
where $K^{(t)}:=\Sigma^{2}_{k}C^{(t)}B^{(t)}$. The main progress
term $\Sigma^{2}_{>k}E^{(t)}\Sigma^{-2}_{k}$ is the same as in the
ideal update, and its spectral norm decreases exponentially. Thus
if we can show that the clipping remains inactive throughout the algorithm
and the contribution of the error terms is sufficiently small, we
can show that we approximately preserve the exponential decrease in
the error.

We show that the error terms can be bounded with high probability
using Gaussian concentration and the properties of the preconditioner
$B^{(t)}$ used by the main update. We have
\[
\left\Vert F^{(t)}\right\Vert _{2}\leq\Delta R\left\Vert V^{\top}_{k}\widetilde{G}^{(t)}\right\Vert _{2}\left\Vert \left(K^{(t)}\right)^{-1}\right\Vert _{2}\quad\left\Vert L^{(t)}\right\Vert _{2}\leq\Delta R\left\Vert V^{\top}_{>k}\widetilde{G}^{(t)}\right\Vert _{2}\left\Vert \left(K^{(t)}\right)^{-1}\right\Vert _{2}.
\]
 Since the Gaussian distribution is rotationally invariant, $V^{\top}_{k}\widetilde{G}^{(t)}$
is a $k\times k$ Gaussian matrix and $V^{\top}_{>k}\widetilde{G}^{(t)}$
is a $(d-k)\times k$ Gaussian matrix. Using Gaussian concentration
results (Lemma \ref{lem:gaussian-norm}), we can show that we have
$\left\Vert V^{\top}_{k}\widetilde{G}^{(t)}\right\Vert _{2}\leq\widetilde{O}\left(\sqrt{k}\right)$
and $\left\Vert V^{\top}_{>k}\widetilde{G}^{(t)}\right\Vert _{2}\leq\widetilde{O}\left(\sqrt{d-k}+\sqrt{k}\right)$
with high probability. We can also show that $\left\Vert \left(K^{(t)}\right)^{-1}\right\Vert _{2}\leq O\left(\frac{1}{\sigma_{k}}\right)$
(Lemma \ref{lem:K-inverse}) as follows. We have $K^{(t)}:=\Sigma^{2}_{k}C^{(t)}B^{(t)}$.
We can show that the preconditioner $B^{(t)}$ roughly whitens $\Sigma_{k}V^{\top}_{k}X^{(t)}=\Sigma_{k}C^{(t)}$,
i.e., the singular values of $\Sigma_{k}C^{(t)}B^{(t)}$ are bounded
from above and below by absolute constants (Lemma \ref{lem:whitening}).
This allows us to bound $\left\Vert \left(K^{(t)}\right)^{-1}\right\Vert _{2}\leq\left\Vert \left(\Sigma_{k}C^{(t)}B^{(t)}\right)^{-1}\right\Vert _{2}\left\Vert \Sigma^{-1}_{k}\right\Vert _{2}=O(1/\sigma_{k})$
since $\left\Vert \left(\Sigma_{k}C^{(t)}B^{(t)}\right)^{-1}\right\Vert _{2}=O(1)$.

If the clipping is inactive, we can use the recurrence above together
with the upper bounds on the noise terms described above to analyze
the error progress. In order to obtain a stronger bound, we use a
martingale argument to show that the contribution of some of the error
terms is smaller. To simplify notation, for a matrix $M$, let us
define $\calT\left(M\right):=\Sigma^{2}_{>k}M\Sigma^{-2}_{k}$. Thus
the main recurrence (\ref{eq:main-recurrence-body}) becomes:
\[
E^{(t+1)}=\left(\calT\left(E^{(t)}\right)+L^{(t)}\right)\left(I+F^{(t)}\right)^{-1}.
\]
Using the identity $\left(I+F\right)^{-1}=I-F\left(I+F\right)^{-1}$,
we can rewrite the recurrence as follows:
\[
E^{(t+1)}=\calT\left(E^{(t)}\right)+L^{(t)}\underbrace{-\left(\calT\left(E^{(t)}\right)+L^{(t)}\right)F^{(t)}\left(I+F^{(t)}\right)^{-1}}_{:=R^{(t)}}=\calT\left(E^{(t)}\right)+L^{(t)}+R^{(t)}.
\]
The reason this is beneficial is that it separates the noise terms
depending on $\widetilde{G}^{(t)}$ into the term $L^{(t)}$ that
is linear in $\widetilde{G}^{(t)}$ and a more complex term $R^{(t)}$.
The simple linear structure of the terms $L^{(t)}$ allows us to bound
their contribution via a martingale argument. By unrolling the recurrence,
we obtain
\[
E^{(t)}=\mathcal{T}^{t}\left(E^{(0)}\right)+\sum^{t-1}_{j=0}\mathcal{T}^{t-1-j}\left(L^{(j)}\right)+\sum^{t-1}_{j=0}\mathcal{T}^{t-1-j}\left(R^{(j)}\right).
\]
We upper bound the spectral norm of the middle term $\sum^{t-1}_{j=0}\mathcal{T}^{t-1-j}\left(L^{(j)}\right)$
via a martingale concentration argument (Lemma (\ref{lem:martingale})).
For the last term, we upper bound its spectral norm term by term,
and we bound $\left\Vert \calT^{r}\left(R^{(j)}\right)\right\Vert _{2}=\left\Vert \Sigma^{2r}_{>k}R^{(j)}\Sigma^{-2r}_{k}\right\Vert _{2}$
using a similar approach as discussed above.

The analysis discussed above relied on the clipping being inactive
in every iteration. A significant part of our analysis is devoted
to showing that this event holds with high probability. Note that
the clipping is inactive in iteration $t$ if $\left\Vert a^{\top}_{i}X^{(t)}B^{(t)}\right\Vert _{2}\leq R$
for all $i\in[n]$, and thus we need to upper bound the lengths $\left\Vert a^{\top}_{i}X^{(t)}B^{(t)}\right\Vert _{2}$.
We can write $AX^{(t)}B^{(t)}=\left(U_{k}+U_{>k}S^{(t)}\right)P^{(t)}$
where $S^{(t)}:=\Sigma_{>k}E^{(t)}\Sigma^{-1}_{k}$ and $P^{(t)}:=\Sigma_{k}C^{(t)}B^{(t)}$.
Thus, for each row $i$, we have $a^{\top}_{i}X^{(t)}B^{(t)}=e^{\top}_{i}\left(U_{k}+U_{>k}S^{(t)}\right)P^{(t)}$
and thus 
\[
\left\Vert a^{\top}_{i}X^{(t)}B^{(t)}\right\Vert _{2}\leq\left(\left\Vert e^{\top}_{i}U_{k}\right\Vert _{2}+\left\Vert e^{\top}_{i}U_{>k}S^{(t)}\right\Vert _{2}\right)\left\Vert P^{(t)}\right\Vert _{2}.
\]
 We have $\left\Vert e^{\top}_{i}U_{k}\right\Vert _{2}\le\sqrt{k}\left\Vert U\right\Vert _{\infty}$
and $\left\Vert P^{(t)}\right\Vert _{2}\leq O(1)$ with high probability
as noted earlier, and thus it only remains to upper bound $\left\Vert e^{\top}_{i}U_{>k}S^{(t)}\right\Vert _{2}$.
To this end, we introduce the following potential for every $0\leq r,t\leq T$:
\[
\Phi_{r,t}:=\max_{i\in\left[n\right]}\left\Vert e^{\top}_{i}U_{>k}\Sigma^{2r}_{>k}S^{(t)}\right\Vert _{2}.
\]
The potential values with $r=0$ are precisely the terms that we need
to upper bound, and the values with larger $r$ allows us to analyze
these terms. We first use a concentration argument to show that the
initialization allows us to upper bound $\Phi_{r,0}$ for all $r$
with high probability (Lemma \ref{lem:initial-potential}). We then
use an inductive argument to show that we can upper bound the potentials
$\Phi_{0,t}$ using the initial potentials $\Phi_{r,0}$. Here we
assume inductively that the clipping was inactive in all prior iterations
and we use the recurrence for $E^{(t+1)}$ discussed above, and show
that the clipping is inactive in the current iteration as well (Lemma
\ref{lem:clipping-inactive}). Overall, we obtain that, a radius $R\geq\widetilde{O}\left(\sqrt{k}\left\Vert U\right\Vert _{\infty}\right)$
suffices to ensure that clipping is inactive in all iterations with
high probability.

\section{Related work}

\textbf{Privacy models on coherence-sensitive PCA.} We study row-level
adjacency, under which two matrices differ in one row of $\ell_{2}$
length at most $1$. The guarantees of \citet{DBLP:conf/stoc/DworkTT014}
and \citet{NguyenEN26} apply under this model. By contrast, most
earlier coherence-sensitive guarantees were established under different
notions of adjacency. Earlier power-method results consider entry-level
adjacency on the input matrix \citep{DBLP:conf/stoc/HardtR13,DBLP:conf/nips/HardtP14,DBLP:conf/colt/BalcanDWY16}.
\citet{NicolasSMMC25} and \citet{dOrsiN26} consider more general
adjacency models; when an entire row changes, these models restrict
the change to have bounded $\ell_{1}$ norm.

\textbf{Worst-case bounds.} \citet{DBLP:conf/stoc/DworkTT014} give
a rank-$k$ guarantee without any assumption on coherence. It provides
a benchmark for coherence-independent $k$-PCA. Their guarantee gives
the principal-angle error bound $\widetilde{O}\left(\frac{\sqrt{d}}{\epsilon\left(\sigma^{2}_{k}-\sigma^{2}_{k+1}\right)}\right)$.
\citet{TranVV26} study the spectral-norm error $\left\Vert \widetilde{M}_{k}-M_{k}\right\Vert _{2}$,
where $M_{k}$ and $\widetilde{M}_{k}$ are the best rank-$k$ approximations
before and after the Gaussian perturbation. Under their gap condition,
converting this spectral-norm guarantee to principal-angle error recovers
the bound of \citet{DBLP:conf/stoc/DworkTT014} up to logarithmic
factors.

\textbf{Beyond-worst-case bounds.} \citet{DBLP:conf/stoc/HardtR13}
is the first paper to study coherence-sensitive guarantees for private
singular-vector computation. They give a private rank-$1$ algorithm
and use repeated deflation to obtain a rank-$k$ approximation. The
resulting guarantee deteriorates when $\sigma_{1}/\sigma_{k}$ is
large. Subsequent works analyze block power methods, which update
$k$ directions simultaneously and avoid the dependence on $\sigma_{1}$
introduced by repeated deflation \citep{DBLP:conf/nips/HardtP14,DBLP:conf/colt/BalcanDWY16,NicolasSMMC25}.
A different approach is taken by \citet{dOrsiN26}, who privately
estimate the spectral gap and coherence and use these estimates as
input to a Gaussian mechanism for estimating the top-$k$ subspace.

\textbf{Private block power method.} \citet{DBLP:conf/nips/HardtP14}
give a general analysis of the noisy block power method and apply
it to private rank-$k$ PCA. Their algorithm uses oversampling, maintaining
$p\geq k$ vectors to approximate the subspace spanned by the top
$k$ singular vectors. Its coherence-sensitive principal-angle guarantee
depends on the gap $\sigma_{k}-\sigma_{k+1}$. They conjecture that,
with oversampling, this dependence can be replaced by dependence on
the larger gap $\sigma_{k}-\sigma_{p+1}$. \citet{DBLP:conf/colt/BalcanDWY16}
partially address this conjecture by replacing the dependence on $\sigma_{k}-\sigma_{k+1}$
with dependence on $\sigma_{k}-\sigma_{q+1}$ for an intermediate
rank $q\in\left[k,p/2\right]$. Building on this framework, \citet{NicolasSMMC25}
sharpen the sensitivity bound from $\sqrt{p}\left\Vert X_{t}\right\Vert _{\max}$
to $\max_{i}\left\Vert \left(X_{t}\right)_{i,:}\right\Vert _{2}$.
By contrast, our algorithm does not use oversampling and maintains
exactly $k$ vectors throughout the iterations.

\textbf{Statistical / i.i.d private PCA.} This literature obtains
stronger statistical guarantees by assuming that the data are i.i.d
samples from a structured distribution. \citet{DBLP:conf/nips/LiuK0O22}
estimates the top singular vector with sub-Gaussian distributed data.
\citet{DunglerS25} generalizes this work for top $k$ singular vectors.

\section{AI use statement}

In this work, we used GPT 5.6 Sol, 6 Sol, and 6 Astra for developing
the algorithms and proofs; reviewing mathematical arguments; retrieving
and summarizing related literature; comparing and interpreting theoretical
guarantees; and drafting, restructuring, and editing portions of the
manuscript. The authors verified and rewrote all of the proofs, and
checked the cited literature against the original sources. We did
not use generative AI tools to generate or analyze data, conduct experiments,
or implement the algorithm; these tasks are not applicable to this
theoretical work. We take responsibility for the final content of
this work, including all text and mathematical claims produced with
the aid of generative AI. 

\newpage

\bibliographystyle{plainnat}
\bibliography{ref}

\appendix
\newpage

\section{Main Tools}

\label{sec:utility-theorems}

In this section, we state the main theorems from prior work that we
will use in our utility analysis shown in Section \ref{sec:alg-utility}.
Our algorithm and analysis use the following tools from differential
privacy.

\textbf{Gaussian mechanism.} Let $f$ be a function that takes as
input a matrix $A$ and outputs a vector. We say $f$ has $\ell_{2}$-sensitivity
$\Delta$ if on any neighboring inputs $A$ and $A'$, we have $\left\Vert f(A)-f(A')\right\Vert _{2}\le\Delta$.
The following result was shown by \citet{dwork2006our}.
\begin{thm}[Gaussian mechanism, \citet{dwork2006our}]
 Assume $0<\epsilon\le1$ and $0<\delta<1$. Given an input matrix
$A$ and a function $f$ with $\ell_{2}$-sensitivity $\Delta$, the
algorithm ${\cal M}$, which outputs $f(A)+G$ where $G$ is sampled
from the Gaussian distribution $\N(0,\sigma^{2}I)$ with $\sigma\ge\frac{\Delta}{\epsilon}\sqrt{2\log\frac{2}{\delta}}$,
is $(\epsilon,\delta)$-DP.
\end{thm}

\begin{thm}[Advanced composition \citet{KairouzOV17}]
 \label{thm:advanced-composition} For all $\epsilon,\delta$, the
adaptive composition of $k$ algorithms, each of which is $(\epsilon,\delta)$-DP,
is $(\epsilon',\delta')$-private where
\[
\epsilon'=\epsilon\sqrt{2\ln\left(\frac{1}{\delta}\right)k}+\epsilon\frac{e^{\epsilon}-1}{e^{\epsilon}+1}k\quad\text{and}\quad\delta'=\delta(k+1)
\]
\end{thm}

\textbf{Gaussian concentration results.} We also use the following
concentration results for Gaussian random variables.
\begin{thm}[{\citet[Theorem 5.6, Tsirelson–Ibragimov–Sudakov inequality]{BoucheronLM13}}]
\label{thm:gaussian-concentration} If $G\sim\N(0,1)^{d\times d}$
and $f:\mathbb{R}^{d\times d}\to\mathbb{R}$ is $L$-Lipschitz in
Euclidean norm, then, for every $x\geq0$, 
\[
\Pr\left[f\left(G\right)-\E\left[f\left(G\right)\right]\geq x\right]\leq\exp\left(-\frac{x^{2}}{2L^{2}}\right),\qquad\Pr\left[f\left(G\right)-\E\left[f\left(G\right)\right]\leq-x\right]\leq\exp\left(-\frac{x^{2}}{2L^{2}}\right).
\]
\end{thm}

\begin{lem}[{\citet[Theorem II.13]{DavidsonSzarek01}}]
\label{lem:gaussian-norm} Let $G\sim\N(0,1)^{q\times k}$. For every
$x\ge0$,
\[
\Pr\left[\left\Vert G\right\Vert _{2}>\sqrt{q}+\sqrt{k}+x\right]\le e^{-x^{2}/2}.
\]
\end{lem}

\begin{proof}
The first assertion follows from Davidson-Szarek bound \citet[Theorem II.13]{DavidsonSzarek01}.
We must rescale our matrix to match their normalization. Assume first
that $k\le q$, and define $\Gamma:=\frac{1}{\sqrt{q}}G$. Then $\Gamma\in\R^{q\times k}$
has independent $\N\left(0,1/q\right)$ entries. Apply the Davidson-Szarek
Theorem with parameters $n=q,m=k,\beta=\frac{k}{q},t=\frac{x}{\sqrt{q}},$and
using the identity $\left\Vert G\right\Vert _{2}=\sqrt{q}\,s_{1}(\Gamma)$,
we obtain: 
\begin{align*}
\Pr\left[\left\Vert G\right\Vert _{2}\ge\sqrt{q}+\sqrt{k}+x\right] & =\Pr\left[s_{1}(\Gamma)\ge1+\sqrt{\frac{k}{q}}+\frac{x}{\sqrt{q}}\right]\\
 & \le\exp\left(-\frac{q}{2}\left(\frac{x}{\sqrt{q}}\right)^{2}\right)=e^{-x^{2}/2}.
\end{align*}
If $k>q$, we apply the preceding argument to $G^{\top}$, since $\left\Vert G^{\top}\right\Vert _{2}=\left\Vert G\right\Vert _{2}$.

\end{proof}

\begin{lem}
\label{lem:gaussian-cross} Let $M\in\mathbb{R}^{q\times k}$ and
$R\in\mathbb{R}^{k\times k}$ be fixed matrices, and let $G\sim\mathcal{N}\left(0,1\right)^{q\times k}$.
For every $x\geq0$, $\Pr\left[\left\Vert M^{\top}GR\right\Vert _{2}>\left\Vert M\right\Vert _{2}\left\Vert R\right\Vert _{2}\left(2\sqrt{k}+x\right)\right]\leq e^{-x^{2}/2}$.
\end{lem}

\begin{proof}
Let $r=\mathrm{rank}\left(M\right)\leq k$, and write a compact singular
value decomposition $M=P\Lambda Q^{\top}$, where $P\in\mathbb{R}^{q\times r}$
has orthonormal columns. By rotational invariance, $P^{\top}G\sim\mathcal{N}\left(0,1\right)^{r\times k}$.
Hence Lemma~\ref{lem:gaussian-norm} gives 
\[
\Pr\left[\left\Vert P^{\top}G\right\Vert _{2}\leq\sqrt{r}+\sqrt{k}+x\leq2\sqrt{k}+x\right]\geq1-e^{-x^{2}/2}
\]
But $\left\Vert M^{\top}GR\right\Vert _{2}=\left\Vert Q\Lambda P^{\top}GR\right\Vert _{2}\leq\left\Vert M\right\Vert _{2}\left\Vert P^{\top}G\right\Vert _{2}\left\Vert R\right\Vert _{2}$,
which implies that:
\begin{align*}
\Pr\left[\left\Vert M^{\top}GR\right\Vert _{2}>\left\Vert M\right\Vert _{2}\left\Vert R\right\Vert _{2}\left(2\sqrt{k}+x\right)\right] & \leq\Pr\left[\left\Vert M\right\Vert _{2}\left\Vert P^{\top}G\right\Vert _{2}\left\Vert R\right\Vert _{2}>\left\Vert M\right\Vert _{2}\left\Vert R\right\Vert _{2}\left(2\sqrt{k}+x\right)\right]\\
 & =\Pr\left[\left\Vert P^{\top}G\right\Vert _{2}>2\sqrt{k}+x\right]\\
 & \leq e^{-x^{2}/2}
\end{align*}
\end{proof}

\textbf{Spectral analysis tools. }Our analysis uses Weyl's inequality
and the Davis-Kahan $\sin\Theta$ theorem. 
\begin{thm}[Weyl's inequality]
\label{thm:weyl} Let $M$ and $W$ be two symmetric matrices and
for a matrix $A$, let $\lambda_{i}\left(A\right)$ denote the $i$-th
largest eigenvalue of $A$. For every $i$, we have
\begin{align*}
\left|\lambda_{i}\left(M+W\right)-\lambda_{i}\left(M\right)\right| & \le\left\Vert W\right\Vert _{2}.
\end{align*}
\end{thm}

Let us also recall the definition of principal angles between two
subspaces. Let $V,\widehat{V}\in\R^{d\times k}$ have orthonormal
columns. Let $\sigma_{1},\dots,\sigma_{k}$ be the singular values
of $\widehat{V}^{\top}V$. The principal angles between the column
spaces of $V$ and $\widehat{V}$ are $\theta_{i}:=\arccos\left(\sigma_{i}\right)$
for $1\leq i\leq k$. We let $\Theta\left(V,\widehat{V}\right)=\diag\left(\theta_{1},\dots,\theta_{k}\right)$
and $\sin\Theta\left(V,\widehat{V}\right)=\diag\left(\sin\left(\theta_{1}\right),\dots,\sin\left(\theta_{k}\right)\right)$.
\begin{thm}[{Davis-Kahan $\sin\Theta$ theorem, \citet[Theorem V.3.6]{StewartSun1990},
\citet[Theorem 1]{YuWangSamworth2015}}]
\label{thm:davis-kahan} Let $M,\widehat{M}\in\R^{d\times d}$ be
symmetric matrices with eigenvalues $\lambda_{1}\geq\dots\geq\lambda_{d}$
and $\widehat{\lambda}_{1}\geq\dots\geq\widehat{\lambda}_{d}$, respectively.
Fix $1\leq r\leq s\leq d$, let $k=s-r+1$ and let $V=\left(v_{r},v_{r+1},\dots,v_{s}\right)\in\R^{d\times k}$
and $\widehat{V}=\left(\widehat{v}_{r},\widehat{v}_{r+1},\dots,\widehat{v}_{s}\right)\in\R^{d\times k}$
have orthonormal columns satisfying $Mv_{i}=\lambda_{i}v_{i}$ and
$\widehat{M}\widehat{v}_{i}=\widehat{\lambda}_{i}\widehat{v}_{i}$
for all $r\leq i\leq s$. Let 
\[
\delta:=\inf\left\{ \left|\lambda-\widehat{\lambda}\right|\colon\lambda\in\left[\lambda_{s},\lambda_{r}\right],\widehat{\lambda}\in(-\infty,\widehat{\lambda}_{s+1}]\cup[\widehat{\lambda}_{r-1},+\infty)\right\} 
\]
where $\widehat{\lambda}_{0}:=+\infty$ and $\widehat{\lambda}_{d+1}=-\infty$. 

If $\delta>0$, we have
\[
\left\Vert \left(I-\widehat{V}\widehat{V}^{\top}\right)V\right\Vert _{F}=\left\Vert \sin\Theta\left(V,\widehat{V}\right)\right\Vert _{F}\leq\frac{\left\Vert M-\widehat{M}\right\Vert _{F}}{\delta},
\]
and
\[
\left\Vert \left(I-\widehat{V}\widehat{V}^{\top}\right)V\right\Vert _{2}=\left\Vert \sin\Theta\left(V,\widehat{V}\right)\right\Vert _{2}\leq\frac{\left\Vert M-\widehat{M}\right\Vert _{2}}{\delta}.
\]
\end{thm}

Weyl's inequality (Theorem \ref{thm:weyl}) and the Davis-Kahan theorem
(Theorem \ref{thm:davis-kahan}) imply the following lemma. A similar
argument appears in \citet{DBLP:conf/stoc/DworkTT014}.
\begin{lem}
\label{lem:davis-kahan-error} Let $\widehat{M}=A^{\top}A+W$ be a
symmetric matrix, and let $\Delta_{k}=\lambda_{k}\left(A^{\top}A\right)-\lambda_{k+1}\left(A^{\top}A\right)>0$
denote the spectral gap. If $\left\Vert W\right\Vert _{2}<\Delta_{k}$
, the top-$k$ eigenspace $\widehat{V}_{k}$ of $\widehat{M}$ satisfies
\[
\left\Vert \left(I-\widehat{V}_{k}\widehat{V}^{\top}_{k}\right)V_{k}\right\Vert _{2}\le\frac{\left\Vert W\right\Vert _{2}}{\Delta_{k}-\left\Vert W\right\Vert _{2}}.
\]
\end{lem}

\begin{proof}
We apply the Davis--Kahan $\sin\Theta$ theorem (Theorem \ref{thm:davis-kahan})
with $M=A^{\top}A$ and $\widehat{M}=A^{\top}A+W$, and $r=1$ and
$s=k$. The spectral separation appearing in the theorem is $\delta:=\lambda_{k}\left(A^{\top}A\right)-\lambda_{k+1}\left(A^{\top}A+W\right)$.
We obtain
\[
\left\Vert \left(I-\widehat{V}_{k}\widehat{V}^{\top}_{k}\right)V_{k}\right\Vert _{2}=\left\Vert \sin\Theta\left(\widehat{V}_{k},V_{k}\right)\right\Vert _{2}\le\frac{\left\Vert W\right\Vert _{2}}{\delta}.
\]
We can upper bound $\lambda_{k+1}\left(A^{\top}A+W\right)$ using
Weyl's inequality (Theorem \ref{thm:weyl}):
\[
\lambda_{k+1}\left(A^{\top}A+W\right)\le\lambda_{k+1}\left(A^{\top}A\right)+\left\Vert W\right\Vert _{2},
\]
and thus
\begin{align*}
\delta & \geq\lambda_{k}\left(A^{\top}A\right)-\lambda_{k+1}\left(A^{\top}A\right)-\left\Vert W\right\Vert _{2}=\Delta_{k}-\left\Vert W\right\Vert _{2}.
\end{align*}
\end{proof}

\begin{lem}
\label{lem:neumann-series}Let $M$ be a matrix. If $\left\Vert M\right\Vert _{2}<1$,
$I+M$ is invertible and $\left\Vert \left(I+M\right)^{-1}\right\Vert _{2}\leq\frac{1}{1-\left\Vert M\right\Vert _{2}}$.
\end{lem}

\begin{proof}
We can show that $I+M$ is invertible as follows. Suppose we have
$\left(I+M\right)z=0$ for some vector $z$. Then $z=-Mz$ and thus
\[
\left\Vert z\right\Vert _{2}=\left\Vert Mz\right\Vert _{2}\leq\left\Vert M\right\Vert _{2}\left\Vert z\right\Vert _{2}
\]
Since $\left\Vert M\right\Vert _{2}<1$, the above inequality can
only hold if $\left\Vert z\right\Vert _{2}=0$. Thus $z=0$ and hence
$I+M$ is invertible.

Next, we show the inequality in the lemma statement. For any vector
$z$, we have
\[
\left\Vert \left(I+M\right)z\right\Vert _{2}\geq\left\Vert z\right\Vert _{2}-\left\Vert Mz\right\Vert _{2}\geq\left(1-\left\Vert M\right\Vert _{2}\right)\left\Vert z\right\Vert _{2}
\]
Thus
\[
\sigma_{\min}\left(I+M\right)\geq1-\left\Vert M\right\Vert _{2}
\]
Combining with $\left\Vert M\right\Vert _{2}<1$ gives
\[
\left\Vert \left(I+M\right)^{-1}\right\Vert _{2}=\frac{1}{\sigma_{\min}\left(1+M\right)}\leq\frac{1}{1-\left\Vert M\right\Vert _{2}},
\]
as needed.
\end{proof}

\section{Algorithm Analysis}

\label{sec:analysis}

\subsection{Privacy Analysis}

\label{sec:alg-privacy}
\begin{lem}
\label{lem:private-initialization} Assume every row $a_{i}$ of input
matrix $A$ satisfies $\left\Vert a_{i}\right\Vert _{2}\leq1$, and
let $\epsilon_{\mathrm{init}},\delta_{\mathrm{init}}\in\left(0,1\right)$.
Algorithm \ref{alg:Algorithm-initialize} is $\left(\epsilon_{\mathrm{init}},\delta_{\mathrm{init}}\right)$-differentially
private. 
\end{lem}

\begin{proof}
Under replacement adjacency, two neighboring datasets $A$ and $A'$
differ by exactly one row, meaning a row $a$ is replaced by $a'$.
Therefore:
\begin{align*}
A^{\top}A-A'^{\top}A' & =aa^{\top}-a'a'^{\top}
\end{align*}
and using triangle inequality, we have:
\[
\left\Vert A^{\top}A-A'^{\top}A'\right\Vert _{F}\le\left\Vert aa^{\top}\right\Vert _{F}+\left\Vert a'a'^{\top}\right\Vert _{F}=\left\Vert a\right\Vert ^{2}_{2}+\left\Vert a'\right\Vert ^{2}_{2}\leq2
\]
Thus, $\Delta\left\Vert A^{\top}A\right\Vert _{F}\leq2$.

In the orthonormal basis $\left\{ E_{ii}:1\le i\le d\right\} \cup\left\{ \frac{E_{ij}+E_{ji}}{\sqrt{2}}:1\le i<j\le d\right\} $
of the symmetric-matrix space, the coordinates of the matrix $\sym\left(G^{(\mathrm{init})}\right)$
are independent standard Gaussian random variables. Hence $\widehat{M}=A^{\top}A+\Delta_{\mathrm{init}}\sym\left(G^{(\mathrm{init})}\right)$
is the Gaussian mechanism with sensitivity at most $2$, and is therefore
$\left(\epsilon_{\mathrm{init}},\delta_{\mathrm{init}}\right)$-differentially
private. The subsequent step of computing $X^{(0)}$ from $\widehat{M}$
is a deterministic post-processing step.
\end{proof}

\begin{lem}
\label{lem:clipped-iterations-privacy} Fix an initial matrix $X^{(0)}$
with orthonormal columns, a number of iterations $T\geq1$, and a
clipping radius $R$. Let $\epsilon_{\mathrm{pow}},\delta_{\mathrm{pow}}\in\left(0,1\right)$
and set

\[
\begin{aligned}\epsilon_{\mathrm{iter}} & =\frac{\epsilon_{\mathrm{pow}}}{4\sqrt{2T\log\left(\left(T+1\right)/\delta_{\mathrm{pow}}\right)}},\\
\delta_{\mathrm{iter}} & =\frac{\delta_{\mathrm{pow}}}{2\left(T+1\right)}.
\end{aligned}
\]

Then Algorithm \ref{alg:Algorithm-power-method} is $\left(\epsilon_{\mathrm{pow}},\delta_{\mathrm{pow}}\right)$-differentially
private. 
\end{lem}

\begin{proof}
Since the algorithm is adaptive, we condition on the transcript preceding
each release. Previously released quantities may therefore be treated
as fixed when analyzing the conditional privacy guarantee.

First we analyze the release of $H^{(t)}$. For a single row $a$,
its contribution to $X^{(t)\top}A^{\top}AX^{(t)}$ is:
\begin{align*}
X^{(t)\top}aa^{\top}X^{(t)} & =\left(X^{(t)\top}a\right)\left(X^{(t)\top}a\right)^{\top}
\end{align*}
Because $X^{(t)}$ is a matrix with orthonormal columns, $\left\Vert X^{(t)\top}a\right\Vert _{2}\leq\left\Vert a\right\Vert _{2}$,
and thus:
\[
\left\Vert X^{(t)\top}aa^{\top}X^{(t)}\right\Vert _{F}=\left\Vert \left(X^{(t)\top}a\right)\left(X^{(t)\top}a\right)^{\top}\right\Vert _{F}=\left\Vert X^{(t)\top}a\right\Vert ^{2}_{2}\le\left\Vert a\right\Vert ^{2}_{2}\le1
\]
When comparing neighboring datasets where row $a$ is replaced by
$a'$, the total change is:
\[
\left\Vert X^{(t)\top}aa^{\top}X^{(t)}-X^{(t)\top}a'a'^{\top}X^{(t)}\right\Vert _{F}\le\left\Vert X^{(t)\top}aa^{\top}X^{(t)}\right\Vert _{F}+\left\Vert X^{(t)\top}a'a'^{\top}X^{(t)}\right\Vert _{F}\le2
\]
Therefore, $\Delta\left\Vert X^{(t)\top}A^{\top}AX^{(t)}\right\Vert _{F}\leq2$.
Since $G^{(t)}\sim\mathcal{N}\left(0,1\right)^{k\times k}$, adding
$\Delta\sym\left(G^{(t)}\right)$ applies the symmetric Gaussian mechanism.
Because $\Delta=2\sqrt{2\log\left(2/\delta_{\mathrm{iter}}\right)}/\epsilon_{\mathrm{iter}}$,
the release $H^{(t)}=X^{(t)\top}A^{\top}AX^{(t)}+\Delta\sym\left(G^{(t)}\right)$
is $\left(\epsilon_{\mathrm{iter}},\delta_{\mathrm{iter}}\right)$-differentially
private. Construction of $B^{(t)}$ are post-processing steps.

Second, we analyze the release of $Z^{(t+1)}$. The contribution of
one row $a$ to $Z^{(t+1)}$ is $f_{t}(a):=a\Clip_{R}\left(a^{\top}X^{(t)}B^{(t)}\right)$.
Since the norm of an outer product is the product of the two vector
norms, 
\[
\left\Vert f_{t}\left(a_{i}\right)\right\Vert {}_{F}=\left\Vert a_{i}\right\Vert {}_{2}\left\Vert \mathrm{Clip}_{R}\left(a^{\top}_{i}X^{(t)}B^{(t)}\right)\right\Vert {}_{2}\leq R
\]
Thus $\Delta\left(\sum^{n}_{i=1}a_{i}\Clip_{R}\left(a^{\top}_{i}X^{(t)}B^{(t)}\right)\right)\leq2R$.
Consequently, $Z^{(t+1)}=\sum^{n}_{i=1}f_{t}\left(a_{i}\right)+\Delta R\widetilde{G}^{(t)}$
is another $\left(\epsilon_{\mathrm{iter}},\delta_{\mathrm{iter}}\right)$-differentially
private Gaussian mechanism. The QR factorization is post-processing.

By sequential composition, a single iteration is $\left(2\epsilon_{\mathrm{iter}},2\delta_{\mathrm{iter}}\right)$-differentially
private. Advanced composition over the $T$ iterations states that
the whole algorithm is $\left(\epsilon_{\mathrm{total}},\delta_{\mathrm{total}}\right)$-differentially
private, where:
\begin{align*}
\delta_{\mathrm{total}} & =T\left(2\delta_{\mathrm{iter}}\right)+\frac{\delta_{\mathrm{pow}}}{T+1}=\delta_{\mathrm{pow}}\\
\epsilon_{\mathrm{total}}\leq & 2\epsilon_{\mathrm{iter}}\sqrt{2T\log\left(\frac{T+1}{\delta_{\mathrm{pow}}}\right)}+T\left(2\epsilon_{\mathrm{iter}}\right)\left(e^{2\epsilon_{\mathrm{iter}}}-1\right)
\end{align*}
We now show that $\epsilon_{\mathrm{total}}\leq\epsilon_{\mathrm{pow}}$.
We have $2\epsilon_{\mathrm{iter}}<1$ so $e^{2\epsilon_{\mathrm{iter}}}-1\leq2\left(2\epsilon_{\mathrm{iter}}\right)$.
Our choice of $\epsilon_{\mathrm{iter}}$ gives 
\begin{align*}
2\epsilon_{\mathrm{iter}}\sqrt{2T\log\left(\frac{T+1}{\delta_{\mathrm{pow}}}\right)} & =\frac{\epsilon_{\mathrm{pow}}}{2}\\
T\left(2\epsilon_{\mathrm{iter}}\right)\left(e^{2\epsilon_{\mathrm{iter}}}-1\right) & \leq2T\left(2\epsilon_{\mathrm{iter}}\right)^{2}\\
 & =2T\left(\frac{\epsilon_{\mathrm{pow}}}{2\sqrt{2T\log\left(\frac{T+1}{\delta_{\mathrm{pow}}}\right)}}\right)^{2}\\
 & =\frac{\epsilon^{2}_{\mathrm{pow}}}{4\log\left(\frac{T+1}{\delta_{\mathrm{pow}}}\right)}\leq\frac{\epsilon_{\mathrm{pow}}}{2}
\end{align*}
where the last inequality is true since $\epsilon_{\mathrm{pow}}<1\leq2\log\left(\frac{T+1}{\delta_{\mathrm{pow}}}\right)$.

Hence $\epsilon_{\mathrm{total}}\leq\epsilon_{\mathrm{pow}}$, and
Algorithm \ref{alg:Algorithm-power-method} is $\left(\epsilon_{\mathrm{pow}},\delta_{\mathrm{pow}}\right)$-differentially
private.
\end{proof}

Lemma \ref{lem:private-initialization}, Lemma \ref{lem:clipped-iterations-privacy},
and sequential composition imply that Algorithm \ref{alg:Algorithm-main}
is $\left(\epsilon,\delta\right)$-differentially private. 

\subsection{Utility Analysis}

\label{sec:alg-utility}

We start by introducing some notation that we will use throughout
the analysis. We define
\[
\kappa_{k}=\frac{\sigma^{2}_{k}-\sigma^{2}_{k+1}}{\sigma^{2}_{k}}=\frac{\Delta_{k}}{\sigma^{2}_{k}},\qquad\|U\|_{\infty}=\max_{i,j}|U_{ij}|.
\]
For every iteration $t\geq0$ of Algorithm \ref{alg:Algorithm-power-method},
we define
\[
C^{(t)}:=V^{\top}_{k}X^{(t)},D^{(t)}:=V^{\top}_{>k}X^{(t)},s_{t}:=\left\Vert \left(I-X^{(t)}X^{(t)\top}\right)V_{k}\right\Vert _{2}.
\]
For every iteration $t$ in which $C^{(t)}$ is invertible, we define
\[
E^{(t)}:=D^{(t)}\left(C^{(t)}\right)^{-1},m_{t}:=\left\Vert E^{(t)}\right\Vert _{2}.
\]
Our goal is to analyze $m_{t}$ and show that it decreases appropriately
with high probability. Our analysis will show that the matrices $C^{(t)}$
are invertible with high probability and thus the matrices $E^{(t)}$
are well-defined. As shown in the following lemma, having $s_{t}<1$
suffices to ensure that $C^{(t)}$ is invertible.
\begin{lem}
\label{lem:s-to-C-invertible}Let $X\in\R^{d\times k}$ be any matrix
with orthonormal columns. Let $s=\left\Vert \left(I-XX^{\top}\right)V_{k}\right\Vert _{2}$
and $C=V^{\top}_{k}X$. We have $s^{2}=1-\sigma^{2}_{\min}\left(C\right)$,
where $\sigma_{\min}\left(C\right)$ is the smallest singular value
of $C$. Thus, if $s<1$, then $\sigma_{\min}\left(C\right)>0$ and
thus $C$ is invertible.
\end{lem}

\begin{proof}
Since $X$ has orthonormal columns, $P=I-XX^{\top}$ is an orthogonal
projector and thus $P^{\top}P=P^{2}=P$. Thus we have
\begin{align*}
s^{2} & =\left\Vert \left(I-XX^{\top}\right)V_{k}\right\Vert ^{2}_{2}\\
 & =\left\Vert V^{\top}_{k}\left(I-XX^{\top}\right)V_{k}\right\Vert _{2}\\
 & =\left\Vert I-CC^{\top}\right\Vert _{2}\\
 & =1-\sigma^{2}_{\min}\left(C\right)
\end{align*}
as needed.
\end{proof}

We will use the following standard identity relating $s_{t}$ and
$m_{t}$.
\begin{lem}
\label{lem:s-to-m}Let $X\in\R^{d\times k}$ be any matrix with orthonormal
columns. Let $s=\left\Vert \left(I-XX^{\top}\right)V_{k}\right\Vert _{2}$.
Let $C=V^{\top}_{k}X$ and $D=V^{\top}_{>k}X$. Suppose that $C$
is invertible and let $m=\left\Vert DC^{-1}\right\Vert _{2}$. We
have $s=\frac{m}{\sqrt{1+m^{2}}}$. 
\end{lem}

\begin{proof}
Let $E=DC^{-1}$. Since $V=\left[V_{k}\;V_{>k}\right]$ is orthogonal,
we can write
\[
X=V_{k}C+V_{>k}D=\left(V_{k}+V_{>k}E\right)C
\]
Since $C$ is invertible, we have
\[
\range\left(X\right)=\range\left(V_{k}+V_{>k}E\right)
\]
Note that we have
\[
\left(V_{k}+V_{>k}E\right)^{\top}\left(V_{k}+V_{>k}E\right)=I+E^{\top}E
\]
Thus $Y=\left(V_{k}+V_{>k}E\right)\left(I+E^{\top}E\right)^{-1/2}$
is an orthonormal basis for $\range\left(X\right)$, and therefore
$YY^{\top}=XX^{\top}$. Thus
\begin{align*}
\left\Vert \left(I-XX^{\top}\right)V_{k}\right\Vert ^{2}_{2} & =\left\Vert \left(I-YY^{\top}\right)V_{k}\right\Vert ^{2}_{2}=\left\Vert V^{\top}_{k}\left(I-YY^{\top}\right)^{\top}\left(I-YY^{\top}\right)V_{k}\right\Vert _{2}
\end{align*}
Since $Y$ has orthonormal columns, $YY^{\top}$ is an orthogonal
projector. Thus $P=I-YY^{\top}$ is also an orthogonal projector,
and thus $P^{\top}P=P^{2}=P$. Thus
\begin{align*}
\left\Vert \left(I-X^{(t)}X^{(t)\top}\right)V_{k}\right\Vert ^{2}_{2} & =\left\Vert V^{\top}_{k}\left(I-YY^{\top}\right)V_{k}\right\Vert _{2}=\left\Vert I-\left(Y^{\top}V_{k}\right)^{\top}Y^{\top}V_{k}\right\Vert _{2}
\end{align*}
Recall that $Y=\left(V_{k}+V_{>k}E\right)\left(I+E^{\top}E\right)^{-1/2}$.
Thus we have
\[
Y^{\top}V_{k}=\left(I+E^{\top}E\right)^{-1/2}
\]

Thus
\[
\left\Vert \left(I-XX^{\top}\right)V_{k}\right\Vert ^{2}_{2}=\left\Vert I-\left(I+E^{\top}E\right)^{-1}\right\Vert _{2}
\]
If $E$ has singular values $s_{i}$, $I-\left(I+E^{\top}E\right)^{-1}$
has eigenvalues $1-\frac{1}{1+s^{2}_{i}}=\frac{s^{2}_{i}}{1+s^{2}_{i}}$.
Thus
\[
\left\Vert \left(I-XX^{\top}\right)V_{k}\right\Vert ^{2}_{2}=\frac{\left\Vert E\right\Vert ^{2}_{2}}{1+\left\Vert E\right\Vert ^{2}_{2}}
\]
as needed.
\end{proof}

We now proceed with the main analysis that bounds $m_{t}$. For every
$t\geq0$, we also define
\begin{align*}
K^{(t)} & :=\Sigma^{2}_{k}C^{(t)}B^{(t)}\\
S^{(t)} & :=\Sigma_{>k}E^{(t)}\Sigma^{-1}_{k}\\
P^{(t)} & :=\Sigma_{k}C^{(t)}B^{(t)}\\
Y^{(t)} & :=\Sigma^{2}_{>k}E^{(t)}\Sigma^{-2}_{k}\\
F^{(t)} & :=\Delta RV^{\top}_{k}\widetilde{G}^{(t)}\left(K^{(t)}\right)^{-1}\\
L^{(t)} & :=\Delta RV^{\top}_{>k}\widetilde{G}^{(t)}\left(K^{(t)}\right)^{-1}
\end{align*}
We also define the following event.
\begin{defn}
\label{def:gclip} Define 
\begin{align*}
q_{k} & :=2\sqrt{k}+\sqrt{2\log(32T/\beta)},\\
q_{d-k} & :=\sqrt{d-k}+\sqrt{k}+\sqrt{2\log(32T/\beta)},\\
q_{\mathrm{row}} & :=\sqrt{k}+\sqrt{2\log(16nT(T+1)/\beta)}.
\end{align*}
Let $\gclip$ be the event where, simultaneously for all $0\leq t<T$
and $0\leq r\leq T$, we have
\begin{align}
\left\Vert \Delta\sym\left(G^{(t)}\right)\right\Vert _{2}\leq\frac{\gamma}{4},\left\Vert V^{\top}_{k}\widetilde{G}^{(t)}\right\Vert _{2}\leq q_{k},\left\Vert V^{\top}_{>k}\widetilde{G}^{(t)}\right\Vert _{2}\leq q_{d-k},\label{eq:clip-gaussian-blocks}\\
\max_{i}\left\Vert e^{\top}_{i}U_{>k}\Sigma^{2r+1}_{>k}V^{\top}_{>k}\widetilde{G}^{(t)}\right\Vert _{2}\leq\sigma^{2r}_{k+1}q_{\mathrm{row}}\label{eq:clip-gaussian-row}
\end{align}
\end{defn}

\begin{lem}
\label{lem:clip-good-event} We have $\Pr\left(\gclip\right)\geq1-\frac{\beta}{4}$. 
\end{lem}

\begin{proof}
Note that $\left\Vert \sym(G^{(t)})\right\Vert _{2}\leq\left\Vert G^{(t)}\right\Vert _{2}$.
By Lemma \ref{lem:gaussian-norm}, for a fixed $t$, we have $\left\Vert \Delta\sym(G^{(t)})\right\Vert _{2}\leq\gamma/4$
with probability $1-\frac{\beta}{8T}$. Union bound over all $t$
gives the total failure probability $\beta/8$.

By rotational invariance of the normal distribution, $V^{\top}\widetilde{G}^{(t)}$
is identically distributed as $\widetilde{G}^{(t)}$. Consequently,
$V^{\top}_{k}\widetilde{G}^{(t)}\sim\mathcal{N}\left(0,1\right)^{k\times k}$
and $V^{\top}_{>k}\widetilde{G}^{(t)}\sim\mathcal{N}\left(0,1\right)^{(d-k)\times k}$.
The two bounds in (\ref{eq:clip-gaussian-blocks}) then follow from
Lemma \ref{lem:gaussian-norm} with failure probability $\frac{\beta}{32T}$
for each bound and each $t$. Union bound over all $t$ and two bounds
gives the total failure probability of $\beta/16$.

For (\ref{eq:clip-gaussian-row}), the $i$th row $e^{\top}_{i}U_{>k}\Sigma^{2r+1}_{>k}V^{\top}_{>k}\widetilde{G}^{(t)}$
is a $k$-dimensional Gaussian vector whose coordinate variance is
at most $\sigma^{4r}_{k+1}$. Indeed, each coordinate is identically
distributed as a standard Gaussian scaled by
\[
\left\Vert e^{\top}_{i}U_{>k}\Sigma^{2r+1}_{>k}\right\Vert _{2}\leq\sigma^{2r}_{k+1}\left\Vert e^{\top}_{i}U_{>k}\Sigma_{>k}\right\Vert _{2}\leq\sigma^{2r}_{k+1}\|a_{i}\|_{2}
\]

Applying Theorem \ref{thm:gaussian-concentration} and a union bound
over $i,t,r$ give the stated row-wise bound with total failure probability
$\beta/16$.

The total failure probability for all conditions is $\beta/8+\beta/16+\beta/16\le\beta/4$.
\end{proof}

 Next, we show that, under certain assumptions, the preconditioner
$B^{(t)}$ roughly whitens $\Sigma_{k}V^{\top}_{k}X^{(t)}=\Sigma_{k}C^{(t)}$,
i.e., the singular values of $\Sigma_{k}C^{(t)}B^{(t)}$ are bounded
from above and below by absolute constants. To this end, we assume
that $\left\Vert \Delta\sym\left(G^{(t)}\right)\right\Vert _{2}\leq\frac{\gamma}{4}$,
$s_{t}\leq\frac{1}{2}$, and $\gamma\leq\frac{\sigma^{2}_{k}}{4}$.
As shown in the proof of Lemma \ref{lem:clip-good-event}, the first
condition holds with high probability. Later in this section, we will
show that the second condition also holds with high probability and
the third condition is implied by a sufficiently large eigengap $\Delta_{k}$.
Recall that $P^{(t)}:=\Sigma_{k}C^{(t)}B^{(t)}$.
\begin{lem}
\label{lem:whitening}Suppose that $\left\Vert \Delta\sym\left(G^{(t)}\right)\right\Vert _{2}\leq\frac{\gamma}{4}$,
$s_{t}\leq\frac{1}{2}$, and $\gamma\leq\frac{\sigma^{2}_{k}}{4}$.
We have $H^{(t)}=\widetilde{H}^{(t)}$, $B^{(t)}=\left(H^{(t)}\right)^{-1/2}$,
and $\sigma_{\min}\left(P^{(t)}\right)\geq\frac{2}{3}$ and $\left\Vert P^{(t)}\right\Vert _{2}\leq\sqrt{2}$.
\end{lem}

\begin{proof}
We first show that $H^{(t)}=\widetilde{H}^{(t)}$, i.e., we have $H^{(t)}\succeq\gamma I_{k}$.
We have
\begin{align*}
H^{(t)} & =X^{(t)\top}A^{\top}AX^{(t)}+\Delta\sym\left(G^{(t)}\right)\\
 & =X^{(t)\top}\left(V_{k}\Sigma^{2}_{k}V^{\top}_{k}+V_{>k}\Sigma^{2}_{>k}V^{\top}_{>k}\right)X^{(t)}+\Delta\sym\left(G^{(t)}\right)\\
 & =C^{(t)\top}\Sigma^{2}_{k}C^{(t)}+D^{(t)\top}\Sigma^{2}_{>k}D^{(t)}+\Delta\sym\left(G^{(t)}\right).
\end{align*}
We lower bound each of the terms in turn. Consider $C^{(t)\top}\Sigma^{2}_{k}C^{(t)}$.
Since $\left\Vert D^{(t)}\right\Vert _{2}=s_{t}\leq\frac{1}{2}$,
we have $D^{(t)\top}D^{(t)}\preceq\frac{1}{4}I_{k}$. Since $C^{(t)\top}C^{(t)}+D^{(t)\top}D^{(t)}=I_{k}$,
it follows that $C^{(t)\top}C^{(t)}\succeq\frac{3}{4}I_{k}$ and thus
$C^{(t)\top}\Sigma^{2}_{k}C^{(t)}\succeq\frac{3}{4}\sigma^{2}_{k}I_{k}$.
Next, we have $D^{(t)\top}\Sigma^{2}_{>k}D^{(t)}\succeq0$. We have
$\left\Vert \Delta\sym\left(G^{(t)}\right)\right\Vert _{2}\leq\frac{\gamma}{4}\leq\frac{\sigma^{2}_{k}}{16}$
and thus $\Delta\sym\left(G^{(t)}\right)\succeq-\frac{1}{16}\sigma^{2}_{k}I_{k}$.
Putting everything together, we obtain
\[
H^{(t)}\succeq\frac{3}{4}\sigma^{2}_{k}I_{k}-\frac{1}{16}\sigma^{2}_{k}I_{k}=\frac{11}{16}\sigma^{2}_{k}I_{k}\geq\gamma I_{k}.
\]
Thus we have $H^{(t)}=\widetilde{H}^{(t)}$ and $B^{(t)}=\left(H^{(t)}\right)^{-1/2}$.

Finally, consider $P^{(t)}:=\Sigma_{k}C^{(t)}B^{(t)}$. We will lower
and upper bound $P^{(t)\top}P^{(t)}$. We have
\[
P^{(t)\top}P^{(t)}=B^{(t)\top}\left(C^{(t)\top}\Sigma^{2}_{k}C^{(t)}\right)B^{(t)}
\]
Let $M:=X^{(t)\top}A^{\top}AX^{(t)}=C^{(t)\top}\Sigma^{2}_{k}C^{(t)}+D^{(t)\top}\Sigma^{2}_{>k}D^{(t)}$
and $N:=\Delta\sym\left(G^{(t)}\right)$, and thus $H^{(t)}=M+N$.
As shown above, we have $M\succeq\frac{3}{4}\sigma^{2}_{k}I_{k}$
and $N\preceq\frac{1}{16}\sigma^{2}_{k}I_{k}$, and thus
\begin{align*}
-\frac{1}{12}M & \preceq N\preceq\frac{1}{12}M,\\
\frac{11}{12}M & \preceq H^{(t)}\preceq\frac{13}{12}M.
\end{align*}
Thus
\[
\frac{11}{12}B^{(t)\top}MB^{(t)}\preceq B^{(t)\top}H^{(t)}B^{(t)}\preceq\frac{13}{12}B^{(t)\top}MB^{(t)}.
\]
Since $B^{(t)}=\left(H^{(t)}\right)^{-1/2}$ , we have $B^{(t)\top}H^{(t)}B^{(t)}=I_{k}$.
Thus
\[
\frac{12}{13}I_{k}\preceq B^{(t)\top}MB^{(t)}\preceq\frac{12}{11}I_{k}.
\]
Since $D^{(t)\top}\Sigma^{2}_{>k}D^{(t)}\succeq0$, we have $M\succeq C^{(t)\top}\Sigma^{2}_{k}C^{(t)}$
and thus
\[
B^{(t)\top}\left(C^{(t)\top}\Sigma^{2}_{k}C^{(t)}\right)B^{(t)}\preceq B^{(t)\top}MB^{(t)}\preceq\frac{12}{11}I_{k}.
\]
Since $C^{(t)\top}\Sigma^{2}_{k}C^{(t)}\succeq\frac{3}{4}\sigma^{2}_{k}I_{k}$
and $D^{(t)\top}\Sigma^{2}_{>k}D^{(t)}\preceq\frac{1}{4}\sigma^{2}_{k}I_{k}$,
we have $M\preceq\frac{4}{3}\left(C^{(t)\top}\Sigma^{2}_{k}C^{(t)}\right)$
and thus
\[
B^{(t)\top}\left(C^{(t)\top}\Sigma^{2}_{k}C^{(t)}\right)B^{(t)}\succeq\frac{3}{4}B^{(t)\top}MB^{(t)}\succeq\frac{9}{13}I_{k}.
\]
Putting everything together, we have
\[
\frac{9}{13}I_{k}\preceq P^{(t)\top}P^{(t)}\preceq\frac{12}{11}I_{k}
\]
and thus
\[
\sigma_{\min}\left(P^{(t)}\right)\geq\sqrt{\frac{9}{13}}\geq\frac{2}{3},
\]
\[
\left\Vert P^{(t)}\right\Vert _{2}\leq\sqrt{\frac{12}{11}}\leq\sqrt{2}.
\]
\end{proof}

Using Lemma \ref{lem:whitening}, we can show the following.
\begin{lem}
\label{lem:K-inverse} Suppose that $\left\Vert \Delta\sym\left(G^{(t)}\right)\right\Vert _{2}\leq\frac{\gamma}{4}$,
$s_{t}\leq\frac{1}{2}$, and $\gamma\leq\frac{\sigma^{2}_{k}}{4}$.
We have $\left\Vert \left(K^{(t)}\right)^{-1}\right\Vert _{2}\leq\frac{3}{2}\frac{1}{\sigma_{k}}$.
\end{lem}

\begin{proof}
We have $K^{(t)}=\Sigma^{2}_{k}C^{(t)}B^{(t)}=\Sigma_{k}P^{(t)}$.
By Lemma \ref{lem:whitening}, we have $\sigma_{\min}\left(P^{(t)}\right)\ge\frac{2}{3}$.
Thus $P^{(t)}$ is invertible and
\begin{align*}
\left\Vert \left(K^{(t)}\right)^{-1}\right\Vert _{2} & =\left\Vert \left(\Sigma_{k}P^{(t)}\right)^{-1}\right\Vert _{2}=\left\Vert \left(P^{(t)}\right)^{-1}\Sigma^{-1}_{k}\right\Vert _{2}\leq\left\Vert \left(P^{(t)}\right)^{-1}\right\Vert _{2}\left\Vert \Sigma^{-1}_{k}\right\Vert _{2}\\
 & =\frac{1}{\sigma_{\min}\left(P^{(t)}\right)}\frac{1}{\sigma_{k}}\leq\frac{3}{2}\frac{1}{\sigma_{k}}.
\end{align*}
\end{proof}

\paragraph{Main recurrence.}

In the following, we say that clipping is inactive in iteration $t$
if $\Clip_{R}\left(a^{\top}_{i}X^{(t)}B^{(t)}\right)=a^{\top}_{i}X^{(t)}B^{(t)}$
for all $i\in[n]$, i.e., we have $\left\Vert a^{\top}_{i}X^{(t)}B^{(t)}\right\Vert _{2}\leq R$
for all $i\in[n]$. Note that, in this case, the main update in Algorithm
\ref{alg:Algorithm-power-method} is $Z^{(t+1)}=A^{\top}AX^{(t)}B^{(t)}+\Delta R\widetilde{G}^{(t)}$.
First, we will show a recurrence for $E^{(t+1)}$ assuming that clipping
is inactive in iteration $t$ (Lemma \ref{lem:tangent-recurrence}).
This recurrence expresses $E^{(t+1)}$ in terms of the ideal progress
term $\Sigma^{2}_{>k}E^{(t)}\Sigma^{-2}_{k}$ (note that $\left\Vert \Sigma^{2}_{>k}E^{(t)}\Sigma^{-2}_{k}\right\Vert _{2}\leq\frac{\sigma^{2}_{k+1}}{\sigma^{2}_{k}}\left\Vert E^{(t)}\right\Vert _{2}=\left(1-\kappa_{k}\right)\left\Vert E^{(t)}\right\Vert _{2}$),
and additional error terms that we will bound in Theorem \ref{thm:convergence-martingale}.
\begin{lem}
\label{lem:tangent-recurrence} Suppose $C^{(t)}$ is invertible
and clipping is inactive in iteration $t$. If $\left\Vert F^{(t)}\right\Vert _{2}<1$,
then $C^{(t+1)}$ is invertible and 
\begin{align*}
E^{(t+1)}= & \left(\Sigma^{2}_{>k}E^{(t)}\Sigma^{-2}_{k}+L^{(t)}\right)\left(I+F^{(t)}\right)^{-1}.
\end{align*}
\end{lem}

\begin{proof}
Let us start by observing that 
\begin{align*}
C^{(t+1)}R^{(t+1)} & =V^{\top}_{k}X^{(t+1)}R^{(t+1)}\\
 & =V^{\top}_{k}Z^{(t+1)}\\
 & =V^{\top}_{k}\left(A^{\top}AX^{(t)}B^{(t)}+\Delta R\widetilde{G}^{(t)}\right)\\
 & =V^{\top}_{k}A^{\top}AX^{(t)}B^{(t)}+\Delta RV^{\top}_{k}\widetilde{G}^{(t)}\\
 & =\Sigma^{2}_{k}V^{\top}_{k}X^{(t)}B^{(t)}+\Delta RV^{\top}_{k}\widetilde{G}^{(t)}\\
 & =\Sigma^{2}_{k}C^{(t)}B^{(t)}+\Delta RV^{\top}_{k}\widetilde{G}^{(t)}\\
 & =K^{(t)}+\Delta RV^{\top}_{k}\widetilde{G}^{(t)}.
\end{align*}
Since $C^{(t)}$, $\Sigma^{2}_{k}$, and $B^{(t)}$ are invertible,
$K^{(t)}=\Sigma^{2}_{k}C^{(t)}B^{(t)}$ is invertible. Then $F^{(t)}:=\Delta RV^{\top}_{k}\widetilde{G}^{(t)}\left(K^{(t)}\right)^{-1}$
is well-defined.

Multiplying the above identity by $\left(K^{(t)}\right)^{-1}$ then
gives 
\[
C^{(t+1)}R^{(t+1)}\left(K^{(t)}\right)^{-1}=I+\left(\Delta RV^{\top}_{k}\widetilde{G}^{(t)}\right)\left(K^{(t)}\right)^{-1}=I+F^{(t)}
\]
Since $\left\Vert F^{(t)}\right\Vert _{2}<1$, Lemma \ref{lem:neumann-series}
implies that $I+F^{(t)}$ is invertible. Thus $C^{(t+1)}R^{(t+1)}$
is invertible and therefore $C^{(t+1)}$ and $R^{(t+1)}$ are invertible.

This identity also gives 
\[
C^{(t+1)}=\left(I+F^{(t)}\right)K^{(t)}\left(R^{(t+1)}\right)^{-1}.
\]
Now we establish the recurrence. Using the above identity, we obtain
\begin{align*}
E^{(t+1)} & =D^{(t+1)}\left(C^{(t+1)}\right)^{-1}\\
 & =V^{\top}_{>k}X^{(t+1)}\left(C^{(t+1)}\right)^{-1}\\
 & =V^{\top}_{>k}Z^{(t+1)}\left(R^{(t+1)}\right)^{-1}\left(C^{(t+1)}\right)^{-1}\\
 & =V^{\top}_{>k}Z^{(t+1)}\left(C^{(t+1)}R^{(t+1)}\right)^{-1}\\
 & =V^{\top}_{>k}Z^{(t+1)}\left(K^{(t)}\right)^{-1}\left(I+F^{(t)}\right)^{-1}.
\end{align*}
Note that 
\begin{align*}
V^{\top}_{>k}Z^{(t+1)} & =V^{\top}_{>k}\left(A^{\top}AX^{(t)}B^{(t)}+\Delta R\widetilde{G}^{(t)}\right)\\
 & =\Sigma^{2}_{>k}V^{\top}_{>k}X^{(t)}B^{(t)}+\Delta RV^{\top}_{>k}\widetilde{G}^{(t)}\\
 & =\Sigma^{2}_{>k}D^{(t)}B^{(t)}+\Delta RV^{\top}_{>k}\widetilde{G}^{(t)}\\
 & =\Sigma^{2}_{>k}E^{(t)}C^{(t)}B^{(t)}+\Delta RV^{\top}_{>k}\widetilde{G}^{(t)}\\
 & =\Sigma^{2}_{>k}E^{(t)}\Sigma^{-2}_{k}K^{(t)}+\Delta RV^{\top}_{>k}\widetilde{G}^{(t)}\\
 & =\left(\Sigma^{2}_{>k}E^{(t)}\Sigma^{-2}_{k}+\Delta RV^{\top}_{>k}\widetilde{G}^{(t)}\left(K^{(t)}\right)^{-1}\right)K^{(t)}\\
 & =\left(\Sigma^{2}_{>k}E^{(t)}\Sigma^{-2}_{k}+L^{(t)}\right)K^{(t)}.
\end{align*}
Plugging this into the previous identity for $E^{(t+1)}$, we obtain
the claimed recurrence.
\end{proof}

\paragraph{Showing that clipping is inactive.}

Using Lemma \ref{lem:tangent-recurrence}, we will show that, for
an appropriate choice of $R$, clipping is inactive in every iteration
with high probability. Recall that the clipping is inactive in iteration
$t$ if $\left\Vert a^{\top}_{i}X^{(t)}B^{(t)}\right\Vert _{2}\leq R$
for all $i\in[n]$. In the following, we analyze $a^{\top}_{i}X^{(t)}B^{(t)}$
for each $i\in[n]$. Recall that $S^{(t)}:=\Sigma_{>k}E^{(t)}\Sigma^{-1}_{k}$.
\begin{lem}
\label{lem:row-length-bound}Suppose that $\left\Vert \Delta\sym\left(G^{(t)}\right)\right\Vert _{2}\leq\frac{\gamma}{4}$,
$s_{t}\leq\frac{1}{2}$, and $\gamma\leq\frac{\sigma^{2}_{k}}{4}$.
For each $1\leq i\leq n$, we have
\[
\left\Vert a^{\top}_{i}X^{(t)}B^{(t)}\right\Vert _{2}\leq\sqrt{2}\left(\sqrt{k}\left\Vert U\right\Vert _{\infty}+\left\Vert e^{\top}_{i}U_{>k}S^{(t)}\right\Vert _{2}\right).
\]
\end{lem}

\begin{proof}
We have
\begin{align*}
AX^{(t)}B^{(t)} & =\left(U_{k}\Sigma_{k}V^{\top}_{k}+U_{>k}\Sigma_{>k}V^{\top}_{>k}\right)X^{(t)}B^{(t)}\\
 & =U_{k}\underbrace{\Sigma_{k}V^{\top}_{k}X^{(t)}B^{(t)}}_{P^{(t)}}+U_{>k}\Sigma_{>k}\underbrace{V^{\top}_{>k}X^{(t)}}_{D^{(t)}=E^{(t)}C^{(t)}}B^{(t)}\\
 & =\left(U_{k}+U_{>k}S^{(t)}\right)P^{(t)}.
\end{align*}
Thus, for each row $i\in[n]$, we have
\[
a^{\top}_{i}X^{(t)}B^{(t)}=e^{\top}_{i}AX^{(t)}B^{(t)}=e^{\top}_{i}\left(U_{k}+U_{>k}S^{(t)}\right)P^{(t)}
\]
and thus
\[
\left\Vert a^{\top}_{i}X^{(t)}B^{(t)}\right\Vert _{2}\leq\left(\left\Vert e^{\top}_{i}U_{k}\right\Vert _{2}+\left\Vert e^{\top}_{i}U_{>k}S^{(t)}\right\Vert _{2}\right)\left\Vert P^{(t)}\right\Vert _{2}.
\]
By Lemma \ref{lem:whitening}, we have $\left\Vert P^{(t)}\right\Vert _{2}\le\sqrt{2}$.
We also have $\left\Vert e^{\top}_{i}U_{k}\right\Vert _{2}\leq\sqrt{k}\left\Vert U\right\Vert _{\infty}$.
\end{proof}

Next, we upper bound $\max_{i\in\left[n\right]}\left\Vert e^{\top}_{i}U_{>k}S^{(t)}\right\Vert _{2}$.
To this end, we introduce the following potential. For every $0\leq r,t\le T$,
let
\[
\Phi_{r,t}:=\max_{i\in\left[n\right]}\left\Vert e^{\top}_{i}U_{>k}\Sigma^{2r}_{>k}S^{(t)}\right\Vert _{2}.
\]
 Our goal is to upper bound $\Phi_{0,t}=\max_{i\in\left[n\right]}\left\Vert e^{\top}_{i}U_{>k}S^{(t)}\right\Vert _{2}$,
and we will use $\Phi_{r,\cdot}$ for larger $r$ to do so. First,
we show that the initialization allows us to upper bound $\Phi_{r,0}$
for all $0\leq r\leq T$ as follows.
\begin{lem}
\label{lem:initial-potential} Suppose that $\Delta_{k}\geq8\Delta_{\mathrm{init}}\left(2\sqrt{d}+\sqrt{2\log\frac{8}{\beta}}\right)$.
With probability at least $1-\frac{\beta}{4}$, simultaneously for
all $0\leq r\leq T$, we have
\begin{equation}
\Phi_{r,0}\leq\sigma^{2r}_{k+1}\omega_{\mathrm{init}}\quad\omega_{\mathrm{init}}:=3\left\Vert U\right\Vert _{\infty}\left(\sqrt{k}+\sqrt{\log\left(\frac{16n\left(T+1\right)}{\beta}\right)}\right)\label{eq:initial-potential}
\end{equation}
 Let $\ginit$ be this event.
\end{lem}

\begin{proof}
 Recall that we have $\Phi_{r,0}=\max_{i\in\left[n\right]}\left\Vert e^{\top}_{i}U_{>k}\Sigma^{2r+1}_{>k}E^{(0)}\Sigma^{-1}_{k}\right\Vert _{2}$.
To analyze the norm, it is useful to view $E^{(0)}$ as a function
of the Gaussian noise matrix added by Algorithm \ref{alg:Algorithm-initialize}.
More precisely, for a Gaussian matrix $G\sim\N(0,1)^{d\times d}$,
let $X\left(\Delta_{\mathrm{init}}\sym\left(G\right)\right)$ be the
top-$k$ eigenvectors of the matrix $A^{\top}A+\Delta_{\mathrm{init}}\sym\left(G\right)$,
and let
\begin{align*}
C\left(\Delta_{\mathrm{init}}\sym\left(G\right)\right) & :=V^{\top}_{k}X\left(\Delta_{\mathrm{init}}\sym\left(G\right)\right)\\
D\left(\Delta_{\mathrm{init}}\sym\left(G\right)\right) & :=V^{\top}_{>k}X\left(\Delta_{\mathrm{init}}\sym\left(G\right)\right)\\
E\left(\Delta_{\mathrm{init}}\sym\left(G\right)\right) & :=D\left(\Delta_{\mathrm{init}}\sym\left(G\right)\right)\left(C\left(\Delta_{\mathrm{init}}\sym\left(G\right)\right)\right)^{-1}
\end{align*}
where $E\left(\Delta_{\mathrm{init}}\sym\left(G\right)\right)$ is
defined whenever $C\left(\Delta_{\mathrm{init}}\sym\left(G\right)\right)$
is invertible. If the function $M\to E\left(M\right)$ was Lipschitz
under the Frobenius norm, we could use Theorem \ref{thm:gaussian-concentration}
to analyze $\left\Vert e^{\top}_{i}U_{>k}\Sigma^{2r+1}_{>k}E\left(M\right)\Sigma^{-1}_{k}\right\Vert _{2}$.
Unfortunately, the function is not globally Lipschitz, but it is Lipchitz
if we restrict it to the following ball:
\[
\mathcal{B}:=\left\{ M\colon M=M^{\top},\left\Vert M\right\Vert _{2}\leq\frac{\Delta_{k}}{8}\right\} .
\]
 We show in the following claim that the function $M\to E\left(M\right)$
is Lipschitz on $\mathcal{B}$. It is convenient to construct a Lipschitz
extension of this function that is defined everywhere as follows.
Let $\Pi_{\mathcal{B}}$ be the Frobenius norm projection onto $\mathcal{B}$,
i.e., for every matrix $M$, we let $\Pi_{\mathcal{B}}\left(M\right)=\arg\min_{M'\in\mathcal{B}}\left\Vert M-M'\right\Vert _{F}$.
Let
\[
\overline{E}\left(G\right)=E\left(\Pi_{\mathcal{B}}\left(\Delta_{\mathrm{init}}\sym\left(G\right)\right)\right).
\]
Next, we show that $M\to E\left(M\right)$ is Lipschitz on $\mathcal{B}$
under the Frobenius norm, which will then imply that $G\to\overline{E}\left(G\right)$
is globally Lipschitz under the Frobenius norm.
\begin{claim}
\label{claim:em-lipchitz}For any two matrices $M,M'\in\mathcal{B}$
for which $E\left(M\right)$ and $E\left(M'\right)$ are well-defined
(i.e., $C\left(M\right)$ and $C\left(M'\right)$ are invertible),
we have
\[
\left\Vert E\left(M\right)-E\left(M'\right)\right\Vert _{F}\leq\frac{6}{\Delta_{k}}\left\Vert M-M'\right\Vert _{F}.
\]
\end{claim}

\begin{proof}
We start by writing $E\left(M\right)$ and $E\left(M'\right)$ in
a more convenient form. For $W\in\left\{ M,M'\right\} $, let $P\left(W\right):=X\left(W\right)X\left(W\right)^{\top}$.
We can write
\[
V^{\top}P\left(W\right)V=\begin{pmatrix}P_{11}\left(W\right) & P_{12}\left(W\right)\\
P_{21}\left(W\right) & P_{22}\left(W\right)
\end{pmatrix},
\]
where
\begin{align*}
P_{11}\left(W\right) & =V^{\top}_{k}P\left(W\right)V_{k}=C\left(W\right)C\left(W\right)^{\top}\in\R^{k\times k}\\
P_{12}\left(W\right) & =V^{\top}_{k}P\left(W\right)V_{>k}=C\left(W\right)D\left(W\right)^{\top}\in\R^{k\times(d-k)}\\
P_{21}\left(W\right) & =V^{\top}_{>k}P\left(W\right)V_{k}=D\left(W\right)C\left(W\right)^{\top}\in\R^{(d-k)\times k}\\
P_{22}\left(W\right) & =V^{\top}_{>k}P\left(W\right)V_{>k}=D\left(W\right)D\left(W\right)^{\top}\in\R^{(d-k)\times(d-k)}.
\end{align*}
We have
\[
E\left(W\right)=D\left(W\right)C\left(W\right)^{-1}=P_{21}\left(W\right)P_{11}\left(W\right)^{-1}.
\]
Using the above identity, we obtain
\begin{align*}
 & \left\Vert E\left(M\right)-E\left(M'\right)\right\Vert _{F}\\
 & =\left\Vert P_{21}\left(M\right)P_{11}\left(M\right)^{-1}-P_{21}\left(M'\right)P_{11}\left(M'\right)^{-1}\right\Vert _{F}\\
 & =\left\Vert \left(P_{21}\left(M\right)-P_{21}\left(M'\right)\right)P_{11}\left(M\right)^{-1}+P_{21}\left(M'\right)\left(P_{11}\left(M\right)^{-1}-P_{11}\left(M'\right)^{-1}\right)\right\Vert _{F}\\
 & \overset{(a)}{=}\left\Vert \left(P_{21}\left(M\right)-P_{21}\left(M'\right)\right)P_{11}\left(M\right)^{-1}+P_{21}\left(M'\right)P_{11}\left(M\right)^{-1}\left(P_{11}\left(M'\right)-P_{11}\left(M\right)\right)P_{11}\left(M'\right)^{-1}\right\Vert _{F}\\
 & \overset{(b)}{\leq}\left\Vert P_{21}\left(M\right)-P_{21}\left(M'\right)\right\Vert _{F}\left\Vert P_{11}\left(M\right){}^{-1}\right\Vert _{2}\\
 & \ \ +\left\Vert P_{21}\left(M'\right)\right\Vert _{2}\left\Vert P_{11}\left(M\right)^{-1}\right\Vert _{2}\left\Vert P_{11}\left(M'\right)^{-1}\right\Vert _{2}\left\Vert P\left(M\right)-P\left(M'\right)\right\Vert _{F}\\
 & \overset{(c)}{\leq}\left(1+\left\Vert P_{21}\left(M'\right)\right\Vert _{2}\left\Vert P_{11}\left(M'\right)^{-1}\right\Vert _{2}\right)\left\Vert P_{11}\left(M\right)^{-1}\right\Vert _{2}\left\Vert P\left(M\right)-P\left(M'\right)\right\Vert _{F}
\end{align*}
In $(a)$, we used the identity $A^{-1}-B^{-1}=A^{-1}\left(B-A\right)B^{-1}$.
In $(b)$, we used the inequalities $\left\Vert A+B\right\Vert _{F}\leq\left\Vert A\right\Vert _{F}+\left\Vert B\right\Vert _{F}$,
$\left\Vert AB\right\Vert _{F}\leq\left\Vert A\right\Vert _{F}\left\Vert B\right\Vert _{2}$,
and $\left\Vert AB\right\Vert _{2}\leq\left\Vert A\right\Vert _{2}\left\Vert B\right\Vert _{2}$.
In $(c)$, we used that $\left\Vert P_{21}\left(M\right)-P_{21}\left(M'\right)\right\Vert _{F}\leq\left\Vert P\left(M\right)-P\left(M'\right)\right\Vert _{F}$,
which we can show as follows: 
\begin{align*}
\left\Vert P_{21}\left(M\right)-P_{21}\left(M'\right)\right\Vert _{F} & =\left\Vert V^{\top}_{>k}\left(P\left(M\right)-P\left(M'\right)\right)V_{k}\right\Vert _{F}\\
 & \leq\left\Vert V^{\top}_{>k}\right\Vert _{2}\left\Vert P\left(M\right)-P\left(M'\right)\right\Vert _{F}\left\Vert V_{k}\right\Vert _{2}\\
 & =\left\Vert P\left(M\right)-P\left(M'\right)\right\Vert _{F}.
\end{align*}
where the last equality follows from the fact that $V_{k}$ and $V_{>k}$
have orthonormal columns.

Next, we bound the terms in the above inequality. Consider $\left\Vert P_{21}\left(M'\right)\right\Vert _{2}$.
We have
\begin{align*}
\left\Vert P_{21}\left(M'\right)\right\Vert _{2} & =\left\Vert V^{\top}_{>k}X\left(M'\right)X\left(M'\right)^{\top}V_{k}\right\Vert _{2}\\
 & \leq\left\Vert V^{\top}_{>k}\right\Vert _{2}\left\Vert X\left(M'\right)\right\Vert ^{2}_{2}\left\Vert V_{k}\right\Vert ^{2}\leq1
\end{align*}
where the last inequality follows from the fact that $V_{k}$, $V_{>k}$
, and $X\left(M'\right)$ have orthonormal columns.

Consider $\left\Vert P_{11}\left(M\right){}^{-1}\right\Vert _{2}$
and $\left\Vert P_{11}\left(M'\right)^{-1}\right\Vert _{2}$. For
$W\in\left\{ M,M'\right\} $, we have
\begin{align*}
\left\Vert I_{k}-P_{11}\left(W\right)\right\Vert _{2} & =\left\Vert V^{\top}_{k}\left(I-P\left(W\right)\right)V_{k}\right\Vert _{2}=\left\Vert \left(I-P\left(W\right)\right)V_{k}\right\Vert ^{2}_{2}
\end{align*}
where we used that $P_{11}\left(W\right)=V^{\top}_{k}P\left(W\right)V_{k}$
and $Q:=I-P\left(W\right)$ is an orthogonal projector and thus $Q^{\top}Q=Q$.

Lemma \ref{lem:davis-kahan-error} together with $\left\Vert W\right\Vert _{2}\leq\frac{\Delta_{k}}{8}$
give
\[
\left\Vert \left(I-P\left(W\right)\right)V_{k}\right\Vert _{2}\leq\frac{\left\Vert W\right\Vert _{2}}{\Delta_{k}-\left\Vert W\right\Vert _{2}}\leq\frac{2\left\Vert W\right\Vert _{2}}{\Delta_{k}}\leq\frac{1}{4}.
\]
Thus
\[
\left\Vert I_{k}-P_{11}\left(W\right)\right\Vert _{2}\leq\frac{1}{16}
\]
and therefore $P_{11}\left(W\right)\succeq\frac{15}{16}I_{k}$ and
$\left\Vert P_{11}\left(W\right)^{-1}\right\Vert _{2}\leq\frac{16}{15}$.

Finally, consider $\left\Vert P\left(M\right)-P\left(M'\right)\right\Vert _{F}$.
We have
\begin{align*}
\left\Vert P\left(M\right)-P\left(M'\right)\right\Vert _{F} & =\left\Vert X\left(M\right)X\left(M\right)^{\top}-X\left(M'\right)X\left(M'\right)^{\top}\right\Vert _{F}\\
 & =\sqrt{2}\left\Vert \left(I-X\left(M'\right)X\left(M'\right)^{\top}\right)X\left(M\right)\right\Vert _{F}
\end{align*}
We can upper bound the last term using the Davis-Kahan theorem (Theorem
\ref{thm:davis-kahan}) and Weyl's inequality (Theorem \ref{thm:weyl})
as follows. Let
\[
\delta=\max\left\{ \lambda_{k}\left(A^{\top}A+M\right)-\lambda_{k+1}\left(A^{\top}A+M'\right),\lambda_{k}\left(A^{\top}A+M'\right)-\lambda_{k+1}\left(A^{\top}A+M\right)\right\} .
\]
We apply the Davis-Kahan theorem with matrices $A^{\top}A+M$ and
$A^{\top}A+M'$, and $r=1$ and $s=k$. We obtain
\[
\left\Vert \left(I-X\left(M'\right)X\left(M'\right)^{\top}\right)X\left(M\right)\right\Vert _{F}\leq\frac{\left\Vert M-M'\right\Vert _{F}}{\delta}.
\]
To lower bound $\delta$, we use Weyl's inequality and the fact that
$\left\Vert M\right\Vert _{2},\left\Vert M'\right\Vert _{2}\leq\frac{\Delta_{k}}{8}$.
By Weil's inequality, for $W\in\left\{ M,M'\right\} $ and $j\in\left\{ k,k+1\right\} $,
we have
\[
\left|\lambda_{j}\left(A^{\top}A+W\right)-\lambda_{j}\left(A^{\top}A\right)\right|\leq\left\Vert W\right\Vert _{2}.
\]
Thus
\begin{align*}
\lambda_{k}\left(A^{\top}A+M\right)-\lambda_{k+1}\left(A^{\top}A+M'\right) & \geq\left(\lambda_{k}\left(A^{\top}A\right)-\left\Vert M\right\Vert _{2}\right)-\left(\lambda_{k+1}\left(A^{\top}A\right)+\left\Vert M'\right\Vert _{2}\right)\\
 & =\sigma^{2}_{k}-\sigma^{2}_{k+1}-\left\Vert M\right\Vert _{2}-\left\Vert M'\right\Vert _{2}\\
 & =\Delta_{k}-\left\Vert M\right\Vert _{2}-\left\Vert M'\right\Vert _{2}\\
 & \ge\frac{3}{4}\Delta_{k}
\end{align*}
where in the last inequality we used that $\left\Vert M\right\Vert _{2},\left\Vert M'\right\Vert _{2}\leq\frac{\Delta_{k}}{8}$.

Similarly, we have
\[
\lambda_{k}\left(A^{\top}A+M'\right)-\lambda_{k+1}\left(A^{\top}A+M\right)\geq\frac{3}{4}\Delta_{k}
\]
and thus $\delta\geq\frac{3}{4}\Delta_{k}$.

Thus we obtain
\begin{align*}
\left\Vert \left(I-X\left(M'\right)X\left(M'\right)^{\top}\right)X\left(M\right)\right\Vert _{F} & \leq\frac{\left\Vert M-M'\right\Vert _{F}}{\delta}\leq\frac{4}{3\Delta_{k}}\left\Vert M-M'\right\Vert _{F}.
\end{align*}
and
\[
\left\Vert P\left(M\right)-P\left(M'\right)\right\Vert _{F}=\sqrt{2}\left\Vert \left(I-X\left(M'\right)X\left(M'\right)^{\top}\right)X\left(M\right)\right\Vert _{F}\leq\frac{8}{3\Delta_{k}}\left\Vert M-M'\right\Vert _{F}.
\]
Putting everything together, we have
\[
\left\Vert E\left(M\right)-E\left(M'\right)\right\Vert _{F}\leq\left(1+\frac{16}{15}\right)\frac{16}{15}\frac{8}{3\Delta_{k}}\left\Vert M-M'\right\Vert _{F}\leq\frac{6}{\Delta_{k}}\left\Vert M-M'\right\Vert _{F},
\]
as needed.
\end{proof}

Using the above claim and the fact that $\Pi_{\mathcal{B}}\left(\cdot\right)$
and $\sym\left(\cdot\right)$ are non-expansive in Frobenius norm,
we can show that $G\to\overline{E}\left(G\right)$ is Lipschitz.
\begin{claim}
For any two matrices $G,G'$ for which $\overline{E}\left(G\right)$
and $\overline{E}\left(G'\right)$ are well-defined, we have
\[
\left\Vert \overline{E}\left(G\right)-\overline{E}\left(G'\right)\right\Vert _{F}\leq\frac{6\Delta_{\mathrm{init}}}{\Delta_{k}}\left\Vert G-G'\right\Vert _{F}.
\]

\begin{proof}
We have
\begin{align*}
\left\Vert \overline{E}\left(G\right)-\overline{E}\left(G'\right)\right\Vert _{F} & =\left\Vert E\left(\Pi_{\mathcal{B}}\left(\Delta_{\mathrm{init}}\sym\left(G\right)\right)\right)-E\left(\Pi_{\mathcal{B}}\left(\Delta_{\mathrm{init}}\sym\left(G'\right)\right)\right)\right\Vert _{F}\\
 & \leq\frac{6}{\Delta_{k}}\left\Vert \Pi_{\mathcal{B}}\left(\Delta_{\mathrm{init}}\sym\left(G\right)\right)-\Pi_{\mathcal{B}}\left(\Delta_{\mathrm{init}}\sym\left(G'\right)\right)\right\Vert _{F}\\
 & \leq\frac{6}{\Delta_{k}}\left\Vert \Delta_{\mathrm{init}}\sym\left(G\right)-\Delta_{\mathrm{init}}\sym\left(G'\right)\right\Vert _{F}\\
 & \leq\frac{6\Delta_{\mathrm{init}}}{\Delta_{k}}\left\Vert G-G'\right\Vert _{F}
\end{align*}
where we used the previous claim and the fact that $\Pi_{\mathcal{B}}\left(\cdot\right)$
and $\sym\left(\cdot\right)$ are non-expansive in Frobenius norm.
\end{proof}

\end{claim}

Next, we use Gaussian concentration to upper bound $\left\Vert e^{\top}_{i}U_{>k}\Sigma^{2r+1}_{>k}\overline{E}\left(G^{(\mathrm{init})}\right)\Sigma^{-1}_{k}\right\Vert _{2}$
for every row $1\leq i\leq n$ and every $0\leq r\leq T$. To this
end, for a fixed row $i$, exponent $r$, and an arbitrary unit vector
$z\in\R^{k}$, we consider the following scalar-valued function:
\[
f_{i,r,z}\left(G\right):=e^{\top}_{i}U_{>k}\Sigma^{2r+1}_{>k}\overline{E}\left(G\right)\Sigma^{-1}_{k}z
\]
Since $\overline{E}\left(G\right)$ is Lipschitz, we can show that
$f_{i,r,z}$ is Lipschitz.
\begin{claim}
For any two matrices $G,G'$ for which $\overline{E}\left(G\right)$
and $\overline{E}\left(G'\right)$ are well-defined, we have
\[
\left|f_{i,r,z}\left(G\right)-f_{i,r,z}\left(G'\right)\right|\leq L_{i,r}\left\Vert G-G'\right\Vert _{F},
\]
where $L_{i,r}:=\frac{6\Delta_{\mathrm{init}}}{\Delta_{k}}\frac{\sigma^{2r+1}_{k+1}}{\sigma_{k}}\left\Vert e^{\top}_{i}U_{>k}\right\Vert _{2}$.
\end{claim}

\begin{proof}
We have
\begin{align*}
\left|f_{i,r,z}\left(G\right)-f_{i,r,z}\left(G'\right)\right| & =\left|e^{\top}_{i}U_{>k}\Sigma^{2r+1}_{>k}\left(\overline{E}\left(G\right)-\overline{E}\left(G'\right)\right)\Sigma^{-1}_{k}z\right|\\
 & \leq\left\Vert e^{\top}_{i}U_{>k}\right\Vert _{2}\left\Vert \Sigma^{2r+1}_{>k}\right\Vert _{2}\left\Vert \overline{E}\left(G\right)-\overline{E}\left(G'\right)\right\Vert _{F}\left\Vert \Sigma^{-1}_{k}\right\Vert _{2}\left\Vert z\right\Vert _{2}\\
 & =\left\Vert e^{\top}_{i}U_{>k}\right\Vert _{2}\sigma^{2r+1}_{k+1}\left\Vert \overline{E}\left(G\right)-\overline{E}\left(G'\right)\right\Vert _{F}\frac{1}{\sigma_{k}}\\
 & \leq\left\Vert e^{\top}_{i}U_{>k}\right\Vert _{2}\sigma^{2r+1}_{k+1}\frac{6\Delta_{\mathrm{init}}}{\Delta_{k}}\left\Vert G-G'\right\Vert _{F}\frac{1}{\sigma_{k}}
\end{align*}
where in the last inequality we used the fact that $\overline{E}\left(G\right)$
is $\frac{6\Delta_{\mathrm{init}}}{\Delta_{k}}$-Lipschitz.
\end{proof}

Next, we show that $\E\left[f_{i,r,z}\left(G\right)\right]=0$.
\begin{claim}
We have $\E\left[f_{i,r,z}\left(G\right)\right]=0$.
\end{claim}

\begin{proof}
Consider a Gaussian matrix $G\sim\mathcal{N}\left(0,1\right)^{d\times d}$.
Let
\[
J=V\begin{pmatrix}-I_{k} & 0\\
0 & I_{d-k}
\end{pmatrix}V^{\top}.
\]
The matrix $J$ is orthogonal and symmetric. Let $G'=JGJ^{\top}$.
Since the Gaussian distribution is rotationally invariant, $G'$ has
the same distribution as $G$. Next, we will show that $f_{i,r,z}\left(G'\right)=-f_{i,r,z}\left(G\right)$
which will imply that $\E\left[f_{i,r,z}\left(G\right)\right]=0$.
Let $M:=\Pi_{\mathcal{B}}\left(\Delta_{\mathrm{init}}\sym\left(G\right)\right)$
and $M':=\Pi_{\mathcal{B}}\left(\Delta_{\mathrm{init}}\sym\left(G'\right)\right)$.
Note that we have $M'=JMJ^{\top}$ and $J\left(A^{\top}A\right)J^{\top}=A^{\top}A$,
and thus
\[
A^{\top}A+M'=J\left(A^{\top}A+M\right)J^{\top}.
\]
As in the proof of Claim \ref{claim:em-lipchitz}, for $W\in\left\{ M,M'\right\} $,
we let $P\left(W\right)=X\left(W\right)X\left(W\right)^{\top}$ and
we write it as follows:
\[
V^{\top}P\left(W\right)V=\begin{pmatrix}P_{11}\left(W\right) & P_{12}\left(W\right)\\
P_{21}\left(W\right) & P_{22}\left(W\right)
\end{pmatrix}
\]
where $P_{11}\left(W\right)=V^{\top}_{k}P\left(W\right)V_{k}$, $P_{12}\left(W\right)=V^{\top}_{k}P\left(W\right)V_{>k}$,
$P_{21}\left(W\right)=V^{\top}_{>k}P\left(W\right)V_{k}$, and $P_{22}\left(W\right)=V^{\top}_{>k}P\left(W\right)V_{>k}$.

Since $A^{\top}A+M'=J\left(A^{\top}A+M\right)J^{\top}$, $JV_{k}=-V_{k}$,
and $JV_{>k}=V_{>k}$, we have $P\left(M'\right)=JP\left(M\right)J^{\top}$
and
\[
V^{\top}P\left(M'\right)V=\begin{pmatrix}P_{11}\left(M\right) & -P_{12}\left(M\right)\\
-P_{21}\left(M\right) & P_{22}\left(M\right)
\end{pmatrix}
\]
Since $E\left(M\right)=P_{21}\left(M\right)P_{11}\left(M\right)^{-1}$,
we obtain that $\overline{E}\left(G'\right)=-\overline{E}\left(G\right)$.
Therefore $f_{i,r,z}\left(G'\right)=-f_{i,r,z}\left(G\right)$ and
\[
\E\left[f_{i,r,z}\left(G\right)\right]=\E\left[f_{i,r,z}\left(G'\right)\right]=-\E\left[f_{i,r,z}\left(G\right)\right].
\]
The expectation exists since $f_{i,r,z}$ is globally Lipschitz, and
thus $\E\left[f_{i,r,z}\left(G\right)\right]=0$.
\end{proof}

By applying Theorem \ref{thm:gaussian-concentration} to $f_{i,r,z}$
and using a net argument, we obtain the following guarantee.
\begin{claim}
\label{claim:norm-concentration-initial}With probability at least
$1-\frac{\beta}{8}$, simultaneously for all $1\leq i\leq n$ and
$0\leq r\leq T$, we have
\[
\left\Vert e^{\top}_{i}U_{>k}\Sigma^{2r+1}_{>k}\overline{E}\left(G^{(\mathrm{init})}\right)\Sigma^{-1}_{k}\right\Vert _{2}\leq2L_{i,r}\left(2\sqrt{k}+\sqrt{2\log\left(\frac{16n\left(T+1\right)}{\beta}\right)}\right).
\]
\end{claim}

\begin{proof}
We discretize the set of all possible unit vectors $z\in\R^{k}$ (i.e.,
the unit sphere in $\R^{k}$) using a $1/2$-net $\calN$, i.e., $\calN$
is a set of unit vectors with the property that, for every unit vector
$z\in\R^{k}$, there is a vector $z'\in\mathcal{N}$ such that $\left\Vert z-z'\right\Vert _{2}\leq1/2$.
A standard volume argument shows that there exists a net of size $\left|\calN\right|\leq5^{k}$.
In the following, we let $\calN$ be a $1/2$-net of size $\left|\calN\right|\leq5^{k}$.

Fix a row index $i$ and exponent $r$. Let $w:=e^{\top}_{i}U_{>k}\Sigma^{2r+1}_{>k}\overline{E}\left(G^{(\mathrm{init})}\right)\Sigma^{-1}_{k}$.
Note that $w\in\mathbb{R}^{1\times k}$ is a row vector. In the following,
we upper bound $\left\Vert w\right\Vert _{2}.$ We may assume that
$w\neq0$, since otherwise $\left\Vert w\right\Vert _{2}=0$ and we
are done. Let $\overline{w}=w^{\top}/\left\Vert w\right\Vert _{2}$
and thus $\left\Vert w\right\Vert _{2}=w\overline{w}$. Since $\calN$
is a $1/2$-net, there is a vector $z'\in\calN$ such that $\left\Vert \overline{w}-z'\right\Vert _{2}\leq1/2$.
We have
\begin{align*}
\left\Vert w\right\Vert _{2} & =w\overline{w}=wz'+w\left(\overline{w}-z'\right)\\
 & \leq\left|wz'\right|+\left\Vert w\right\Vert _{2}\left\Vert \overline{w}-z'\right\Vert _{2}\\
 & \leq\max_{z\in\calN}\left|wz\right|+\frac{1}{2}\left\Vert w\right\Vert _{2}.
\end{align*}
Rearranging, we obtain
\[
\left\Vert w\right\Vert _{2}\leq2\max_{z\in\calN}\left|wz\right|.
\]
Since $w=e^{\top}_{i}U_{>k}\Sigma^{2r+1}_{>k}\overline{E}\left(G^{(\mathrm{init})}\right)\Sigma^{-1}_{k}$
and $f_{i,r,z}\left(G^{(\mathrm{init})}\right)=wz$, we have
\[
\left\Vert e^{\top}_{i}U_{>k}\Sigma^{2r+1}_{>k}\overline{E}\left(G^{(\mathrm{init})}\right)\Sigma^{-1}_{k}\right\Vert _{2}\leq2\max_{z\in\mathcal{N}}\left|f_{i,r,z}\left(G^{(\mathrm{init})}\right)\right|.
\]
Since $f_{i,r,z}$ is $L_{i,r}$-Lipschitz and has expectation $0$,
Theorem \ref{thm:gaussian-concentration} gives that, for every $x\geq0$,
\[
\Pr\left[\left|f_{i,r,z}\left(G^{(\mathrm{init})}\right)\right|\geq L_{i,r}x\right]\le2\exp\left(-\frac{x^{2}}{2}\right).
\]
There are $n\left(T+1\right)\left|\calN\right|\leq n\left(T+1\right)5^{k}$
triples $\left(i,r,z\right)$ where $1\leq i\leq n$ is a row index,
$0\leq r\leq T$ is an exponent, and $z$ is a vector in $\calN$.
Setting $x=\sqrt{2\log\left(\frac{16n\left(T+1\right)5^{k}}{\beta}\right)}$
in the above inequality and using an union bound over all triples
$\left(i,r,z\right)$, we obtain that, with probability at least $1-\frac{\beta}{8}$,
simultaneously for all $i$ and $r$, we have
\begin{align*}
\left\Vert e^{\top}_{i}U_{>k}\Sigma^{2r+1}_{>k}\overline{E}\left(G^{(\mathrm{init})}\right)\Sigma^{-1}_{k}\right\Vert _{2} & \leq2\max_{z\in\mathcal{N}}\left|f_{i,r,z}\left(G^{(\mathrm{init})}\right)\right|\\
 & \leq2L_{i,r}\sqrt{2\log\left(\frac{16n\left(T+1\right)5^{k}}{\beta}\right)}\\
 & \leq2L_{i,r}\left(2\sqrt{k}+\sqrt{2\log\left(\frac{16n\left(T+1\right)}{\beta}\right)}\right)
\end{align*}
as needed.
\end{proof}

Finally, we show that the assumption on $\Delta_{k}$ implies that,
with probability at least $1-\frac{\beta}{8}$, $\Delta_{\mathrm{init}}\sym\left(G^{(\mathrm{init})}\right)\in\mathcal{B}$
and thus $\overline{E}\left(G^{(\mathrm{init})}\right)=E^{(0)}$.
\begin{claim}
Suppose that $\Delta_{k}\geq8\Delta_{\mathrm{init}}\left(2\sqrt{d}+\sqrt{2\log\frac{8}{\beta}}\right)$.
With probability at least $1-\frac{\beta}{8}$, $\left\Vert \Delta_{\mathrm{init}}\sym\left(G^{(\mathrm{init})}\right)\right\Vert _{2}\leq\frac{1}{8}\Delta_{k}$
and thus $\Delta_{\mathrm{init}}\sym\left(G^{(\mathrm{init})}\right)\in\mathcal{B}$
and $\overline{E}\left(G^{(\mathrm{init})}\right)=E^{(0)}$.
\end{claim}

\begin{proof}
By Lemma \ref{lem:gaussian-norm}, we have
\[
\Pr\left[\left\Vert \sym\left(G^{(\mathrm{init})}\right)\right\Vert _{2}>2\sqrt{d}+\sqrt{2\log\frac{8}{\beta}}\right]\leq\exp\left(-\frac{1}{2}\left(\sqrt{2\log\frac{8}{\beta}}\right)^{2}\right)=\frac{\beta}{8}.
\]
Thus, with probability $1-\beta/8$, we have
\[
\Delta_{\mathrm{init}}\left\Vert \sym\left(G^{(\mathrm{init})}\right)\right\Vert _{2}\leq\Delta_{\mathrm{init}}\left(2\sqrt{d}+\sqrt{2\log\frac{8}{\beta}}\right)\le\frac{1}{8}\Delta_{k}.
\]
\end{proof}

Thus, by combining the last two claims, we obtain that, with probability
at least $1-\frac{\beta}{4}$, simultaneously for all $0\leq r\leq T$,
we have
\begin{align*}
\Phi_{r,0} & :=\max_{i\in\left[n\right]}\left\Vert e^{\top}_{i}U_{>k}\Sigma^{2r+1}_{>k}E^{(0)}\Sigma^{-1}_{k}\right\Vert _{2}\\
 & =\max_{i\in\left[n\right]}\left\Vert e^{\top}_{i}U_{>k}\Sigma^{2r+1}_{>k}\overline{E}\left(G^{(\mathrm{init})}\right)\Sigma^{-1}_{k}\right\Vert _{2}\\
 & \leq\max_{i\in\left[n\right]}2L_{i,r}\left(2\sqrt{k}+\sqrt{2\log\left(\frac{16n\left(T+1\right)}{\beta}\right)}\right)\\
 & =\max_{i\in\left[n\right]}\frac{12\Delta_{\mathrm{init}}}{\Delta_{k}}\frac{\sigma^{2r+1}_{k+1}}{\sigma_{k}}\left\Vert e^{\top}_{i}U_{>k}\right\Vert _{2}\left(2\sqrt{k}+\sqrt{2\log\left(\frac{16n\left(T+1\right)}{\beta}\right)}\right).
\end{align*}
Using that $\left\Vert e^{\top}_{i}U_{>k}\right\Vert _{2}\leq\sqrt{d-k}\left\Vert U\right\Vert _{\infty}$
and the assumption $\Delta_{k}\geq8\Delta_{\mathrm{init}}\left(2\sqrt{d}+\sqrt{2\log\frac{8}{\beta}}\right)$,
we can further upper bound
\begin{align*}
\Phi_{r,0} & \leq\frac{12\Delta_{\mathrm{init}}}{\Delta_{k}}\frac{\sigma^{2r+1}_{k+1}}{\sigma_{k}}\sqrt{d-k}\left\Vert U\right\Vert _{\infty}\left(2\sqrt{k}+\sqrt{2\log\left(\frac{16n\left(T+1\right)}{\beta}\right)}\right)\\
 & \leq\frac{12}{8\left(2\sqrt{d}+\sqrt{2\log\frac{8}{\beta}}\right)}\frac{\sigma^{2r+1}_{k+1}}{\sigma_{k}}\sqrt{d-k}\left\Vert U\right\Vert _{\infty}\left(2\sqrt{k}+\sqrt{2\log\left(\frac{16n\left(T+1\right)}{\beta}\right)}\right)\\
 & \leq\frac{3}{2}\frac{\sigma^{2r+1}_{k+1}}{\sigma_{k}}\left\Vert U\right\Vert _{\infty}\left(2\sqrt{k}+\sqrt{2\log\left(\frac{16n\left(T+1\right)}{\beta}\right)}\right)\\
 & \leq\frac{3}{2}\sigma^{2r}_{k+1}\left\Vert U\right\Vert _{\infty}\left(2\sqrt{k}+\sqrt{2\log\left(\frac{16n\left(T+1\right)}{\beta}\right)}\right)\\
 & \leq3\sigma^{2r}_{k+1}\left\Vert U\right\Vert _{\infty}\left(\sqrt{k}+\sqrt{\log\left(\frac{16n\left(T+1\right)}{\beta}\right)}\right)
\end{align*}
as needed.
\end{proof}

Using Lemma \ref{lem:initial-potential}, we will show by induction
that the clipping is inactive in all iterations with high probability.
We will use the following lemma for the inductive step. Recall that
$\gclip$ is the event defined in Definition \ref{def:gclip} and
$\ginit$ is the event in Lemma \ref{lem:initial-potential}. 
\begin{lem}
\label{lem:clipping-inactive} Suppose that $\Delta_{k}\geq\max\left\{ 8\Delta_{\mathrm{init}}\left(2\sqrt{d}+\sqrt{2\log\frac{8}{\beta}}\right),8\Delta q_{\mathrm{row}}\right\} $,
$\gamma\leq\frac{\sigma^{2}_{k}}{4},$ $p_{R}:=\frac{3\Delta Rq_{k}}{2\sigma_{k}}\leq\frac{\kappa_{k}}{16}$,
and $R\geq15\left\Vert U\right\Vert _{\infty}\left(\sqrt{k}+\sqrt{\log\left(\frac{16n\left(T+1\right)}{\beta}\right)}\right)$.
Suppose that $\mathcal{G}_{\mathrm{clip}}$ and $\ginit$ hold. If
$s_{j}\le\frac{1}{2}$ for all $0\leq j\leq t$ and clipping was inactive
in all iterations $0\leq j\leq t-1$, then clipping is also inactive
in iteration $t$. 
\end{lem}

\begin{proof}
First, let us verify that the conditions of Lemma \ref{lem:tangent-recurrence}
are satisfied for all iterations $j\leq t-1$. By assumption, clipping
is inactive in each of those iterations, and thus it only remains
to check that we have $\left\Vert F^{(j)}\right\Vert _{2}<1$ for
each $j\leq t-1$. We have
\begin{align*}
\left\Vert F^{(j)}\right\Vert _{2} & =\left\Vert \Delta RV^{\top}_{k}\widetilde{G}^{(j)}\left(K^{(j)}\right)^{-1}\right\Vert _{2}\\
 & \leq\Delta R\left\Vert V^{\top}_{k}\widetilde{G}^{(j)}\right\Vert _{2}\left\Vert \left(K^{(j)}\right)^{-1}\right\Vert _{2}\\
 & \leq\Delta Rq_{k}\frac{3}{2}\frac{1}{\sigma_{k}}=p_{R}.
\end{align*}
In the last inequality, we used that $\gclip$ holds (and thus $\left\Vert V^{\top}_{k}\widetilde{G}^{(j)}\right\Vert _{2}\leq q_{k}$)
and Lemma \ref{lem:K-inverse} (note that the conditions of the lemma
are satisfied by assumption). By assumption, we have $p_{R}\leq\frac{\kappa_{k}}{16}\leq\frac{1}{16}$,
and thus $\left\Vert F^{(j)}\right\Vert _{2}<1$ as needed.

Thus, by Lemma \ref{lem:tangent-recurrence}, for all iterations $j\leq t-1$,
we have that $C^{(j+1)}$ is invertible and
\[
E^{(j+1)}=\left(\Sigma^{2}_{>k}E^{(j)}\Sigma^{-2}_{k}+L^{(j)}\right)\left(I+F^{(j)}\right)^{-1}.
\]
Recall that $S^{(j)}:=\Sigma_{>k}E^{(j)}\Sigma^{-1}_{k}$ for all
$j$. Multiplying the above recurrence with $\Sigma_{>k}$ on the
left and $\Sigma^{-1}_{k}$ on the right, we obtain
\[
S^{(j+1)}=\left(\Sigma^{2}_{>k}S^{(j)}\Sigma^{-1}_{k}+\Sigma_{>k}L^{(j)}\right)\left(I+F^{(j)}\right)^{-1}\Sigma^{-1}_{k}.
\]
Using the above recurrence, we can obtain a recurrence for the potential
$\Phi_{r,t}$. Recall that $L^{(j)}=\Delta RV^{\top}_{>k}\widetilde{G}^{(j)}\left(K^{(j)}\right)^{-1}$.
For every $r\geq0$ and every $j\leq t-1$, we have
\begin{align*}
\Phi_{r,j+1} & =\max_{i\in\left[n\right]}\left\Vert e^{\top}_{i}U_{>k}\Sigma^{2r}_{>k}S^{(j+1)}\right\Vert _{2}\\
 & =\max_{i\in\left[n\right]}\left\Vert \left(e^{\top}_{i}U_{>k}\Sigma^{2\left(r+1\right)}_{>k}S^{(j)}\Sigma^{-1}_{k}+e^{\top}_{i}U_{>k}\Sigma^{2r}_{>k}\Sigma_{>k}L^{(j)}\right)\left(I+F^{(j)}\right)^{-1}\Sigma^{-1}_{k}\right\Vert _{2}\\
 & \leq\max_{i\in\left[n\right]}\left(\left\Vert e^{\top}_{i}U_{>k}\Sigma^{2\left(r+1\right)}_{>k}S^{(j)}\right\Vert _{2}\left\Vert \Sigma^{-1}_{k}\right\Vert _{2}+\Delta R\left\Vert e^{\top}_{i}U_{>k}\Sigma^{2r+1}_{>k}V^{\top}_{>k}\widetilde{G}^{(j)}\right\Vert _{2}\left\Vert \left(K^{(j)}\right)^{-1}\right\Vert _{2}\right)\\
 & \qquad\qquad\cdot\left\Vert \left(I+F^{(j)}\right)^{-1}\right\Vert _{2}\left\Vert \Sigma^{-1}_{k}\right\Vert _{2}\\
 & \leq\left(\Phi_{r+1,j}\frac{1}{\sigma_{k}}+\Delta R\sigma^{2r}_{k+1}q_{\mathrm{row}}\frac{3}{2}\frac{1}{\sigma_{k}}\right)\frac{1}{1-p_{R}}\frac{1}{\sigma_{k}}\\
 & =\frac{\Phi_{r+1,j}}{\left(1-p_{R}\right)\sigma^{2}_{k}}+\frac{3\Delta Rq_{\mathrm{row}}\sigma^{2r}_{k+1}}{2\left(1-p_{R}\right)\sigma^{2}_{k}}
\end{align*}
In the last inequality, we used the following inequalities: $\left\Vert e^{\top}_{i}U_{>k}\Sigma^{2r+1}_{>k}V^{\top}_{>k}\widetilde{G}^{(j)}\right\Vert _{2}\leq\sigma^{2r}_{k+1}q_{\mathrm{row}}$
since $\gclip$ holds, $\left\Vert \left(K^{(j)}\right)^{-1}\right\Vert _{2}\leq\frac{3}{2}\frac{1}{\sigma_{k}}$
by Lemma \ref{lem:K-inverse}, and $\left\Vert \left(I+F^{(j)}\right)^{-1}\right\Vert _{2}\leq\frac{1}{1-p_{R}}$
which follows from Lemma \ref{lem:neumann-series} since $\left\Vert F^{(j)}\right\Vert _{2}\leq p_{R}<1$.

Unrolling the recurrence, we obtain
\begin{align*}
\Phi_{0,t} & \leq\frac{\Phi_{t,0}}{\sigma^{2t}_{k}\left(1-p_{R}\right)^{t}}+\frac{3\Delta Rq_{\mathrm{row}}}{2\sigma^{2}_{k}\left(1-p_{R}\right)}\sum^{t-1}_{j=0}\left(\frac{1-\kappa_{k}}{1-p_{R}}\right)^{j}\\
 & =\frac{\Phi_{t,0}}{\sigma^{2t}_{k}\left(1-p_{R}\right)^{t}}+\frac{3\Delta Rq_{\mathrm{row}}}{2\sigma^{2}_{k}\left(\kappa_{k}-p_{R}\right)}\left(1-\left(\frac{1-\kappa_{k}}{1-p_{R}}\right)^{t}\right)
\end{align*}
Since $\ginit$ holds, we have $\Phi_{t,0}\leq\omega_{\mathrm{init}}\sigma^{2t}_{k+1}$.
Using this bound together with the assumptions $\Delta q_{\mathrm{row}}\leq\frac{1}{8}\Delta_{k}$
and $p_{R}\leq\frac{\kappa_{k}}{16}$, we obtain
\begin{align*}
\Phi_{0,t} & \leq\frac{\omega_{\mathrm{init}}\sigma^{2t}_{k+1}}{\sigma^{2t}_{k}\left(1-p_{R}\right)^{t}}+\frac{3\Delta Rq_{\mathrm{row}}}{2\sigma^{2}_{k}\left(\kappa_{k}-p_{R}\right)}\left(1-\left(\frac{1-\kappa_{k}}{1-p_{R}}\right)^{t}\right)\\
 & =\omega_{\mathrm{init}}\left(\frac{1-\kappa_{k}}{1-p_{R}}\right)^{t}+\frac{3\Delta Rq_{\mathrm{row}}}{2\sigma^{2}_{k}\left(\kappa_{k}-p_{R}\right)}\left(1-\left(\frac{1-\kappa_{k}}{1-p_{R}}\right)^{t}\right)\\
 & \leq\omega_{\mathrm{init}}+\frac{3\Delta Rq_{\mathrm{row}}}{2\sigma^{2}_{k}\frac{15}{16}\kappa_{k}}=\omega_{\mathrm{init}}+\frac{8\Delta Rq_{\mathrm{row}}}{5\Delta_{k}}\leq\omega_{\mathrm{init}}+\frac{1}{5}R
\end{align*}
Combining with Lemma \ref{lem:row-length-bound} and using that $R\geq\max\left\{ 15\sqrt{k}\left\Vert U\right\Vert _{\infty},5\omega_{\mathrm{init}}\right\} $,
we obtain
\begin{align*}
\max_{i\in\left[n\right]}\left\Vert a^{\top}_{i}X^{(t)}B^{(t)}\right\Vert _{2} & \leq\sqrt{2}\left(\sqrt{k}\left\Vert U\right\Vert _{\infty}+\Phi_{0,t}\right)\\
 & \leq\sqrt{2}\left(\sqrt{k}\left\Vert U\right\Vert _{\infty}+\omega_{\mathrm{init}}+\frac{1}{5}R\right)\\
 & \leq\sqrt{2}\left(\frac{1}{15}R+\frac{1}{5}R+\frac{1}{5}R\right)\leq R
\end{align*}
and thus clipping is inactive in iteration $t$ as well.
\end{proof}

\paragraph{Convergence analysis.}

Next, using Lemmas \ref{lem:tangent-recurrence} and \ref{lem:clipping-inactive},
we analyze $s_{t}$ and $m_{t}$ for all iterations $t\geq0$. We
proceed by induction and show that the clipping is inactive and thus
we can apply the main recurrence shown in Lemma \ref{lem:tangent-recurrence}.
Recall that the recurrence states that $E^{(t+1)}=\left(\Sigma^{2}_{>k}E^{(t)}\Sigma^{-2}_{k}+L^{(t)}\right)\left(I+F^{(t)}\right)^{-1}$.
The term $\Sigma^{2}_{>k}E^{(t)}\Sigma^{-2}_{k}$ captures the main
progress, and we have $\left\Vert \Sigma^{2}_{>k}E^{(t)}\Sigma^{-2}_{k}\right\Vert _{2}\leq\frac{\sigma^{2}_{k+1}}{\sigma^{2}_{k}}\left\Vert E^{(t)}\right\Vert _{2}=\left(1-\kappa_{k}\right)m_{t}$.
Thus in an ideal case where $L^{(t)}$ and $F^{(t)}$ are zero, $m_{t}$
decreases geometrically in every iteration. In our setting, $L^{(t)}$
and $F^{(t)}$ are non-zero and they depend on the Gaussian noise
matrix $\widetilde{G}^{(t)}$, and our analysis will upper bound their
contribution. 

We start by showing that the initial solution satisfies $s_{0}\le\frac{1}{4}$
and thus $m_{0}\leq\frac{1}{\sqrt{15}}$, which will provide the base
case for the induction.
\begin{lem}
\label{lem:initial-error} Suppose that $\Delta_{k}\geq8\Delta_{\mathrm{init}}\left(2\sqrt{d}+\sqrt{2\log\frac{4}{\beta}}\right)$.
With probability at least $1-\frac{\beta}{4}$, $s_{0}\leq\frac{1}{4}$.
\end{lem}

\begin{proof}
Let $W=\Delta_{\mathrm{init}}\sym\left(G^{(\mathrm{init})}\right)$,
and thus the columns of the initial solution $X^{(0)}$ are the top-$k$
eigenvectors $X^{(0)}$ of $\widehat{M}=A^{\top}A+W$. By Lemma \ref{lem:gaussian-norm},
we have
\[
\Pr\left[\left\Vert \sym\left(G^{(\mathrm{init})}\right)\right\Vert _{2}>2\sqrt{d}+\sqrt{2\log\frac{4}{\beta}}\right]\leq\exp\left(-\frac{1}{2}\left(\sqrt{2\log\frac{4}{\beta}}\right)^{2}\right)=\frac{\beta}{4}.
\]
In the following, we condition on the event $\left\Vert \sym\left(G^{(\mathrm{init})}\right)\right\Vert _{2}\leq2\sqrt{d}+\sqrt{2\log\frac{4}{\beta}}$.
We have
\[
\left\Vert W\right\Vert _{2}=\Delta_{\mathrm{init}}\left\Vert \sym\left(G^{(\mathrm{init})}\right)\right\Vert _{2}\leq\Delta_{\mathrm{init}}\left(2\sqrt{d}+\sqrt{2\log\frac{4}{\beta}}\right)\le\frac{1}{8}\Delta_{k}.
\]
Thus Lemma \ref{lem:davis-kahan-error} gives
\[
s_{0}=\left\Vert \left(I-X^{(0)}X^{(0)\top}\right)V_{k}\right\Vert _{2}\le\frac{\left\Vert W\right\Vert _{2}}{\Delta_{k}-\left\Vert W\right\Vert _{2}}\leq\frac{2\left\Vert W\right\Vert _{2}}{\Delta_{k}}\leq\frac{1}{4}.
\]
\end{proof}

Before giving our inductive analysis, we first show the following
lemma which will be used to upper bound one of the error terms in
the main recurrence. The lemma follows from a martingale argument.
We first introduce some notation. We define a filtration $\left(\calF_{t}\right)_{0\leq t\leq T}$
as follows. For $t<T$, $\calF_{t}$ includes $G^{(\mathrm{init})}$,
$G^{(j)}$ and $\widetilde{G}^{(j)}$ for all iterations $0\leq j<t$
prior to $t$, and $G^{(t)}$; $\widetilde{G}^{(t)}$ is not included
in $\calF_{t}$ and is independent of $\calF_{t}$. We let $\calF_{T}$
include the full transcript.

Let $\tau$ be the following stopping time:
\begin{equation}
\tau:=\inf\left\{ 0\leq t<T\colon C^{(t)}\text{ is not invertible},s_{t}>\frac{1}{2},\text{ or }\left\Vert \Delta\sym\left(G^{(t)}\right)\right\Vert _{2}>\frac{\gamma}{4}\right\} \label{eq:martingale-stopping-time}
\end{equation}
with the convention that $\inf\emptyset=T$. 

We define
\[
W^{(t)}:=\begin{cases}
\left(K^{(t)}\right)^{-1} & t<\tau\\
0 & t\geq\tau
\end{cases}\qquad\widehat{L}^{(t)}:=\Delta RV^{\top}_{>k}\widetilde{G}^{(t)}W^{(t)}.
\]
Note that $\widehat{L}^{(t)}=L^{(t)}$ for $t<\tau$ and $\widehat{L}^{(t)}=0$
for $t\geq\tau$.

Finally, for a matrix $M$, we let
\[
\calT\left(M\right):=\Sigma^{2}_{>k}M\Sigma^{-2}_{k}.
\]

\begin{lem}
\label{lem:martingale} Suppose that $\gamma\leq\frac{\sigma^{2}_{k}}{4}$.
Let
\[
\nu_{R}:=\frac{3\Delta R}{\sigma_{k}}\sqrt{\frac{2d\log9+2\log\left(8T/\beta\right)}{\kappa_{k}\left(2-\kappa_{k}\right)}}.
\]
 With probability at least $1-\frac{\beta}{4}$, simultaneously for
all $1\leq t\leq T$, we have
\begin{equation}
\left\Vert \sum^{t-1}_{j=0}\mathcal{T}^{t-1-j}\left(\widehat{L}^{(j)}\right)\right\Vert _{2}\leq\nu_{R}.\label{eq:martingale-guarantee}
\end{equation}
Let $\mathcal{G}_{\mathrm{mart}}$ be this event. 
\end{lem}

\begin{proof}
Fix an iteration $1\leq t\leq T$ and two unit vectors $u\in\R^{d-k}$
and $v\in\R^{k}$. For every $0\leq j<t$, let
\[
\xi_{j}:=u^{\top}\mathcal{T}^{t-1-j}\left(\widehat{L}^{(j)}\right)v=u^{\top}\Sigma^{2(t-1-j)}_{>k}\widehat{L}^{(j)}\Sigma^{-2(t-1-j)}_{k}v.
\]
Recall that $\widehat{L}^{(j)}:=\Delta RV^{\top}_{>k}\widetilde{G}^{(j)}W^{(j)}$
and $\calF_{j}$ includes all random matrices in the iterations prior
to $j$ and $G^{(j)}$, but not $\widetilde{G}^{(j)}$. Thus, conditioned
on $\calF_{j}$, $\xi_{j}$ is a Gaussian random variable with expectation
$0$ and variance
\begin{align*}
\E\left[\xi^{2}_{j}\vert\calF_{j}\right] & =\left(\Delta R\right)^{2}\left\Vert \Sigma^{2(t-1-j)}_{>k}u\right\Vert ^{2}_{2}\left\Vert W^{(j)}\Sigma^{-2(t-1-j)}_{k}v\right\Vert ^{2}_{2}\\
 & \leq\left(\Delta R\right)^{2}\left\Vert \Sigma^{2(t-1-j)}_{>k}\right\Vert ^{2}_{2}\left\Vert u\right\Vert ^{2}_{2}\left\Vert W^{(j)}\right\Vert ^{2}_{2}\left\Vert \Sigma^{-2(t-1-j)}_{k}\right\Vert ^{2}_{2}\left\Vert v\right\Vert ^{2}_{2}\\
 & =\left(\Delta R\right)^{2}\left(\frac{\sigma^{2}_{k+1}}{\sigma^{2}_{k}}\right)^{2(t-1-j)}\left\Vert W^{(j)}\right\Vert ^{2}_{2}\\
 & \leq\left(\Delta R\right)^{2}\left(\frac{\sigma^{2}_{k+1}}{\sigma^{2}_{k}}\right)^{2(t-1-j)}\left(\frac{3}{2\sigma_{k}}\right)^{2}\\
 & =\left(\frac{3\Delta R}{2\sigma_{k}}\right)^{2}\left(1-\kappa_{k}\right)^{2(t-1-j)}.
\end{align*}
The last inequality follows from Lemma \ref{lem:K-inverse}: if $j<\tau$,
we have $W^{(j)}=\left(K^{(j)}\right)^{-1}$ and the conditions of
Lemma \ref{lem:K-inverse} are satisfied; otherwise, $W^{(j)}=0$.

We can use the subgaussian version of Azuma's inequality to bound
$\sum_{j}\xi_{j}$ (see e.g. \citet[Lemma F.4]{DuchiBartlettWainwright2012}
or \citet[Theorem 1(b)]{HowardRamdasMcAuliffeSekhon2020} with $\psi(\lambda)=\lambda^{2}/2$).
For every $x\geq0$, we have
\begin{align*}
\Pr\left[\left|\sum^{t-1}_{j=0}\xi_{j}\right|\geq x\right] & \leq2\exp\left(-\frac{x^{2}}{2\sum^{t-1}_{j=0}\left(\frac{3\Delta R}{2\sigma_{k}}\right)^{2}\left(1-\kappa_{k}\right)^{2(t-1-j)}}\right)\\
 & \leq2\exp\left(-\frac{x^{2}\kappa_{k}\left(2-\kappa_{k}\right)}{2\left(\frac{3\Delta R}{2\sigma_{k}}\right)^{2}}\right).
\end{align*}
We can use the above guarantee to upper bound $\left\Vert \sum^{t-1}_{j=0}\mathcal{T}^{t-1-j}\left(\widehat{L}^{(j)}\right)\right\Vert _{2}$
using a net argument, similarly to the proof of Claim \ref{lem:clipping-inactive} and \ref{claim:norm-concentration-initial}.
Let $\calN_{1}$ and $\calN_{2}$ be $\frac{1}{4}$-nets of the two
unit spheres in $\R^{d-k}$ and $\R^{k}$ respectively, of sizes $\left|\calN_{1}\right|\leq9^{d-k}$
and $\left|\calN_{2}\right|\leq9^{k}$. For any $(d-k)\times k$ matrix
$M$, we have $\left\Vert M\right\Vert _{2}\le2\max_{u\in\calN_{1},v\in\calN_{2}}\left|u^{\top}Mv\right|$.
Indeed, let $u^{*},v^{*}$ be unit left and right top singular vectors
of $M$ and leet $u,v$ the closest vectors in the nets to $u^{*},v^{*}$.
We have
\[
\left\Vert M\right\Vert _{2}=u^{*\top}Mv^{*}=u^{*\top}M(v^{*}-v)+(u^{*\top}-u^{\top})Mv+u^{\top}Mv\le\frac{1}{4}\left\Vert M\right\Vert _{2}+\frac{1}{4}\left\Vert M\right\Vert _{2}+u^{\top}Mv.
\]
By rearranging, we obtain the claimed inequality.

There are $\left|\calN_{1}\right|\left|\calN_{2}\right|T=9^{d}T$
triples $(u,v,t)$ with $u\in\calN_{1}$, $v\in\calN_{2}$, and $1\leq t\leq T$.
Setting $x=\frac{3\Delta R}{2\sigma_{k}}\sqrt{\frac{2\log\left(\frac{8\cdot9^{d}T}{\beta}\right)}{\kappa_{k}\left(2-\kappa_{k}\right)}}$
in the guarantee above and taking an union bound over all $9^{d}T$
triples, we obtain the following guarantee. With probability at least
$1-\frac{\beta}{4},$ simultaneously for all $1\leq t\leq T$, 
\begin{align*}
\left\Vert \sum^{t-1}_{j=0}\mathcal{T}^{t-1-j}\left(\widehat{L}^{(j)}\right)\right\Vert _{2} & \leq2\max_{u\in\calN_{1},v\in\calN_{2}}\left|u^{\top}Mv\right|\\
 & \leq\frac{3\Delta R}{\sigma_{k}}\sqrt{\frac{2\log\left(\frac{8\cdot9^{d}T}{\beta}\right)}{\kappa_{k}\left(2-\kappa_{k}\right)}}=\nu_{R},
\end{align*}
as needed.
\end{proof}

We now give the main inductive argument. Recall that $\gclip$ is
the event defined in Definition \ref{def:gclip}, and $\ginit$ and
$\gmart$ are the events defined in Lemma \ref{lem:initial-potential}
and Lemma \ref{lem:martingale} respectively. Recall also that $p_{R}:=\frac{3\Delta Rq_{k}}{2\sigma_{k}}$
(see Lemma \ref{lem:clipping-inactive}) and $\nu_{R}:=\frac{3\Delta R}{\sigma_{k}}\sqrt{\frac{2d\log9+2\log\left(8T/\beta\right)}{\kappa_{k}\left(2-\kappa_{k}\right)}}$
(see Lemma \ref{lem:martingale}). 
\begin{thm}
\label{thm:convergence-martingale} Suppose that we have $\Delta_{k}\geq\max\left\{ 8\Delta_{\mathrm{init}}\left(2\sqrt{d}+\sqrt{2\log\frac{8}{\beta}}\right),8\Delta q_{\mathrm{row}}\right\} $,
$\gamma\leq\frac{\sigma^{2}_{k}}{4},$ $p_{R}\leq\frac{\kappa_{k}}{16}$,
$R\geq15\left\Vert U\right\Vert _{\infty}\left(\sqrt{k}+\sqrt{\log\left(\frac{16n\left(T+1\right)}{\beta}\right)}\right)$,
and $\nu_{R}\leq\frac{1}{8}$. Suppose that $\mathcal{G}_{\mathrm{clip}}$,
$\ginit$, $\mathcal{G}_{\mathrm{mart}}$, and the event in Lemma
\ref{lem:initial-potential} and \ref{lem:initial-error} hold. Then,
in every iteration, $C^{(t)}$ is invertible, clipping is inactive,
and
\begin{equation}
s_{t}\leq m_{t}\leq\left(\frac{1-\kappa_{k}}{1-p_{R}}\right)^{t}m_{0}+2\nu_{R}\leq\left(1-\frac{\kappa_{k}}{2}\right)^{t}m_{0}+2\nu_{R}.\label{eq:convergence}
\end{equation}
\end{thm}

\begin{proof}
In the following, we will use the definition of the stopping time
$\tau$ (eqn. (\ref{eq:martingale-stopping-time})) and the identity
$s_{t}=\frac{m_{t}}{\sqrt{1+m^{2}_{t}}}$ (Lemma \ref{lem:s-to-m}).
We will prove by induction on $t\geq0$ that $C^{(j)}$ is invertible
and $m_{j}\leq\frac{1}{\sqrt{3}}$ (and thus $s_{j}\leq\frac{1}{2}$)
for all $j\leq t$, clipping is inactive in all iterations $j<t$,
$\tau>t$ (if $t=T$, $\tau=T$), and (\ref{eq:convergence}) holds
for all iterations $j\leq t$. Consider the base case $t=0$. Since
the event in Lemma \ref{lem:initial-potential} and \ref{lem:initial-error}
holds, we have $s_{0}\leq\frac{1}{4}$ and thus $m_{0}\le\frac{1}{\sqrt{15}}$.
Since $s_{0}<1$, Lemma \ref{lem:s-to-C-invertible} implies that
$C^{(0)}$ is invertible. Additionally, $\tau>0$ and (\ref{eq:convergence})
holds for $t=0$.

Next, we consider the inductive step. Consider $t\geq0$ and suppose
that $C^{(j)}$ is invertible and $m_{j}\leq\frac{1}{\sqrt{3}}$ (and
thus $s_{j}\leq\frac{1}{2}$) for all $j\leq t$, and clipping is
inactive in all iterations $j<t$, and $\tau>t$. These assumptions
and the assumptions in the theorem statement imply that Lemma \ref{lem:clipping-inactive}
applies and thus clipping is inactive in iteration $t$ as well.

For each $j\leq t$, we can show that $F^{(j)}:=\Delta RV^{\top}_{k}\widetilde{G}^{(j)}\left(K^{(j)}\right)^{-1}$
satisfies $\left\Vert F^{(j)}\right\Vert _{2}\leq p_{R}$ as follows.
We have
\begin{align*}
\left\Vert F^{(j)}\right\Vert _{2} & =\left\Vert \Delta RV^{\top}_{k}\widetilde{G}^{(j)}\left(K^{(j)}\right)^{-1}\right\Vert _{2}\\
 & \leq\Delta R\left\Vert V^{\top}_{k}\widetilde{G}^{(j)}\right\Vert _{2}\left\Vert \left(K^{(j)}\right)^{-1}\right\Vert _{2}\\
 & \leq\Delta Rq_{k}\frac{3}{2}\frac{1}{\sigma_{k}}=p_{R}.
\end{align*}
In the last inequality, we used that $\gclip$ holds (and thus $\left\Vert V^{\top}_{k}\widetilde{G}^{(j)}\right\Vert _{2}\leq q_{k}$)
and Lemma \ref{lem:K-inverse}. Note that the conditions of Lemma
\ref{lem:K-inverse} hold: we have $\left\Vert \Delta\sym\left(G^{(j)}\right)\right\Vert _{2}\leq\frac{\gamma}{4}$
since $\gclip$ holds, $\gamma\leq\frac{\sigma^{2}_{k}}{4}$, and
$s_{j}\leq\frac{1}{2}$ for all $j\leq t$.

For all $0\leq j\leq t$, we have $\left\Vert F^{(j)}\right\Vert _{2}\leq p_{R}\leq\frac{\kappa_{k}}{16}<1$,
$C^{(j)}$ is invertible, and clipping is inactive in iteration $j$.
Thus Lemma \ref{lem:tangent-recurrence} gives that $C^{(t+1)}$ is
invertible and the following recurrence holds for all $0\leq j\leq t$:
\begin{align}
E^{(j+1)} & =\left(\Sigma^{2}_{>k}E^{(j)}\Sigma^{-2}_{k}+L^{(j)}\right)\left(I+F^{(j)}\right)^{-1}=\left(\calT\left(E^{(j)}\right)+L^{(j)}\right)\left(I+F^{(j)}\right)^{-1}\label{eq:tangent-recurrence}
\end{align}
where we used the notation $\calT\left(M\right)=\Sigma^{2}_{>k}M\Sigma^{-2}_{k}$
defined earlier. Using the above recurrence, we obtain the following.
\begin{claim}
Let
\[
b_{R}:=\Delta Rq_{d-k}\frac{3}{2\sigma_{k}}\quad\eta_{R}:=\nu_{R}+\frac{b_{R}p_{R}}{\kappa_{k}\left(1-p_{R}\right)}\quad h_{R}:=\frac{\left(1-\kappa_{k}\right)p_{R}}{1-p_{R}}.
\]
For every $1\leq\ell\leq t+1$, we have
\begin{equation}
m_{\ell}\leq\left(1-\kappa_{k}\right)^{\ell}m_{0}+\eta_{R}+h_{R}\sum^{\ell-1}_{j=0}\left(1-\kappa_{k}\right)^{\ell-1-j}m_{j}.\label{eq:recurrence-volterra}
\end{equation}
\end{claim}

\begin{proof}
Using the identity $\left(I+F\right)^{-1}=I-F\left(I+F\right)^{-1}$,
we can write the recurrence (\ref{eq:tangent-recurrence}) as follows:
\begin{equation}
E^{(j+1)}=\calT\left(E^{(j)}\right)+L^{(j)}\underbrace{-\left(\calT\left(E^{(j)}\right)+L^{(j)}\right)F^{(j)}\left(I+F^{(j)}\right)^{-1}}_{:=R^{(j)}}=\calT\left(E^{(j)}\right)+L^{(j)}+R^{(j)}\label{eq:tangent-recurrence-separated}
\end{equation}
We rearranged the recurrence in order to separate the noise terms
depending on $\widetilde{G}^{(j)}$ into the term $L^{(j)}$ that
is linear in $\widetilde{G}^{(j)}$ and a term $R^{(j)}$ which has
a more complex dependency on $\widetilde{G}^{(j)}$. This separation
allows us to bound the contribution of the linear terms $L^{(j)}$
via the martingale argument in Lemma \ref{lem:martingale}. The terms
$R^{(j)}$ are more complex, and we will bound their contribution
by bounding the norm of each $R^{(j)}$ separately.

For each $1\leq\ell\leq t+1$, by unrolling the recurrence (\ref{eq:tangent-recurrence-separated})
and using that $L^{(j)}=\widehat{L}^{(j)}$ since $\tau>t$, we obtain
\begin{align*}
E^{(\ell)} & =\mathcal{T}^{\ell}\left(E^{(0)}\right)+\sum^{\ell-1}_{j=0}\mathcal{T}^{\ell-1-j}\left(L^{(j)}\right)+\sum^{\ell-1}_{j=0}\mathcal{T}^{\ell-1-j}\left(R^{(j)}\right)\\
 & =\mathcal{T}^{\ell}\left(E^{(0)}\right)+\sum^{\ell-1}_{j=0}\mathcal{T}^{\ell-1-j}\left(\widehat{L}^{(j)}\right)+\sum^{\ell-1}_{j=0}\mathcal{T}^{\ell-1-j}\left(R^{(j)}\right).
\end{align*}
Thus, for each $1\leq\ell\leq t+1$, we have
\[
m_{\ell}=\left\Vert E^{(\ell)}\right\Vert _{2}\leq\left\Vert \mathcal{T}^{\ell}\left(E^{(0)}\right)\right\Vert _{2}+\left\Vert \sum^{\ell-1}_{j=0}\mathcal{T}^{\ell-1-j}\left(\widehat{L}^{(j)}\right)\right\Vert _{2}+\sum^{\ell-1}_{j=0}\left\Vert \mathcal{T}^{\ell-1-j}\left(R^{(j)}\right)\right\Vert _{2}.
\]
We have
\begin{align*}
\left\Vert \mathcal{T}^{\ell}\left(E^{(0)}\right)\right\Vert _{2} & =\left\Vert \Sigma^{2\ell}_{>k}E^{(0)}\Sigma^{-2\ell}_{k}\right\Vert _{2}\\
 & \le\left\Vert \Sigma^{2\ell}_{>k}\right\Vert _{2}\left\Vert E^{(0)}\right\Vert _{2}\left\Vert \Sigma^{-2\ell}_{k}\right\Vert _{2}\\
 & =\frac{\sigma^{2\ell}_{k+1}}{\sigma^{2\ell}_{k}}m_{0}=\left(1-\kappa_{k}\right)^{\ell}m_{0}.
\end{align*}
Since $\gmart$ holds, for all $1\leq\ell\leq t+1$, we have
\[
\left\Vert \sum^{\ell-1}_{j=0}\mathcal{T}^{\ell-1-j}\left(\widehat{L}^{(j)}\right)\right\Vert _{2}\leq\nu_{R}.
\]
We have
\begin{align*}
\left\Vert R^{(j)}\right\Vert _{2} & =\left\Vert -\left(\calT\left(E^{(j)}\right)+L^{(j)}\right)F^{(j)}\left(I+F^{(j)}\right)^{-1}\right\Vert _{2}\\
 & \leq\left(\left\Vert \calT\left(E^{(j)}\right)\right\Vert _{2}+\left\Vert L^{(j)}\right\Vert _{2}\right)\left\Vert F^{(j)}\right\Vert _{2}\left\Vert \left(I+F^{(j)}\right)^{-1}\right\Vert _{2}
\end{align*}
Additionally,
\begin{align*}
\left\Vert \calT\left(E^{(j)}\right)\right\Vert _{2} & =\left\Vert \Sigma^{2}_{>k}E^{(j)}\Sigma^{-2}_{k}\right\Vert _{2}\leq\left\Vert \Sigma^{2}_{>k}\right\Vert _{2}\left\Vert E^{(j)}\right\Vert _{2}\left\Vert \Sigma^{-2}_{k}\right\Vert _{2}\\
 & =\frac{\sigma^{2}_{k+1}}{\sigma^{2}_{k}}m_{j}=\left(1-\kappa_{k}\right)m_{j}.
\end{align*}
Using that $\left\Vert V^{\top}_{>k}\widetilde{G}^{(t)}\right\Vert _{2}\leq q_{d-k}$
(since $\gclip$ holds) and Lemma \ref{lem:K-inverse}, we obtain
\begin{align*}
\left\Vert L^{(j)}\right\Vert _{2} & =\Delta R\left\Vert V^{\top}_{>k}\widetilde{G}^{(j)}\left(K^{(j)}\right)^{-1}\right\Vert _{2}\\
 & \leq\Delta R\left\Vert V^{\top}_{>k}\widetilde{G}^{(j)}\right\Vert _{2}\left\Vert \left(K^{(j)}\right)^{-1}\right\Vert _{2}\\
 & \leq\underbrace{\Delta Rq_{d-k}\frac{3}{2\sigma_{k}}}_{:=b_{R}}
\end{align*}
As shown earlier, we have $\left\Vert F^{(j)}\right\Vert _{2}\leq p_{R}$.
Since $\left\Vert F^{(j)}\right\Vert _{2}\leq p_{R}\leq\frac{\kappa_{k}}{16}<1$,
Lemma \ref{lem:neumann-series} gives $\left\Vert \left(I+F^{(j)}\right)^{-1}\right\Vert _{2}\leq\frac{1}{1-p_{R}}$.

Putting everything together, we obtain
\[
\left\Vert R^{(j)}\right\Vert _{2}\leq\left(\left(1-\kappa_{k}\right)m_{j}+b_{R}\right)\frac{p_{R}}{1-p_{R}}=:h_{R}m_{j}+f_{R},
\]
where we let $h_{R}:=$$\frac{\left(1-\kappa_{k}\right)p_{R}}{1-p_{R}}$
and $f_{R}:=\frac{b_{R}p_{R}}{1-p_{R}}$.

Thus
\begin{align*}
\left\Vert \mathcal{T}^{\ell-1-j}\left(R^{(j)}\right)\right\Vert _{2} & =\left\Vert \Sigma^{2(\ell-1-j)}_{>k}R^{(j)}\Sigma^{-2(\ell-1-j)}_{k}\right\Vert _{2}\\
 & \leq\left\Vert \Sigma^{2(\ell-1-j)}_{>k}\right\Vert _{2}\left\Vert R^{(j)}\right\Vert _{2}\left\Vert \Sigma^{-2(\ell-1-j)}_{k}\right\Vert _{2}\\
 & \leq\frac{\sigma^{2(\ell-1-j)}_{k+1}}{\sigma^{2(\ell-1-j)}_{k}}\left\Vert R^{(j)}\right\Vert _{2}\\
 & =\left(1-\kappa_{k}\right)^{\ell-1-j}\left\Vert R^{(j)}\right\Vert _{2}\\
 & \leq\left(1-\kappa_{k}\right)^{\ell-1-j}\left(h_{R}m_{j}+f_{R}\right)
\end{align*}
It follows that
\begin{align*}
m_{\ell} & \leq\left(1-\kappa_{k}\right)^{\ell}m_{0}+\nu_{R}+\sum^{\ell-1}_{j=0}\left(1-\kappa_{k}\right)^{\ell-1-j}\left(h_{R}m_{j}+f_{R}\right)\\
 & \leq\left(1-\kappa_{k}\right)^{\ell}m_{0}+\underbrace{\nu_{R}+\frac{f_{R}}{\kappa_{k}}}_{:=\eta_{R}}+h_{R}\sum^{\ell-1}_{j=0}\left(1-\kappa_{k}\right)^{\ell-1-j}m_{j}
\end{align*}
as needed.
\end{proof}

We can prove by induction that \ref{eq:recurrence-volterra} implies
the following.
\begin{claim}
For every $1\leq\ell\leq t+1$, we have
\begin{equation}
m_{\ell}\leq\left(\frac{1-\kappa_{k}}{1-p_{R}}\right)^{\ell}m_{0}+\frac{\kappa_{k}}{\kappa_{k}-h_{R}}\eta_{R}.\label{eq:recurrence-volterra-solution}
\end{equation}

\begin{proof}
We prove by induction on $\ell$ that
\[
m_{\ell}\leq\left(1-\kappa_{k}+h_{R}\right)^{\ell}m_{0}+\eta_{R}\left(1+h_{R}\sum^{\ell-1}_{j=0}\left(1-\kappa_{k}+h_{R}\right)^{j}\right).
\]
The base case $\ell=1$ follows from (\ref{eq:recurrence-volterra}).
In the inductive step, we have
\begin{align*}
m_{\ell} & \leq\left(1-\kappa_{k}\right)^{\ell}m_{0}+\eta_{R}+h_{R}\sum^{\ell-1}_{j=0}\left(1-\kappa_{k}\right)^{\ell-1-j}m_{j}\\
 & \le\left(1-\kappa_{k}\right)^{\ell}m_{0}+\eta_{R}+h_{R}\sum^{\ell-1}_{j=0}\left(1-\kappa_{k}\right)^{\ell-1-j}\left((1-\kappa_{k}+h_{R})^{j}m_{0}+\eta_{R}\left(1+h_{R}\sum^{j-1}_{r=0}(1-\kappa_{k}+h_{R})^{r}\right)\right)
\end{align*}
First, let us consider the coefficient of $m_{0}$. Since we have
$a^{\ell}=b^{\ell}+\left(a-b\right)\left(\sum^{\ell-1}_{j=0}a^{j}b^{\ell-1-j}\right)$
for all $a,b$, we can add the coefficients of $m_{0}$ and obtain
\[
\left(1-\kappa_{k}\right)^{\ell}+h_{R}\sum^{\ell-1}_{j=0}\left(1-\kappa_{k}\right)^{\ell-1-j}\left(1-\kappa_{k}+h_{R}\right)^{j}=\left(1-\kappa_{k}+h_{R}\right)^{\ell}.
\]
Next, we consider the coefficient of $\eta_{R}$. We have

\begin{align*}
 & 1+h_{R}\sum^{\ell-1}_{j=0}\left(1-\kappa_{k}\right)^{\ell-1-j}\left(1+h_{R}\sum^{j-1}_{r=0}\left(1-\kappa_{k}+h_{R}\right)^{r}\right)\\
 & =1+h_{R}\sum^{\ell-1}_{i=0}\left(\left(1-\kappa_{k}\right)^{i}+h_{R}\sum^{i-1}_{j=0}\left(1-\kappa_{k}\right)^{j}\left(1-\kappa_{k}+h_{R}\right)^{i-1-j}\right)\\
 & =1+h_{R}\sum^{\ell-1}_{i=0}\left(1-\kappa_{k}+h_{R}\right)^{i}
\end{align*}
which completes the induction.

Finally, using that $h_{R}:=\frac{\left(1-\kappa_{k}\right)p_{R}}{1-p_{R}}$
and the geometric sum formula, we obtain the claim.
\end{proof}

\end{claim}

By (\ref{eq:recurrence-volterra-solution}), for every $1\leq\ell\leq t+1$,
we have
\begin{align*}
m_{\ell} & \leq\left(\frac{1-\kappa_{k}}{1-p_{R}}\right)^{\ell}m_{0}+\frac{\kappa_{k}}{\kappa_{k}-h_{R}}\eta_{R}.
\end{align*}
Note that $b_{R}\leq\nu_{R}$, since $q_{d-k}\leq2\sqrt{d+\log(32T/\beta)}$.
The assumption $p_{R}\le\frac{\kappa_{k}}{16}$ implies $h_{R}\le p_{R}\le\frac{\kappa_{k}}{16}$,
and $\frac{p_{R}}{\kappa_{k}}\frac{b_{R}}{1-p_{R}}\le\frac{b_{R}}{15}\le\frac{\nu_{R}}{15}$
.
\[
\eta_{R}:=\nu_{R}+\frac{b_{R}p_{R}}{\kappa_{k}\left(1-p_{R}\right)}\leq\frac{16}{15}\nu_{R}
\]
and
\[
\frac{\kappa_{k}}{\kappa_{k}-h_{R}}\eta_{R}\leq\left(\frac{16}{15}\right)^{2}\nu_{R}.
\]
Thus we obtain
\[
m_{\ell}\leq\left(\frac{1-\kappa_{k}}{1-p_{R}}\right)^{\ell}m_{0}+\left(\frac{16}{15}\right)^{2}\nu_{R}\leq\left(1-\frac{\kappa_{k}}{2}\right)^{\ell}m_{0}+\left(\frac{16}{15}\right)^{2}\nu_{R},
\]
which proves that (\ref{eq:convergence}) holds for all iterations
$\ell\leq t+1$.

In particular, we have
\[
m_{t+1}\leq m_{0}+2\nu_{R}\le\frac{1}{\sqrt{15}}+\frac{1}{4}<\frac{1}{\sqrt{3}}.
\]
This concludes the main induction, including showing that $\tau>t+1$
(if $t+1=T$, $\tau=T$).

\end{proof}

Finally, we show the following corollary of the main theorem for ideal
choices of $T$ and $R$ that depend on the data. In Section \ref{sec:private-radius-selection},
we show how to privately obtain an analogous guarantee via a private
hyperparameter search.
\begin{cor}
\label{cor:martingale-convergence-corollary}Suppose that $R\geq15\left\Vert U\right\Vert _{\infty}\left(\sqrt{k}+\sqrt{\log\left(\frac{16n\left(T+1\right)}{\beta}\right)}\right)$,
$T=\left\lceil \frac{2}{\kappa_{k}}\log\left(\frac{n}{\beta\delta\epsilon}\right)\right\rceil $,
and
\[
\Delta_{k}\geq\widetilde{\Omega}\left(\frac{\sqrt{d}}{\epsilon}+\left(\frac{\sigma_{k}\sqrt{k}}{\epsilon}\right)^{2/3}+\left(\frac{R\sigma^{2}_{k}\sqrt{k}}{\epsilon}\right)^{2/3}\right),
\]
where the $\widetilde{\Omega}\left(\cdot\right)$ hides sufficiently
large poly-logarithmic factors. With probability at least $1-\beta$,
the solution $X^{(T)}$ returned by Algorithm \ref{alg:Algorithm-main}
satisfies
\[
\left\Vert \left(I-X^{(T)}X^{(T)\top}\right)V_{k}\right\Vert ^{2}_{2}\leq\left(\frac{\beta\delta\epsilon}{n}\right)^{2}+\widetilde{O}\left(\frac{d\sigma^{2}_{k}R^{2}}{\epsilon^{2}\Delta^{2}_{k}}\right).
\]
In particular, if $R=\widetilde{\Theta}\left(\sqrt{k}\left\Vert U\right\Vert _{\infty}\right)$,
we have
\[
\left\Vert \left(I-X^{(T)}X^{(T)\top}\right)V_{k}\right\Vert ^{2}_{2}\leq\left(\frac{\beta\delta\epsilon}{n}\right)^{2}+\widetilde{O}\left(\frac{dk\sigma^{2}_{k}\|U\|^{2}_{\infty}}{\epsilon^{2}\Delta^{2}_{k}}\right).
\]
\end{cor}

\begin{proof}
Let us verify that the assumption imply the conditions in Theorem
\ref{thm:convergence-martingale}. The choice of $T$ gives
\[
\Delta=\widetilde{O}\left(\frac{1}{\epsilon\sqrt{\kappa_{k}}}\right)=\widetilde{O}\left(\frac{\sigma_{k}}{\epsilon\sqrt{\Delta_{k}}}\right).
\]
The gap condition $\Delta_{k}\geq\widetilde{\Omega}\left(\left(\frac{\sigma_{k}\sqrt{k}}{\epsilon}\right)^{2/3}\right)$
implies $\Delta q_{\mathrm{row}}=\widetilde{O}\left(\frac{\sigma_{k}\sqrt{k}}{\epsilon\sqrt{\Delta_{k}}}\right)=O\left(\Delta_{k}\right)$
and also implies $\gamma=\widetilde{O}\left(\Delta\sqrt{k}\right)=O\left(\Delta_{k}\right)<\frac{\sigma^{2}_{k}}{4}$.

The gap condition $\Delta_{k}\geq\widetilde{\Omega}\left(\left(\frac{R\sigma^{2}_{k}\sqrt{k}}{\epsilon}\right)^{2/3}\right)$
implies
\[
p_{R}=\frac{3\Delta Rq_{k}}{2\sigma_{k}}=\widetilde{O}\left(\frac{R\sqrt{k}}{\epsilon\sqrt{\Delta_{k}}}\right)\leq\frac{\Delta_{k}}{16\sigma^{2}_{k}}=\frac{\kappa_{k}}{16}.
\]
The gap condition $\Delta_{k}\ge\widetilde{\Omega}\left(\frac{\sqrt{d}}{\epsilon}\right)$
ensures that $\Delta_{k}\geq8\Delta_{\mathrm{init}}\left(2\sqrt{d}+\sqrt{2\log\frac{8}{\beta}}\right)$.

By Lemmas \ref{lem:clip-good-event}, \ref{lem:initial-potential},
\ref{lem:martingale}, and \ref{lem:initial-error}, the events $\gclip$,
$\ginit$, $\gmart$, and the event $s_{0}\leq\frac{1}{4}$ hold jointly
with probability at least $1-\beta$.

If $\nu_{R}\leq\frac{1}{8}$, all of the conditions of Theorem \ref{thm:convergence-martingale}
are satisfied and the theorem implies the desired guarantee. Otherwise,
we have $s^{2}_{T}\leq1\leq64\nu^{2}_{R}$, and we have the same asymptotic
guarantee.

\end{proof}

\section{Private selection of the clipping radius $R$ and iteration count
$T$}

\label{sec:private-radius-selection}

Both the clipping radius $R$ and the iteration count $T$ in Corollary
\ref{cor:martingale-convergence-corollary} depend on the private
data. We select them from public powers-of-two grids using the private
hyperparameter-selection method of \citet{PapernotSteinke2022}. We
use the same initialization and power iterations as before, with noise
scales adjusted to account for the search. The privacy proof uses
the sensitivity bounds from Section \ref{sec:alg-privacy}, and the
utility proof follows Corollary \ref{cor:martingale-convergence-corollary}
for a suitable grid pair, then accounts for sampling this pair and
comparing the candidate subspaces privately.

Fix $\epsilon,\delta,\beta\in(0,1)$ and let $\beta_{0}=\beta/2$.
We use the public grids 
\begin{align}
\mathcal{R} & =\{R_{j}=2^{j}R_{\min}:0\leq j\leq J\}, & m_{R} & =J+1,\nonumber \\
\mathcal{T}_{\mathrm{grid}} & =\{T_{\ell}=2^{\ell}T_{\min}:0\leq\ell\leq H\}, & m_{T} & =H+1,\nonumber \\
T_{\max} & =2^{H}T_{\min},\quad R_{\max}=2^{J}R_{\min}, & m & =m_{R}m_{T}.\label{eq:radius-grid}
\end{align}
Set $T_{\min}=1$ and choose $H$ to be the smallest nonnegative integer
such that 
\[
T_{\max}\geq\max\left\{ 1,\left\lceil \frac{2n\epsilon}{\sqrt{d}}\log\left(\frac{n}{\beta\delta\epsilon}\right)\right\rceil \right\} .
\]
Set $R_{\min}=\sqrt{k/n}$ and choose $J$ to be the smallest nonnegative
integer such that 
\[
R_{\max}\geq15\left(\sqrt{k}+\sqrt{\log\left(\frac{16n(T_{\max}+1)}{\beta_{0}}\right)}\right).
\]
These choices depend only on public parameters. To see why the grids
cover the choices needed for utility, note that the calibration below
satisfies $\Delta_{\mathrm{init}}\geq1/\epsilon$. The initialization
gap condition gives $\Delta_{k}\geq16\Delta_{\mathrm{init}}\sqrt{d}$,
whereas $\sigma^{2}_{k}\leq\tr(A^{\top}A)\leq n$. Thus $\kappa_{k}\geq\sqrt{d}/(n\epsilon)$
whenever the utility hypotheses hold. The radius range follows from
$1/\sqrt{n}\leq\left\Vert U\right\Vert _{\infty}\leq1$.

We run Algorithm \ref{alg:Algorithm-initialize} once, with noise
scale $\Delta_{\mathrm{init}}$, and use its output $X^{(0)}$ for
every run of Algorithm \ref{alg:Algorithm-power-method}. Each run
uses failure parameter $\beta_{0}$. For public noise parameters $\Delta_{\mathrm{init}},\chi,\zeta>0$,
a run with $t$ iterations uses 
\begin{equation}
\Delta(t)=\chi\sqrt{t},\qquad\gamma(t)=4\Delta(t)\left(2\sqrt{k}+\sqrt{2\log(8t/\beta_{0})}\right).\label{eq:joint-duration-noise}
\end{equation}
The privacy calibration is given in Theorem \ref{thm:radius-privacy}.
For each run, we use exactly the definitions of $q_{k}$, $q_{\mathrm{row}}$,
$p_{R}$, $\nu_{R}$, and the good events from Section \ref{sec:alg-utility},
with $T=t$, $\Delta=\Delta(t)$, and failure parameter $\beta_{0}$.

\subsection{A private score for comparing candidate subspaces}

In this section, we describe the score to compare the candidate subspaces.
Ideally we would like to use $\|(I-XX^{\top})V_{k}V^{\top}_{k}\|_{2}$
as the score but a more robust version of this is required to reduce
the sensitivity of the score when the inputs do not have a large eigenvalue
gap.

Recall that the initialization algorithm (Algorithm \ref{alg:Algorithm-initialize})
produces the matrix 
\[
\widehat{M}=A^{\top}A+\Delta_{\mathrm{init}}\sym(G^{(\mathrm{init})}),\qquad e_{\mathrm{w}}:=\Delta_{\mathrm{init}}\left(2\sqrt{d}+\sqrt{2\log(4/\beta_{0})}\right).
\]
Write its eigenvalues as $\widehat{\lambda}_{1}\geq\cdots\geq\widehat{\lambda}_{d}$
and let $X^{(0)}$ be its top-$k$ eigenbasis. Set 
\begin{equation}
a=\widehat{\lambda}_{k+1}+e_{\mathrm{w}},\qquad b=\widehat{\lambda}_{k}-e_{\mathrm{w}},\qquad h=b-a.\label{eq:radius-pilot-band}
\end{equation}
If $h\leq0$, return $(R_{J},0,X^{(0)})$ immediately. Intuitively
this happens when the eigenvalue gap is too small compared with the
noise. If $h>0$, consider the following scalar function 
\[
f_{a,b}(x)=\begin{cases}
0, & x\leq a,\\
(x-a)/h, & a<x<b,\\
1, & x\geq b,
\end{cases}
\]
For a symmetric matrix $M=Q\diag(\lambda_{1},\ldots,\lambda_{d})Q^{\top}$,
define $f_{a,b}(M)=Q\diag(f_{a,b}(\lambda_{1}),\ldots,f_{a,b}(\lambda_{d}))Q^{\top}$.
For a candidate solution $X\in\mathbb{R}^{d\times k}$ with orthonormal
columns, its score is 
\begin{equation}
\mathsf{s}_{a,b}(A,X):=\|(I-XX^{\top})f_{a,b}(A^{\top}A)\|_{2}\in[0,1].\label{eq:radius-validation-score}
\end{equation}
The following lemma bounds the score sensitivity on all inputs and
shows that, under the same initialization gap condition as before,
the score equals the subspace error. 
\begin{lem}[Sensitivity and accuracy of the comparison score]
\label{lem:radius-score} For each fixed initialization with $h>0$,
and every fixed orthonormal $X$, the replacement sensitivity of $\mathsf{s}_{a,b}(A,X)$
is at most $2/h$. If $\|\widehat{M}-A^{\top}A\|_{2}\leq e_{\mathrm{w}},\Delta_{k}\geq8e_{\mathrm{w}}$
then $h\geq\Delta_{k}/2$, we have $f_{a,b}(A^{\top}A)=V_{k}V^{\top}_{k}$
and consequently, simultaneously for all $X$, $\mathsf{s}_{a,b}(A,X)=\|(I-XX^{\top})V_{k}\|_{2}$. 
\end{lem}

\begin{proof}
Fix the initialization and $X$. As in the privacy proofs in Section
\ref{sec:alg-privacy}, we bound the change in the query when one
row is replaced. The scalar function $f_{a,b}$ is $1/h$-Lipschitz.
Consider symmetric matrices $M,M'$ with orthonormal eigenbases $u_{i},v_{j}$
and eigenvalues $\lambda_{i},\lambda'_{j}$ respectively. We will
show that 
\begin{equation}
\sum_{i,j}(\lambda_{i}-\lambda'_{j})^{2}|u^{\top}_{i}v_{j}|^{2}=\|M-M'\|^{2}_{F}\label{eq:Frobenius-identity}
\end{equation}
Because the Frobenius norm does not change under multiplication with
an orthonormal matrix, $\|M-M'\|^{2}_{F}=\left\Vert U^{\top}\left(M-M'\right)V\right\Vert ^{2}_{F}=\sum_{ij}(u^{\top}_{i}\left(M-M'\right)v_{j})^{2}$.
By the definition of eigenvalues, $u^{\top}_{i}M=\lambda_{i}u^{\top}_{i}$
and $M'v_{j}=\lambda'_{j}v_{j}$. Thus, $u^{\top}_{i}\left(M-M'\right)v_{j}=(\lambda_{i}-\lambda'_{j})u^{\top}_{i}v_{j}$
and this proves (\ref{eq:Frobenius-identity}). Applying (\ref{eq:Frobenius-identity}),
we have 
\begin{align*}
\|f_{a,b}(M)-f_{a,b}(M')\|^{2}_{F} & =\sum_{i,j}\bigl(f_{a,b}(\lambda_{i})-f_{a,b}(\lambda'_{j})\bigr)^{2}|u^{\top}_{i}v_{j}|^{2}\\
 & \leq h^{-2}\sum_{i,j}(\lambda_{i}-\lambda'_{j})^{2}|u^{\top}_{i}v_{j}|^{2}=h^{-2}\|M-M'\|^{2}_{F}.
\end{align*}
Consequently, using $\|I-XX^{\top}\|_{2}\leq1$, for neighboring datasets
$A,A'$ we have 
\[
|\mathsf{s}_{a,b}(A,X)-\mathsf{s}_{a,b}(A',X)|\leq\|f_{a,b}(A^{\top}A)-f_{a,b}(A'^{\top}A')\|_{F}\leq\frac{\|A^{\top}A-A'^{\top}A'\|_{F}}{h}\leq\frac{2}{h}.
\]
The last inequality is the sensitivity bound from the proof of Lemma
\ref{lem:private-initialization}.

For the accuracy claim, Weyl's inequality (Theorem \ref{thm:weyl})
gives $\sigma^{2}_{k+1}\leq a$ and $b\leq\sigma^{2}_{k}$. It also
gives $h\geq\Delta_{k}-4e_{\mathrm{w}}\geq\Delta_{k}/2$. Thus $f_{a,b}$
is one on the top $k$ eigenvalues of $A^{\top}A$ and zero on all
remaining eigenvalues. Finally, $\|(I-XX^{\top})V_{k}V^{\top}_{k}\|_{2}=\|(I-XX^{\top})V_{k}\|_{2}$. 
\end{proof}

\subsection{Joint random search and its privacy cost}

We now apply the method of \citet{PapernotSteinke2022} to our problem
with the above score function. Conditioned on the shared initialization
$\widehat{M}$ with $h>0$, define one base trial $Q$ as follows.
Sample $(j,\ell)$ uniformly from $\{0,\ldots,J\}\times\{0,\ldots,H\}$,
and run Algorithm \ref{alg:Algorithm-power-method} from $X^{(0)}$
for $T_{\ell}$ iterations, with radius $R_{j}$, noise parameter
$\Delta(T_{\ell})$, and floor $\gamma(T_{\ell})$. Denote its output
by $X$. With a fresh $Z\sim\mathcal{N}(0,1)$, form the noisy score
\begin{equation}
\widetilde{s}=\mathsf{s}_{a,b}(A,X)+\frac{2\zeta}{h}Z.\label{eq:radius-private-score}
\end{equation}
The base output is the tuple $(R_{j},T_{\ell},X,\widetilde{s})$,
ordered by smaller $\widetilde{s}$ being better, with a fixed total
ordering to break ties.

Independently sample a number of trials $N$ with 
\begin{equation}
\Pr\left[N=r\right]=\theta(1-\theta)^{r-1},\quad r\geq1,\qquad\theta=\frac{\beta}{8m},\qquad\E\left[N\right]=\frac{1}{\theta}.\label{eq:radius-geometric-search}
\end{equation}
Run the base trial $Q$ independently $N$ times with the same initialization.
Release the triple $(\widehat{R},\widehat{T},\widehat{X})$ with the
smallest noisy score. 
\begin{thm}
\label{thm:radius-privacy} Define 
\begin{equation}
\varrho_{\mathrm{init}}=\frac{2}{\Delta^{2}_{\mathrm{init}}},\qquad\varrho_{\mathrm{b}}=\frac{4}{\chi^{2}}+\frac{1}{2\zeta^{2}},\qquad\varrho_{\mathrm{tot}}=\varrho_{\mathrm{init}}+\varrho_{\mathrm{b}}\label{eq:radius-rdp-costs}
\end{equation}
The selection procedure satisfies $(\lambda,\varepsilon_{\mathrm{RDP}}(\lambda))$-Rényi
differential privacy for every $\lambda>1$, with 
\begin{equation}
\varepsilon_{\mathrm{RDP}}(\lambda)\leq\lambda\varrho_{\mathrm{tot}}+4\sqrt{\varrho_{\mathrm{b}}\log(1/\theta)}+\frac{\log(1/\theta)}{\lambda-1}.\label{eq:radius-rdp-bound}
\end{equation}
In particular it is $(\epsilon,\delta)$-DP whenever 
\begin{equation}
\varrho_{\mathrm{tot}}+2\sqrt{\varrho_{\mathrm{tot}}\log(1/(\theta\delta))}+4\sqrt{\varrho_{\mathrm{b}}\log(1/\theta)}\leq\epsilon.\label{eq:radius-dp-calibration}
\end{equation}
For $\epsilon\in(0,1)$, an explicit setting of parameters is 
\begin{equation}
v=\frac{\epsilon^{2}}{100\log(1/(\theta\delta))},\qquad\Delta^{2}_{\mathrm{init}}=\frac{6}{v},\qquad\chi^{2}=\frac{12}{v},\qquad\zeta^{2}=\frac{3}{2v}.\label{eq:radius-explicit-calibration}
\end{equation}
Thus $\Delta(t)=\frac{\sqrt{1200t\log(1/(\theta\delta))}}{\epsilon}$
for each candidate iteration count $t$. 
\end{thm}

\begin{proof}
We follow the proofs of Lemmas \ref{lem:private-initialization} and
\ref{lem:clipped-iterations-privacy}, using Rényi privacy to account
for the hyperparameter search. A Gaussian mechanism with Euclidean
sensitivity $s$ and noise standard deviation $\sigma$ has Rényi
cost $\lambda s^{2}/(2\sigma^{2})$ at order $\lambda$, and these
costs add under adaptive composition \citep{Mironov2017}.

First, the proof of Lemma \ref{lem:private-initialization} shows
that $A^{\top}A$ has Frobenius sensitivity at most $2$ and that
$\sym(G^{(\mathrm{init})})$ is a standard Gaussian in the symmetric-matrix
space. Thus the release of $\widehat{M}$ has Rényi cost 
\[
\frac{\lambda\cdot2^{2}}{2\Delta^{2}_{\mathrm{init}}}=\lambda\varrho_{\mathrm{init}}.
\]
Computing $X^{(0)}$, $a$, $b$, and $h$ is post-processing.

Next, fix any initialization output with $h>0$ and any grid pair
$(R_{j},T_{\ell})$. As in Lemma \ref{lem:clipped-iterations-privacy},
condition on the transcript preceding each release. The query for
$H^{(t)}$ has sensitivity at most $2$, so its Rényi cost is $2\lambda/\Delta(T_{\ell})^{2}$.
The clipped query for $Z^{(t+1)}$ has sensitivity at most $2R_{j}$
and noise scale $\Delta(T_{\ell})R_{j}$, so its Rényi cost is also
$2\lambda/\Delta(T_{\ell})^{2}$. Constructing $B^{(t)}$ and taking
the QR factorization are post-processing. Composition over the $T_{\ell}$
iterations therefore gives cost 
\[
\frac{4T_{\ell}\lambda}{\Delta(T_{\ell})^{2}}=\frac{4\lambda}{\chi^{2}}.
\]
Finally, conditional on the candidate $X$, Lemma \ref{lem:radius-score}
gives sensitivity at most $2/h$ for the score. Its Gaussian release
has additional cost 
\[
\frac{\lambda(2/h)^{2}}{2(2\zeta/h)^{2}}=\frac{\lambda}{2\zeta^{2}}.
\]
Thus a base trial has cost $\lambda\varrho_{\mathrm{b}}$. This bound
is uniform over grid pairs, so sampling the pair independently of
the data and including it in the output preserves the same bound.

Conditional on the fixed initialization, the base trials are independent.
Apply Theorem 2 of \citet{PapernotSteinke2022} with shape parameter
$1$ and stopping parameter $\theta$, corresponding to (\ref{eq:radius-geometric-search}).
For any auxiliary order $\widehat{\lambda}>1$, the cost of returning
the best base tuple is at most 
\[
\lambda\varrho_{\mathrm{b}}+2(\widehat{\lambda}-1)\varrho_{\mathrm{b}}+\frac{2\log(1/\theta)}{\widehat{\lambda}}+\frac{\log(1/\theta)}{\lambda-1}.
\]
Taking $\widehat{\lambda}=1+\sqrt{\log(1/\theta)/\varrho_{\mathrm{b}}}$
bounds the middle two terms by $4\sqrt{\varrho_{\mathrm{b}}\log(1/\theta)}$.
If $h\leq0$, the output is post-processing of the initialization.
Thus the search bound is uniform over initialization outputs. Adaptive
composition with the single initialization gives (\ref{eq:radius-rdp-bound});
discarding the selected noisy score is post-processing.

To obtain an $(\epsilon,\delta)$-DP bound, we convert the Rényi-DP
bound, which adds $\log(1/\delta)/(\lambda-1)$ to the Rényi cost
\citep{Mironov2017}. Setting $\lambda=1+\sqrt{\log(1/(\theta\delta))/\varrho_{\mathrm{tot}}}$
gives (\ref{eq:radius-dp-calibration}). Under (\ref{eq:radius-explicit-calibration}),
$\varrho_{\mathrm{tot}}=v$ and $\varrho_{\mathrm{b}}=2v/3\leq v$.
Writing $D=\log(1/(\theta\delta))>\log8$, the left-hand side of (\ref{eq:radius-dp-calibration})
is bounded by $\epsilon^{2}/(100D)+\epsilon/5+2\epsilon/5<\epsilon$. 
\end{proof}

\subsection{Utility guarantee from hyperparameter search}

We now prove the analogue of Corollary \ref{cor:martingale-convergence-corollary}
when $R$ and $T$ are selected by the search. The good grid pair
used in the proof is defined only for the analysis; the algorithm
does not need $\kappa_{k}$ or $\left\Vert U\right\Vert _{\infty}$.
Let $T_{\star}=\left\lceil \frac{2}{\kappa_{k}}\log\frac{n}{\beta\delta\epsilon}\right\rceil $
and $R_{\star}=15\left\Vert U\right\Vert _{\infty}\left(\sqrt{k}+\sqrt{\log\left(\frac{16n\left(T_{\max}+1\right)}{\beta_{0}}\right)}\right)$.
Under the utility hypotheses below, the public grid bounds give $T_{\min}\leq T_{\star}\leq T_{\max}$
and $R_{\min}\leq R_{\star}\leq R_{\max}$. Consider the smallest
grid iteration count $T^{\dagger}\geq T_{\star}$ and grid radius
$R^{\dagger}\geq R_{\star}$. We have 
\[
T_{\star}\leq T^{\dagger}\leq2T_{\star},\qquad R_{\star}\leq R^{\dagger}\leq2R_{\star}.
\]

\begin{thm}
\label{thm:radius-utility} Let $\epsilon,\delta,\beta\in(0,1)$ and
run the selection procedure with the parameter choices in (\ref{eq:radius-explicit-calibration}).
Suppose that 
\[
\Delta_{k}\geq\widetilde{\Omega}\left(\frac{\sqrt{d}}{\epsilon}+\left(\frac{\sigma_{k}\sqrt{k}}{\epsilon}\right)^{2/3}+\left(\frac{R_{\star}\sigma^{2}_{k}\sqrt{k}}{\epsilon}\right)^{2/3}\right),
\]
where the $\widetilde{\Omega}(\cdot)$ hides sufficiently large poly-logarithmic
factors. With probability at least $1-\beta$, the selected output
satisfies 
\begin{equation}
\|(I-\widehat{X}\widehat{X}^{\top})V_{k}\|^{2}_{2}\leq\left(\frac{\beta\delta\epsilon}{n}\right)^{2}+\widetilde{O}\left(\frac{dk\sigma^{2}_{k}\|U\|^{2}_{\infty}+1}{\epsilon^{2}\Delta^{2}_{k}}\right).\label{eq:radius-oracle-inequality}
\end{equation}
\end{thm}

\begin{proof}
Let $T=T^{\dagger}$, $R=R^{\dagger}$, $\Delta=\Delta(T)$, and $\gamma=\gamma(T)$.
We follow the proof of Corollary \ref{cor:martingale-convergence-corollary},
using the quantities $q_{k}$, $q_{\mathrm{row}}$, $p_{R}$, $\nu_{R}$,
and the good events for this run with failure parameter $\beta_{0}=\beta/2$.

Let us verify that the assumptions imply the conditions in Theorem
\ref{thm:convergence-martingale}. The parameter choices (\ref{eq:radius-explicit-calibration})
give $\Delta_{\mathrm{init}}=\widetilde{O}(1/\epsilon)$ and $\chi=\widetilde{O}(1/\epsilon)$.
The choice $T_{\star}\leq T\leq2T_{\star}$ therefore gives, as before,
\[
\Delta=\widetilde{O}\left(\frac{1}{\epsilon\sqrt{\kappa_{k}}}\right)=\widetilde{O}\left(\frac{\sigma_{k}}{\epsilon\sqrt{\Delta_{k}}}\right).
\]
The gap condition $\Delta_{k}\geq\widetilde{\Omega}\left(\left(\frac{\sigma_{k}\sqrt{k}}{\epsilon}\right)^{2/3}\right)$
implies $\Delta q_{\mathrm{row}}=\widetilde{O}\left(\frac{\sigma_{k}\sqrt{k}}{\epsilon\sqrt{\Delta_{k}}}\right)\leq\Delta_{k}/8$
and also implies $\gamma=\widetilde{O}\left(\Delta\sqrt{k}\right)\leq\sigma^{2}_{k}/4$.

The gap condition $\Delta_{k}\geq\widetilde{\Omega}\left(\left(\frac{R_{\star}\sigma^{2}_{k}\sqrt{k}}{\epsilon}\right)^{2/3}\right)$
and $R\leq2R_{\star}$ imply 
\[
p_{R}=\frac{3\Delta Rq_{k}}{2\sigma_{k}}=\widetilde{O}\left(\frac{R\sqrt{k}}{\epsilon\sqrt{\Delta_{k}}}\right)\leq\frac{\Delta_{k}}{16\sigma^{2}_{k}}=\frac{\kappa_{k}}{16}.
\]
The gap condition $\Delta_{k}\geq\widetilde{\Omega}\left(\frac{\sqrt{d}}{\epsilon}\right)$
ensures that 
\[
\Delta_{k}\geq8\Delta_{\mathrm{init}}\left(2\sqrt{d}+\sqrt{2\log\frac{8}{\beta_{0}}}\right)\geq8e_{\mathrm{w}}.
\]
Finally, $R\geq R_{\star}$ and $T\leq T_{\max}$ give 
\[
R\geq15\left\Vert U\right\Vert _{\infty}\left(\sqrt{k}+\sqrt{\log\left(\frac{16n(T+1)}{\beta_{0}}\right)}\right),
\]
which is the same radius condition as before.

As in the corollary, we next use the good events from Lemmas \ref{lem:clip-good-event},
\ref{lem:initial-potential}, \ref{lem:martingale}, and \ref{lem:initial-error}.
Since the search shares one initialization, we first fix its good
event. By Lemma \ref{lem:initial-potential} and the Gaussian bound
in the proof of Lemma \ref{lem:initial-error}, 
\[
\Pr\left[\ginit\cap\left\{ \left\Vert \widehat{M}-A^{\top}A\right\Vert _{2}\leq e_{\mathrm{w}}\right\} \right]\geq1-\frac{\beta_{0}}{4}-\frac{\beta_{0}}{4}=1-\frac{\beta}{4}.
\]
On this event, the proof of Lemma \ref{lem:initial-error} gives $s_{0}\leq1/4$,
and Lemma \ref{lem:s-to-m} gives $m_{0}\leq1/\sqrt{15}$. Lemma \ref{lem:radius-score}
also gives $h\geq\Delta_{k}/2$ and identifies every comparison score
with the true subspace error.

Condition on any initialization in this event. The proofs of Lemmas
\ref{lem:clip-good-event} and \ref{lem:martingale} apply with this
initialization fixed: the iteration noises are independent of it,
and the stopped martingale coefficients satisfy the same bounds. Thus,
for a run with parameters $(R,T)$, the events $\gclip$ and $\gmart$
hold jointly with conditional probability at least $1-\beta_{0}/4-\beta_{0}/4\geq1/2$.

If $\nu_{R}\leq1/8$, all of the conditions of Theorem \ref{thm:convergence-martingale}
are satisfied on these events, and the theorem gives 
\[
s_{T}\leq\left(1-\frac{\kappa_{k}}{2}\right)^{T}m_{0}+2\nu_{R}\leq\frac{\beta\delta\epsilon}{n\sqrt{15}}+2\nu_{R},
\]
where we used $m_{0}\leq1/\sqrt{15}$ and $T\geq T_{\star}$. Otherwise,
we have $s^{2}_{T}\leq1\leq64\nu^{2}_{R}$, and we have the same asymptotic
guarantee, exactly as in the corollary. In both cases, the run has
error at most 
\[
e_{\dagger}:=\frac{\beta\delta\epsilon}{n\sqrt{15}}+8\nu_{R}
\]
with conditional probability at least $1/2$. Moreover, since $R\leq2R_{\star}$
and $R_{\star}=\widetilde{O}\left(\sqrt{k}\left\Vert U\right\Vert _{\infty}\right)$,
\[
\nu^{2}_{R}=\widetilde{O}\left(\frac{\Delta^{2}R^{2}d}{\sigma^{2}_{k}\kappa_{k}}\right)=\widetilde{O}\left(\frac{d\sigma^{2}_{k}R^{2}}{\epsilon^{2}\Delta^{2}_{k}}\right)=\widetilde{O}\left(\frac{dk\sigma^{2}_{k}\left\Vert U\right\Vert ^{2}_{\infty}}{\epsilon^{2}\Delta^{2}_{k}}\right).
\]
It remains to account for the hyperparameter search and the noisy
comparison scores, still conditional on the fixed good initialization.
Each independent base trial samples $(R,T)$ with probability $1/m$,
so it produces a candidate of error at most $e_{\dagger}$ with probability
$p\geq1/(2m)$. The probability that no such candidate occurs before
the geometric stopping time is 
\[
\E(1-p)^{N}=\frac{\theta(1-p)}{p+\theta(1-p)}\leq\frac{\theta}{p}\leq\frac{\beta}{4}.
\]
Let 
\[
t_{\mathrm{val}}=\frac{2\zeta}{h}\sqrt{2\log\frac{8}{\beta\theta}}.
\]
Each noisy score differs from its true score by more than $t_{\mathrm{val}}$
with probability at most $\beta\theta/4$. Conditioning on $N$ and
taking a union bound over the trials, then using $\E N=1/\theta$,
shows that all score errors are at most $t_{\mathrm{val}}$ with probability
at least $1-\beta/4$. If a candidate $X$ of error at most $e_{\dagger}$
occurs and all score errors satisfy this bound, the selected candidate
satisfies 
\[
\mathsf{s}_{a,b}(A,\widehat{X})\leq\widetilde{s}_{\mathrm{selected}}+t_{\mathrm{val}}\leq\widetilde{s}_{X}+t_{\mathrm{val}}\leq\mathsf{s}_{a,b}(A,X)+2t_{\mathrm{val}}\leq e_{\dagger}+2t_{\mathrm{val}}.
\]
Since $h\geq\Delta_{k}/2$ and $\zeta=\widetilde{O}(1/\epsilon)$,
we have $t^{2}_{\mathrm{val}}=\widetilde{O}\left(1/(\epsilon^{2}\Delta^{2}_{k})\right)$.
Consequently, 
\begin{align*}
\left\Vert (I-\widehat{X}\widehat{X}^{\top})V_{k}\right\Vert ^{2}_{2} & \leq\left(e_{\dagger}+2t_{\mathrm{val}}\right)^{2}\\
 & \leq\frac{2}{15}\left(\frac{\beta\delta\epsilon}{n}\right)^{2}+O\left(\nu^{2}_{R}+t^{2}_{\mathrm{val}}\right)\\
 & \leq\left(\frac{\beta\delta\epsilon}{n}\right)^{2}+\widetilde{O}\left(\frac{dk\sigma^{2}_{k}\left\Vert U\right\Vert ^{2}_{\infty}+1}{\epsilon^{2}\Delta^{2}_{k}}\right).
\end{align*}
The initialization, search, and score-error events fail with total
probability at most $\beta/4+\beta/4+\beta/4\leq\beta$, proving (\ref{eq:radius-oracle-inequality}).
\end{proof}

\subsection{Overall runtime}
\begin{thm}[Expected running time, including hyperparameter search]
\label{thm:expected-runtime-search}Let $A\in\mathbb{R}^{n\times d}$,
with $1\le k<d$, and let $\epsilon,\delta,\beta\in(0,1)$. Let $B=\max\left\{ 1,\left\lceil \frac{2n\epsilon}{\sqrt{d}}\log\left(\frac{n}{\beta\delta\epsilon}\right)\right\rceil \right\} $,
and $L_{R}=2+\log_{2}\!\left[15\sqrt{n}\left(1+\sqrt{\frac{1}{k}\log\frac{64nB}{\beta}}\right)\right]=O(1+\log(n\log(B/\beta)))$.
The expected runtime of the entire algorithm including hyperparameter
selection is $O\!\left(nd^{2}+d^{3}+\frac{BL_{R}}{\beta}(ndk+dk^{2})+L_{R}\left(1+\log B\right)d^{3}\right)$.
\end{thm}

\begin{proof}
For the public grids in the hyperparameter-selection procedure, recall
that we can choose $T_{\max}$ to be the power of 2 satisfying $B\le T_{\max}<2B$.
and $m_{T}=\log_{2}T_{\max}+1\le2+\log_{2}B$. For the clipping radius
parameter, the grid starts at $R_{\min}=\sqrt{k/n}$ and keeps doubling
until it reaches $15\left(\sqrt{k}+\sqrt{\log\frac{32n(T_{\max}+1)}{\beta}}\right)$
so 
\[
m_{R}=1+\left\lceil \log_{2}\!\left[15\sqrt{n}\left(1+\sqrt{\frac{1}{k}\log\frac{32n(T_{\max}+1)}{\beta}}\right)\right]\right\rceil \le L_{R}.
\]
Each trial samples a parameter pair from the grid uniformly, and the
number of trials is independently distributed as $N\sim\operatorname{Geom}\!\left(\frac{\beta}{8m_{R}m_{T}}\right)$.

Forming $M=A^{\top}A$, computing the shared initialization, and caching
an eigendecomposition of $M$ take $O(nd^{2}+d^{3})$ time. The runtime
of each power iteration is
\[
O\!\left(ndk+(n+d)k^{2}+k^{3}\right)=O(ndk+dk^{2}),
\]
including applying the preconditioner, rowwise clipping, and thin
QR decomposition.

Conditional on entering the parameter search,
\[
\mathbb{E}[N]=\frac{8m_{R}m_{T}}{\beta}.
\]
In each trial, the iteration count is sampled uniformly from $\{2^{\ell}:0\le\ell\le H\}$
so the total expected iteration count for all trials is
\[
\mathbb{E}\!\left[\sum^{N}_{i=1}[\text{number of iterations in trial }i]\right]=\frac{8m_{R}m_{T}}{\beta}\frac{\sum^{H}_{\ell=0}2^{\ell}}{m_{T}}=\frac{8m_{R}}{\beta}(2T_{\max}-1).
\]
Thus the total expected runtime of all power iterations is
\[
O\!\left(\frac{m_{R}T_{\max}}{\beta}(ndk+dk^{2})\right).
\]
Recall that in the shared initialization, we release $\widehat{M}$
with eigenvalues $\widehat{\lambda}_{1}\geq\cdots\geq\widehat{\lambda}_{d}$
$\widehat{\lambda}_{1}\geq\cdots\geq\widehat{\lambda}_{d}$ and $a=\widehat{\lambda}_{k+1}+e_{\mathrm{w}},b=\widehat{\lambda}_{k}-e_{\mathrm{w}}$.
For the score computation, one can speed up with high probability
as follows. We check using the cached eigenvalues of $M$, whether
\begin{equation}
\lambda_{k+1}(M)\le a<b\le\lambda_{k}(M).\label{eq:cheap-scores}
\end{equation}
If this holds, then by definition, $f_{a,b}(M)=V_{k}V^{\top}_{k}$,
and each score can be computed exactly as $\mathsf{s}_{a,b}(A,X)=\|V_{k}-X(X^{\top}V_{k})\|_{2}$
in $O(dk^{2})$ operations. Otherwise, as a fallback, the algorithm
computes $f_{a,b}(M)$ and applies the definition $\mathsf{s}_{a,b}(A,X)=\|(I-XX^{\top})f_{a,b}(M)\|_{2}$,
which requires $O(d^{3})$ time per score.

On the initialization event $E_{\mathrm{w}}=\{\|\widehat{M}-M\|_{2}\le e_{\mathrm{w}}\}$,
Weyl's inequality implies $\lambda_{k+1}(M)\le a$ and $b\le\lambda_{k}(M)$.
Hence, whenever the search continues with $h=b-a>0$, the inexpensive
score computation is valid on $E_{\mathrm{w}}$. Since $\|\sym(G^{(\mathrm{init})})\|_{2}\le\|G^{(\mathrm{init})}\|_{2}$,
Lemma \ref{lem:gaussian-norm} applied to the $d\times d$ matrix
$G^{(\mathrm{init})}$ with $x=\sqrt{2\log(4/\beta_{0})}$ gives
\begin{align*}
\Pr(E^{c}_{\mathrm{w}}) & \le\Pr\!\left[\|G^{(\mathrm{init})}\|_{2}>2\sqrt{d}+\sqrt{2\log(4/\beta_{0})}\right]\\
 & \le e^{-\log(4/\beta_{0})}=\frac{\beta_{0}}{4}=\frac{\beta}{8},
\end{align*}
where $\beta_{0}=\beta/2$. Since the number of trials $N$ is independent
of the initialization, the expected fallback cost is at most
\[
O\!\left(d^{3}\,\mathbb{E}[N]\Pr(E^{c}_{\mathrm{w}})\right)=O(m_{R}m_{T}d^{3}).
\]
If the test in (\ref{eq:cheap-scores}) passes, the expected runtime
of all score computation is $O\left(\frac{m_{R}m_{T}}{\beta}dk^{2}\right)$,
which is a lower order term compared with the power iterations because
$m_{T}\le T_{\max}$. The early return when $h\le0$ only reduces
runtime. Therefore, the total expected runtime is 
\[
O\!\left(nd^{2}+d^{3}+\frac{m_{R}T_{\max}}{\beta}(ndk+dk^{2})+m_{R}m_{T}d^{3}\right).
\]
Substituting $T_{\max}<2B$, $m_{R}\le L_{R}$, and $m_{T}\le2+\log_{2}B$
proves the claim.
\end{proof}

\begin{rem}
If we always compute scores using the definition without the speed
up, the expected runtime is $O\!\left(nd^{2}+d^{3}+\frac{BL_{R}}{\beta}(ndk+dk^{2})+\frac{L_{R}(1+\log B)}{\beta}d^{3}\right)$.
\end{rem}

\section{Lower bound}

We first prove a lower bound for the special case $k=1$, and then
extend the construction to any $k$. 

\subsection{Rank one}
\begin{thm}
\label{thm:rank-one-lower-bound} Fix $c_{0}>0$. Let $d\ge2$, $\epsilon\in(0,1]$,
$\delta\in[0,\epsilon/100]$, and $r\in[c_{0}/(d+1),1)$. Set $t=\lfloor1/\epsilon\rfloor$
and let $m$ be an integer satisfying 
\[
m\equiv t\pmod 2,\qquad m\ge\frac{8t}{1-r},\qquad n=md.
\]
There exists a family of datasets $\{A_{\omega}:\omega\in\{-1,1\}^{d}\}\subseteq\mathbb{R}^{n\times(d+1)}$
with the following properties:
\begin{enumerate}
\item Every row of every $A_{\omega}$ has Euclidean norm one.
\item All $A_{\omega}$ have the same singular values and satisfy $\sigma^{2}_{2}/\sigma^{2}_{1}=r$.
\item Every $A_{\omega}$ admits an SVD $A_{\omega}=U_{\omega}\Sigma V^{\top}_{\omega}$
such that 
\[
\frac{1}{n}\le\|U_{\omega}\|^{2}_{\infty}\le\frac{2}{n},\qquad\|U_{\omega}\|_{\infty}:=\max_{i,j}|(U_{\omega})_{ij}|.
\]
\end{enumerate}
Let $\Delta_{1}=\sigma^{2}_{1}-\sigma^{2}_{2}$, and let $v_{\omega}$
be the top right-singular vector of $A_{\omega}$. For every row-replacement
$(\epsilon,\delta)$-differentially private algorithm $M$ that takes
as input a matrix whose rows have norms at most 1 and outputs a unit
vector, half of the $\omega$ satisfies 
\[
\Pr\left[\sin^{2}\angle\bigl(M(A_{\omega}),v_{\omega}\bigr)\ge c_{\star}\frac{(d+1)\sigma^{2}_{1}\|U_{\omega}\|^{2}_{\infty}+1}{\epsilon^{2}\Delta^{2}_{1}}\right]\ge\frac{1}{4},\qquad c_{\star}=\frac{1}{2048}\frac{2c_{0}}{3+2c_{0}}.
\]
The probability is over the randomness of $M$.
\end{thm}

\begin{proof}
Define 
\[
J=1+\frac{rt^{2}/m^{2}}{1-r},\qquad a^{2}=\frac{1}{1+drJ},\qquad b^{2}=1-a^{2},
\]
where $a,b>0$. Index the standard basis of $\mathbb{R}^{d+1}$ by
$e_{0},e_{1},\ldots,e_{d}$. For each $\omega\in\{-1,1\}^{d}$, $A_{\omega}$
consists of $d$ blocks of rows. Block $j$ consists of $(m-t)/2$
copies of each of the two rows 
\[
(ae_{0}+be_{j})^{\top},\qquad(ae_{0}-be_{j})^{\top},
\]
followed by $t$ copies of $(ae_{0}+b\omega_{j}e_{j})^{\top}$. All
rows have norm one, and 
\[
A^{\top}_{\omega}A_{\omega}=\begin{pmatrix}na^{2} & abt\,\omega^{\top}\\
abt\,\omega & mb^{2}I_{d}
\end{pmatrix}.
\]
Next we compute its eigenvalues. For every $x\in\mathbb{R}^{d}$ satisfying
$\omega^{\top}x=0$, 
\[
A^{\top}_{\omega}A_{\omega}\begin{pmatrix}0\\
x
\end{pmatrix}=mb^{2}\begin{pmatrix}0\\
x
\end{pmatrix}.
\]
Thus $mb^{2}$ is an eigenvalue on a $(d-1)$-dimensional subspace.
Its orthogonal complement has the orthonormal basis 
\[
e_{0},\qquad q_{\omega}=\begin{pmatrix}0\\
\omega/\sqrt{d}
\end{pmatrix}.
\]
The restriction of $A^{\top}_{\omega}A_{\omega}$ to this two-dimensional
subspace is represented by 
\[
K=\begin{pmatrix}na^{2} & abt\sqrt{d}\\
abt\sqrt{d} & mb^{2}
\end{pmatrix}.
\]
Since $n=md$, we have $(abt\sqrt{d})^{2}=(na^{2})(mb^{2})t^{2}/m^{2}$.
Therefore the remaining two eigenvalues are the roots of 
\[
\det(\lambda I_{2}-K)=(\lambda-na^{2})(\lambda-mb^{2})-(na^{2})(mb^{2})t^{2}/m^{2}=\lambda^{2}-(na^{2}+mb^{2})\lambda+(na^{2})(mb^{2})(1-t^{2}/m^{2}).
\]
The definitions of $a,b,J$ imply 
\[
b^{2}=da^{2}rJ,\qquad mb^{2}=na^{2}rJ,\qquad J+r(1-t^{2}/m^{2})=1+rJ.
\]
Consequently, the characteristic polynomial factors as 
\begin{align*}
\det(\lambda I_{2}-K) & =\lambda^{2}-na^{2}(1+rJ)\lambda+(na^{2})^{2}Jr(1-t^{2}/m^{2})\\
 & =(\lambda-na^{2}J)(\lambda-na^{2}r(1-t^{2}/m^{2})).
\end{align*}
Combining these two roots with the repeated eigenvalue $mb^{2}$ gives
\[
\lambda_{1}=\frac{mb^{2}}{r}=na^{2}J,\qquad\lambda_{2}=\cdots=\lambda_{d}=mb^{2},\qquad\lambda_{d+1}=\frac{mb^{2}(1-t^{2}/m^{2})}{J}.
\]
Here $0<\lambda_{d+1}<mb^{2}<\lambda_{1}$, so the spectrum is independent
of $\omega$ and $\lambda_{2}/\lambda_{1}=r$.

The leading eigenvector has the form 
\[
v_{\omega}=\cos\theta\,e_{0}+\sin\theta\,(0,\omega/\sqrt{d})
\]
where $\theta\in[0,\pi/2]$. We have 
\[
K\begin{pmatrix}\cos\theta\\
\sin\theta
\end{pmatrix}=\lambda_{1}\begin{pmatrix}\cos\theta\\
\sin\theta
\end{pmatrix}
\]
so $abt\sqrt{d}\cos\theta+mb^{2}\sin\theta=\lambda_{1}\sin\theta$
and $\tan\theta=\frac{abt\sqrt{d}}{\lambda_{1}-mb^{2}}=\frac{abt\sqrt{d}}{\Delta_{1}}$.
In particular, 
\begin{equation}
\tan^{2}\theta=\frac{a^{2}b^{2}t^{2}d}{m^{2}b^{4}(1/r-1)^{2}}=\frac{rt^{2}}{Jm^{2}(1-r)^{2}}\le\frac{r}{64J}\le\frac{1}{64}.\label{eq:lb-angle}
\end{equation}
We next verify the coherence bound for all left singular vectors.
Define the real Fourier basis $H\in\mathbb{R}^{d\times d}$ explicitly
as follows. For each $j=1,\ldots,d$, set $H_{j1}=\frac{1}{\sqrt{d}}$.
For $k=1,\ldots,\lfloor(d-1)/2\rfloor$, set 
\[
\begin{aligned}H_{j,2k} & =\sqrt{\frac{2}{d}}\cos\!\left(\frac{2\pi k(j-1)}{d}\right),\\
H_{j,2k+1} & =\sqrt{\frac{2}{d}}\sin\!\left(\frac{2\pi k(j-1)}{d}\right).
\end{aligned}
\]
If $d$ is even, also set $H_{j,d}=\frac{(-1)^{j-1}}{\sqrt{d}}$.
The paired sine and cosine columns fill columns $2,\ldots,d$ when
$d$ is odd, and columns $2,\ldots,d-1$ when $d$ is even. Standard
trigonometric calculation gives $H^{\top}H=I_{d}$. Its first column
is $\mathbf{1}/\sqrt{d}$, and $|H_{j\ell}|\le\sqrt{2/d}$ for every
$j,\ell$. For the eigenspace of eigenvalue $mb^{2}$, choose the
right eigenvectors 
\[
v_{\ell}=(0,\operatorname{diag}(\omega)H_{\ell}),\qquad\ell=2,\ldots,d.
\]
Let $a_{i}$ be the $i$-th row of $A_{\omega}$. Because $U_{\omega}=A_{\omega}V_{\omega}\Sigma^{-1}$,
we have $\left(U_{\omega}\right)_{i,\ell}=\frac{\left\langle a_{i},v_{\ell}\right\rangle }{\sqrt{\lambda_{\ell}}}$
for all $\ell=1,\ldots,d+1$. First we consider the case $\ell=2,\ldots,d$.
For rows in block $j\in\{1,\ldots,d\}$, we have $a_{i}=\left(ae_{0}+\xi be_{j}\right)^{\top}$
where $\xi\in\{-1,1\}$. Thus, $\left|\left(U_{\omega}\right)_{i,\ell}\right|=\frac{\left|\xi b\omega_{j}H_{j\ell}\right|}{\sqrt{mb^{2}}}\le\frac{\sqrt{2/d}}{\sqrt{m}}=\sqrt{2/n}$.

Next we consider the case $\ell=1$ and $\ell=d+1$. We computed the
leading right singular vector above. The right singular vector corresponding
to $\lambda_{d+1}$ is $-\sin\theta\,e_{0}+\cos\theta\,(0,\omega/\sqrt{d})$,
which can be computed similarly. For the rows $a_{i}=(ae_{0}+\xi be_{j})^{\top}$
where $\xi\in\{-1,1\}$, we have 
\begin{align*}
|\left(U_{\omega}\right)_{i,1}| & =\left|\frac{a\cos\theta+\xi b\omega_{j}\left(\sin\theta\right)/\sqrt{d}}{\sqrt{na^{2}J}}\right|\\
 & =\left|\frac{a\cos\theta+\xi a\sqrt{rJ}\omega_{j}\sin\theta}{\sqrt{na^{2}J}}\right|\le\frac{\cos\theta+\sqrt{rJ}\sin\theta}{\sqrt{nJ}}\\
\left|\left(U_{\omega}\right)_{i,d+1}\right| & =\left|\frac{-a\sin\theta+\xi b\omega_{j}\left(\cos\theta\right)/\sqrt{d}}{\sqrt{na^{2}r(1-t^{2}/m^{2})}}\right|\\
 & =\left|\frac{-a\sin\theta+\xi a\sqrt{rJ}\omega_{j}\cos\theta}{\sqrt{na^{2}r(1-t^{2}/m^{2})}}\right|\le\frac{\cos\theta+\sin\theta/\sqrt{rJ}}{\sqrt{n(1-t^{2}/m^{2})/J}}.
\end{align*}
Our choice of $m$ gives 
\[
\frac{t^{2}}{m^{2}}\le\frac{(1-r)^{2}}{64},\qquad1\le J\le\frac{257}{256},\qquad\frac{1-t^{2}/m^{2}}{J}\ge\frac{252}{257}.
\]
By (\ref{eq:lb-angle}), we have $\tan\theta\le\sqrt{r/J}/8$. Thus,
we can bound 
\begin{align*}
|U_{i,1}| & \le\frac{1+\sqrt{rJ}\sqrt{r/J}/8}{\sqrt{nJ}}\le\frac{1+r/8}{\sqrt{nJ}}\le\sqrt{\frac{2}{n}}\\
\left|U_{i,d+1}\right| & \le\frac{1+\sqrt{r/J}/(8\sqrt{rJ})}{\sqrt{n(1-t^{2}/m^{2})/J}}\le\frac{1+1/(8J)}{\sqrt{252n/257}}\le\sqrt{\frac{2}{n}}
\end{align*}
Hence $\|U_{\omega}\|^{2}_{\infty}\le2/n$. The reverse bound $\|U_{\omega}\|^{2}_{\infty}\ge1/n$
follows because every column of $U_{\omega}$ has norm one.

As we saw above, $v_{\omega}$ contains $\omega$ so if we can recover
$v_{\omega}$, we also recover $\omega$. We will show that private
algorithm cannot recover $\omega$ with low error, and hence, it cannot
recover $v_{\omega}$ with low error. Let $y=M(A_{\omega})$ be the
algorithm's output. We define the decoded vector $\widehat{\omega}(y)$
as follows. 
\[
\widehat{\omega}_{j}(y)=\operatorname{sign}(y_{0}y_{j}),
\]
with an arbitrary fixed convention for the sign of zero. Datasets
whose indices $\omega$ differ in one bit differ in exactly $t$ rows.
Group privacy therefore applies with parameters 
\[
\epsilon_{t}=t\epsilon\le1,\qquad\delta_{t}=\delta\sum^{t-1}_{\ell=0}e^{\ell\epsilon}\le et\delta\le e/100.
\]
Let $d_{H}$ denote Hamming distance.

First suppose $d\ge32$. For an integer $q\in\{0,\ldots,d-1\}$ and
a coordinate $j$, consider the output event 
\[
E_{\omega,j}=\{y:d_{H}(\widehat{\omega}(y),\omega)\le q,\ \widehat{\omega}_{j}(y)=\omega_{j}\}.
\]
Note that $\mathbf{1}_{\{y\in E_{\omega,j}\}}=\mathbf{1}_{\{d_{H}(\widehat{\omega}(y),\omega)\le q\}}\cdot\mathbf{1}_{\{\widehat{\omega}_{j}(y)=\omega_{j}\}}$.
Let $\omega^{(j)}$ be obtained from $\omega$ by flipping its $j$th
coordinate. For $y\in E_{\omega,j}$, we have 
\[
d_{H}(\widehat{\omega}(y),\omega^{(j)})=d_{H}(\widehat{\omega}(y),\omega)+1,\qquad\widehat{\omega}_{j}(y)\ne\omega^{(j)}_{j}.
\]
Apply group privacy to $E_{\omega,j}$:
\[
\mathbb{E}[\mathbf{1}_{\{M(A_{\omega})\in E_{\omega,j}\}}]\le e^{\epsilon_{t}}\mathbb{E}[\mathbf{1}_{\{M(A_{\omega^{(j)}})\in E_{\omega,j}\}}]+\delta_{t}
\]
Sum over $j$, and average uniformly over $\omega$:
\[
\mathbb{E}_{\omega}\left[\sum^{d}_{j=1}\mathbb{E}[\mathbf{1}_{\{M(A_{\omega})\in E_{\omega,j}\}}]\right]\le e^{\epsilon_{t}}\mathbb{E}_{\omega}\left[\sum^{d}_{j=1}\mathbb{E}[\mathbf{1}_{\{M(A_{\omega^{(j)}})\in E_{\omega,j}\}}]\right]+d\delta_{t}.
\]
Observe that $\sum_{j}\mathbf{1}_{\{\widehat{\omega}_{j}(y)=\omega_{j}\}}=d-d_{H}(\widehat{\omega}(y),\omega)$
so the left-hand side can be rewritten as 
\[
\mathbb{E}_{\omega,y=M(A_{\omega})}\left[\left(d-d_{H}(\widehat{\omega}(y),\omega)\right)\mathbf{1}_{\{d_{H}(\widehat{\omega}(y),\omega)\le q\}}\right].
\]
For the right-hand side, let $\nu=\omega^{(j)}$. Note that for any
fixed $j$, when $\omega$ is uniformly distributed, $\nu$ is also
uniformly distributed. Furthermore, $y=M(A_{\nu})\in E_{\omega,j}$
if and only if $\widehat{\omega}_{j}(y)\ne\nu_{j}$ and $d_{H}(\widehat{\omega}(y),\nu)=d_{H}(\widehat{\omega}(y),\omega)+1\le q+1$.
We have $\sum_{j}\mathbf{1}_{\{\widehat{\omega}_{j}(y)\ne\nu_{j}\}}=d_{H}(\widehat{\omega}(y),\nu)$.
Thus, the right-hand side can be rewritten as $e^{\epsilon_{t}}\mathbb{E}_{\nu,y=M(A_{\nu})}\left[d_{H}(\widehat{\omega}(y),\nu)\mathbf{1}_{\{d_{H}(\widehat{\omega}(y),\nu)\le q+1\}}\right]+d\delta_{t}$.
Putting them together, we obtain
\begin{align*}
 & \mathbb{E}_{\omega,y=M(A_{\omega})}\left[\left(d-d_{H}(\widehat{\omega}(y),\omega)\right)\mathbf{1}_{\{d_{H}(\widehat{\omega}(y),\omega)\le q\}}\right]\\
 & \le e^{\epsilon_{t}}\mathbb{E}_{\nu,y=M(A_{\nu})}\left[d_{H}(\widehat{\omega}(y),\nu)\cdot\mathbf{1}_{\{d_{H}(\widehat{\omega}(y),\nu)\le q+1\}}\right]+d\delta_{t},
\end{align*}
Consequently, 
\begin{equation}
\mathbb{E}_{\omega,y=M(A_{\omega})}\left[\mathbf{1}_{\{d_{H}(\widehat{\omega}(y),\omega)\le q\}}\right]\le\frac{e^{\epsilon_{t}}(q+1)+d\delta_{t}}{d-q}.\label{eq:lb-reconstruction}
\end{equation}
Take $q=\lceil d/32\rceil-1$. Since $d\ge32$, we have $q+1\le d/16$
and $d-q\ge31d/32$. Thus (\ref{eq:lb-reconstruction}) implies 
\[
\Pr_{\omega,y=M(A_{\omega})}[d_{H}(\widehat{\omega}(y),\omega)<d/32]\le\frac{32}{31}\left(\frac{e}{16}+\frac{e}{100}\right)<\frac{1}{4}.
\]
The average, over $\omega$, of the conditional failure probability
$\Pr_{M}[d_{H}(\widehat{\omega}(y),\omega)\ge d/32]$ therefore exceeds
$3/4$. Since every conditional probability is at most one, at least
half of the choices for $\omega$ have conditional failure probability
at least $1/2$.

For $2\le d<32$, pair all choices of $\omega$ by flipping the first
coordinate. For a pair $\omega^{+}$ and $\omega^{-}$ differing only
in the first bit ($\omega^{+}_{1}=1$,$\omega^{-}_{1}=-1$), let $p_{+}$
and $p_{-}$ be the respective bit-error probabilities $p_{+}=\Pr_{y=M(A_{\omega^{+}})}[\widehat{\omega}_{1}(y)=-1]$
and $p_{-}=\Pr_{y=M(A_{\omega^{-}})}[\widehat{\omega}_{1}(y)=1]$.
Group privacy in both directions gives 
\[
1-p_{+}\le e^{\epsilon_{t}}p_{-}+\delta_{t},\qquad1-p_{-}\le e^{\epsilon_{t}}p_{+}+\delta_{t}.
\]
Consequently, 
\[
\frac{p_{+}+p_{-}}{2}\ge\frac{1-\delta_{t}}{1+e^{\epsilon_{t}}}\ge\frac{1}{4}.
\]
Thus, at least one of $p_{+},p_{-}$ is at least $1/4$. One bit error
implies $d_{H}(\widehat{\omega}(y),\omega)\ge1\ge d/32$, so at least
half of the choices for $\omega$ again have conditional failure probability
at least $1/4$.

Combining the two cases $d\ge32$ and $2\le d<32$, there is a set
$\mathcal{W}_{M}\subseteq\{-1,1\}^{d}$ such that 
\begin{equation}
\begin{aligned}|\mathcal{W}_{M}| & \ge2^{d-1},\\
\Pr_{y=M(A_{\omega})}\left[d_{H}(\widehat{\omega}(y),\omega)\ge d/32\right] & \ge\frac{1}{4}\qquad(\forall\omega\in\mathcal{W}_{M}).
\end{aligned}
\label{eq:lb-hamming}
\end{equation}

Consider an arbitrary $\omega\in\mathcal{W}_{M}$. Let $\widehat{P}=yy^{\top}$
and $P_{\omega}=v_{\omega}v^{\top}_{\omega}$. For each bit $j$ such
that $\widehat{\omega}_{j}(y)=\operatorname{sign}(y_{0}y_{j})\ne\omega_{j}$,
it contributes at least $2\cos^{2}\theta\sin^{2}\theta/d$ to $\|\widehat{P}-P_{\omega}\|^{2}_{F}$,
through entries $(0,j)$ and $(j,0)$. Therefore 
\begin{equation}
\sin^{2}\angle(y,v_{\omega})=\frac{1}{2}\|\widehat{P}-P_{\omega}\|^{2}_{F}\ge\frac{\cos^{2}\theta\sin^{2}\theta}{d}d_{H}(\widehat{\omega}(y),\omega).\label{eq:lb-projector}
\end{equation}

It remains to compare this quantity to the claimed rate. Since $dr\ge2c_{0}/3$
and $J\ge1$, 
\[
b^{2}=\frac{drJ}{1+drJ}\ge b_{0}:=\frac{2c_{0}}{3+2c_{0}}.
\]
Also $t\ge1/(2\epsilon)$ and, by (\ref{eq:lb-angle}), 
\[
\sin^{2}\theta=\frac{\tan^{2}\theta}{1+\tan^{2}\theta}\ge\frac{da^{2}b^{2}t^{2}}{\Delta^{2}_{1}}\cdot\frac{1}{1+1/64}\ge\frac{16}{65}\frac{da^{2}b_{0}}{\epsilon^{2}\Delta^{2}_{1}}.
\]
Every row of $A_{\omega}$ has norm one, so $\lambda_{1}\ge\operatorname{tr}(A^{\top}_{\omega}A_{\omega})/(d+1)=n/(d+1)$.
Together with $\|U_{\omega}\|^{2}_{\infty}\ge1/n$, this shows that
$(d+1)\lambda_{1}\|U_{\omega}\|^{2}_{\infty}\ge1$. Hence 
\[
\frac{(d+1)\lambda_{1}\|U_{\omega}\|^{2}_{\infty}+1}{\epsilon^{2}\Delta^{2}_{1}}\le\frac{2(d+1)(na^{2}J)(2/n)}{\epsilon^{2}\Delta^{2}_{1}}=\frac{4(d+1)a^{2}J}{\epsilon^{2}\Delta^{2}_{1}}.
\]
Using $d/(d+1)\ge2/3$, $J\le257/256$, and $\cos^{2}\theta=\frac{1}{1+\tan^{2}\theta}\ge64/65$,
we obtain 
\[
\frac{\cos^{2}\theta\sin^{2}\theta}{32}\ge\frac{2}{65}\cdot\frac{16}{65}\cdot\frac{2}{3}\frac{(d+1)a^{2}b_{0}}{\epsilon^{2}\Delta^{2}_{1}}\ge\frac{b_{0}}{2048}\frac{(d+1)\lambda_{1}\|U_{\omega}\|^{2}_{\infty}+1}{\epsilon^{2}\Delta^{2}_{1}}.
\]
Combining this with (\ref{eq:lb-hamming}) and (\ref{eq:lb-projector})
proves the claimed lower bound, with failure probability at least
$1/4$ for every $\omega\in\mathcal{W}_{M}$. Since $|\mathcal{W}_{M}|\ge2^{d-1}$,
this accounts for at least half of the family.
\end{proof}

\subsection{Rank $k$}
\begin{thm}
\label{thm:rank-k-lower-bound} Let $k\ge1$ and $d\ge2$ be integers.
Let $\epsilon\in(0,1]$, $\delta\in[0,\epsilon/100]$, and $r\in(0,1)$.
Set $t=\lfloor1/\epsilon\rfloor$, and let $m$ be an integer satisfying
\[
m\equiv t\pmod 2,\qquad m\ge\frac{8t}{1-r},\qquad n_{0}=md,\qquad n=(2k-1)n_{0}.
\]
There exists a family $\{A_{\omega}:\omega\in\{-1,1\}^{d}\}\subseteq\mathbb{R}^{n\times(d+k)}$
with the following properties: 
\begin{enumerate}
\item Every row of every $A_{\omega}$ has Euclidean norm at most one. 
\item All $A_{\omega}$ have the same singular values, with $\sigma^{2}_{1}=\cdots=\sigma^{2}_{k}$
and $\sigma^{2}_{k+1}/\sigma^{2}_{k}=r$. 
\item Every $A_{\omega}$ admits a thin SVD $A_{\omega}=U_{\omega}\Sigma V^{\top}_{\omega}$
such that 
\[
\frac{1}{n}\le\|U_{\omega}\|^{2}_{\infty}\le\frac{2}{kn_{0}}\le\frac{4}{n}.
\]
\end{enumerate}
Let $\Delta_{k}=\sigma^{2}_{k}-\sigma^{2}_{k+1}$, and let $V_{k,\omega}$
be an orthonormal basis for the top-$k$ right-singular subspace of
$A_{\omega}$. For every row-replacement $(\epsilon,\delta)$-differentially
private algorithm $M$ that takes as input a matrix whose rows have
norms at most one and outputs $\widehat{X}\in\mathbb{R}^{(d+k)\times k}$
with orthonormal columns, at least half of the choices of $\omega$
satisfy 
\begin{equation}
\Pr\left[\left\Vert (I-\widehat{X}\widehat{X}^{\top})V_{k,\omega}\right\Vert ^{2}_{2}\ge\frac{1}{1024}\frac{d}{d+k}\frac{dr}{k+dr}\frac{(d+k)k\sigma^{2}_{k}\|U_{\omega}\|^{2}_{\infty}+1}{\epsilon^{2}\Delta^{2}_{k}}\right]\ge\frac{1}{4},\qquad\widehat{X}=M(A_{\omega}).\label{eq:lb-rank-k-general}
\end{equation}
In particular, fix $c_{0}>0$. If $d\ge k$ and $r\ge c_{0}k/(d+k)$,
then the same conclusion holds with threshold 
\begin{equation}
c_{\star,k}\frac{(d+k)k\sigma^{2}_{k}\|U_{\omega}\|^{2}_{\infty}+1}{\epsilon^{2}\Delta^{2}_{k}},\qquad c_{\star,k}=\frac{1}{2048}\frac{c_{0}}{2+c_{0}}.\label{eq:lb-rank-k-matching}
\end{equation}
The probabilities are over the randomness of $M$. 
\end{thm}

\begin{proof}
Set 
\[
\rho=\frac{t^{2}}{km^{2}},\qquad J=1+\frac{r\rho}{1-r},\qquad a^{2}=\frac{k}{k+drJ},\qquad b^{2}=1-a^{2}=\frac{drJ}{k+drJ},
\]
where $a,b>0$. Let $e_{0},e_{1},\ldots,e_{d},f_{1},\ldots,f_{k-1}$
be the standard orthonormal basis of $\mathbb{R}^{d+k}$. The first
$n_{0}$ rows of $A_{\omega}$ are the same rows as in the rank-one
construction: for each $j\in\{1,\ldots,d\}$, take $(m-t)/2$ copies
of each of $(ae_{0}+be_{j})^{\top}$ and $(ae_{0}-be_{j})^{\top}$,
followed by $t$ copies of $(ae_{0}+b\omega_{j}e_{j})^{\top}$. Append
the following rows, which do not depend on $\omega$: 
\begin{itemize}
\item $(k-1)m$ copies of $(be_{j})^{\top}$ for each $j=1,\ldots,d$; 
\item $n_{0}$ copies of $(a\sqrt{J}f_{\ell})^{\top}$ for each $\ell=1,\ldots,k-1$. 
\end{itemize}
There are $n=(2k-1)n_{0}$ rows. The first $n_{0}$ rows of $A_{\omega}$
have norm one, the next $(k-1)md$ rows have norm $b\le1$, and the
last $(k-1)n_{0}$ rows have squared norm 
\[
a^{2}J=\frac{kJ}{k+drJ}=\frac{k+rt^{2}/((1-r)m^{2})}{k+drJ}\le\frac{k+r(1-r)/64}{k+drJ}\le1
\]
The first inequality above follows from $8t/(1-r)\le m$. Thus all
rows of $A_{\omega}$ have norm at most 1.

\textbf{Spectrum and the top subspace.} In the coordinates of the
basis $e_{0},e_{1},\ldots,e_{d},f_{1},\ldots,f_{k-1}$, we have 
\[
A^{\top}_{\omega}A_{\omega}=\begin{pmatrix}n_{0}a^{2} & abt\,\omega^{\top} & 0\\
abt\,\omega & kmb^{2}I_{d} & 0\\
0 & 0 & n_{0}a^{2}JI_{k-1}
\end{pmatrix}.
\]
Let $\alpha=n_{0}a^{2}$ and $\beta=kmb^{2}=\alpha rJ$, and let $q_{\omega}=d^{-1/2}\sum^{d}_{j=1}\omega_{j}e_{j}$.
As in the rank one case, observe that $\beta$ is an eigenvalue on
the $(d-1)$-dimensional subspace spanned by $e_{1},\ldots,e_{d}$
and orthogonal to $q_{\omega}$. Observe that $\alpha J$ is an eigenvalue
on the span of $f_{1},\ldots,f_{k-1}$. The orthogonal complement
of the direct sum of those two subspaces is spanned by $e_{0}$ and
$q_{\omega}$. $A^{\top}_{\omega}A_{\omega}$ after projection to
the span of $e_{0},q_{\omega}$ becomes 
\[
K=\begin{pmatrix}\alpha & abt\sqrt{d}\\
abt\sqrt{d} & \beta
\end{pmatrix},\qquad(abt\sqrt{d})^{2}=\alpha\beta\rho.
\]
Using $J+r(1-\rho)=1+rJ$, its characteristic polynomial factors as
\[
\det(\lambda I_{2}-K)=(\lambda-\alpha J)(\lambda-\alpha r(1-\rho)).
\]
Consequently the full spectrum of $A^{\top}_{\omega}A_{\omega}$,
in decreasing order, is 
\[
\underbrace{\lambda_{+},\ldots,\lambda_{+}}_{k\text{ times}},\quad\underbrace{\beta,\ldots,\beta}_{d-1\text{ times}},\quad\lambda_{-},\qquad\lambda_{+}=\alpha J,\quad\lambda_{-}=\alpha r(1-\rho)=\frac{\beta(1-\rho)}{J}.
\]
We have $0<\lambda_{-}<\beta<\lambda_{+}$ and $\lambda_{k+1}/\lambda_{k}=\beta/\lambda_{+}=r$.
The leading eigenvector from the span of $e_{0},q_{\omega}$ is 
\[
v_{\omega}=\cos\theta\,e_{0}+\sin\theta\,q_{\omega},\qquad0\le\theta\le\frac{\pi}{2}.
\]

We have 
\[
K\begin{pmatrix}\cos\theta\\
\sin\theta
\end{pmatrix}=\alpha J\begin{pmatrix}\cos\theta\\
\sin\theta
\end{pmatrix}
\]
so $abt\sqrt{d}\cos\theta+\beta\sin\theta=\alpha J\sin\theta$ and
$\tan\theta=\frac{abt\sqrt{d}}{\alpha J-\beta}=\frac{abt\sqrt{d}}{\Delta_{k}}$.
In particular, 
\begin{equation}
\tan^{2}\theta=\frac{da^{2}b^{2}t^{2}}{\Delta^{2}_{k}}=\frac{rt^{2}}{kJm^{2}(1-r)^{2}}\le\frac{r}{64kJ}\le\frac{1}{64}.\label{eq:lb-rank-k-angle}
\end{equation}
The projector onto the subspace spanned by the top $k$ eigenvectors
of $A^{\top}_{\omega}A_{\omega}$ is
\[
P_{\omega}=v_{\omega}v^{\top}_{\omega}+\sum^{k-1}_{\ell=1}f_{\ell}f^{\top}_{\ell}.
\]
\textbf{Coherence of all left singular vectors.} Let $H_{s}$ denote
the real Fourier basis of size $s$ used in the rank-one proof, with
$H_{1}=(1)$. Thus its first column is constant $1/\sqrt{s}$ and
every entry has magnitude at most $\sqrt{2/s}$. Set $R=H^{\top}_{k}$,
so the first row of $R$ is constant $1/\sqrt{k}$. We can choose
any orthonormal basis of the span of $v_{\omega},f_{1},\ldots,f_{k-1}$
as the $k$ leading right singular vectors. Choose the $k$ leading
right singular vectors to be the columns of 
\[
\begin{pmatrix}v_{\omega} & f_{1} & \cdots & f_{k-1}\end{pmatrix}R.
\]
For the eigenspace corresponding to the eigenvalue $\beta$, choose
the following $d-1$ vectors as the right singular vectors: $\sum^{d}_{j=1}\omega_{j}(H_{d})_{j\ell}e_{j}$,
$\ell=2,\ldots,d$. Finally, choose $w_{\omega}=-\sin\theta\,e_{0}+\cos\theta\,q_{\omega}$
for the last right singular vector corresponding to the eigenvalue
$\lambda_{-}$.

We use $U_{\omega}=A_{\omega}V_{\omega}\Sigma^{-1}$ to bound every
entry $\left(U_{\omega}\right)_{i,\ell}$. First consider $1\le\ell\le k$
for the leading singular vectors. Consider row $i\le n_{0}$ of $A_{\omega},$which
is of the form $ae_{0}+\xi be_{j}$, where $\xi\in\{-1,1\}$ and $j\in\{1,\ldots,d\}$.
We have 
\begin{align*}
\left|\left(U_{\omega}\right)_{i,\ell}\right| & \le\frac{\left(a\left|\left\langle e_{0},v_{\omega}\right\rangle \right|+b\left|\left\langle e_{j},v_{\omega}\right\rangle \right|\right)R_{1,\ell}}{\sqrt{\lambda_{+}}}\\
 & =\frac{a\cos\theta+b\sin\theta/\sqrt{d}}{\sqrt{k\lambda_{+}}}=\frac{\cos\theta+\sqrt{rJ/k}\sin\theta}{\sqrt{kn_{0}J}}\le\frac{1+r/(8k)}{\sqrt{kn_{0}J}}\le\sqrt{\frac{2}{kn_{0}}}.
\end{align*}
The second inequality above used $\sin\theta\le\tan\theta\le\sqrt{r/(kJ)}/8$.
For a row $i$ of the form $be_{j}$ in the next group of rows, we
get a bound that is no larger. 
\[
\left|\left(U_{\omega}\right)_{i,\ell}\right|\le\frac{\left(b\left|\left\langle e_{j},v_{\omega}\right\rangle \right|\right)R_{1,\ell}}{\sqrt{\lambda_{+}}}\le\sqrt{\frac{2}{kn_{0}}}
\]
For a row $i$ of the form $a\sqrt{J}f_{h}$, we have 
\[
\left|\left(U_{\omega}\right)_{i,\ell}\right|=\frac{a\sqrt{J}\left|R_{h+1,\ell}\right|}{\sqrt{\lambda_{+}}}\le\sqrt{\frac{2}{kn_{0}}}
\]
 Next consider $k+1\le\ell\le d+k-1$ for the singular vectors corresponding
to the singular value $\sqrt{\beta}$. Consider row $i\le n_{0}$
of $A_{\omega},$which is of the form $ae_{0}+\xi be_{j}$, where
$\xi\in\{-1,1\}$ and $j\in\{1,\ldots,d\}$. We have 
\[
\left|\left(U_{\omega}\right)_{i,\ell}\right|=\frac{b\left|(H_{d})_{j,\ell-k+1}\right|}{\sqrt{\beta}}\le\frac{b\sqrt{2/d}}{\sqrt{\beta}}=\sqrt{\frac{2}{kn_{0}}}.
\]
 For a row $i$ of the form $be_{j}$ in the next group of rows, we
get the same bound. 
\[
\left|\left(U_{\omega}\right)_{i,\ell}\right|=\frac{b\left|(H_{d})_{j,\ell-k+1}\right|}{\sqrt{\beta}}\le\frac{b\sqrt{2/d}}{\sqrt{\beta}}=\sqrt{\frac{2}{kn_{0}}}
\]
 For a row $i$ of the form $a\sqrt{J}f_{h}$, we have $\left(U_{\omega}\right)_{i,\ell}=0$. 

Finally consider $\ell=d+k$ for the last singular vector. Consider
row $i\le n_{0}$ of $A_{\omega},$which is of the form $ae_{0}+\xi be_{j}$,
where $\xi\in\{-1,1\}$ and $j\in\{1,\ldots,d\}$. We have 
\begin{align*}
\left|\left(U_{\omega}\right)_{i,d+k}\right| & \leq\frac{a\sin\theta+b\cos\theta/\sqrt{d}}{\sqrt{\lambda_{-}}}\\
 & =\frac{a\sin\theta+a\sqrt{rJ/k}\cos\theta}{a\sqrt{n_{0}r(1-\rho)}}\\
 & =\sqrt{\frac{J}{kn_{0}(1-\rho)}}\left(\cos\theta+\sqrt{\frac{k}{rJ}}\sin\theta\right)\le\sqrt{\frac{J}{kn_{0}(1-\rho)}}\left(1+\frac{1}{8J}\right)\le\sqrt{\frac{2}{kn_{0}}}.
\end{align*}
In the last inequality, we used $\sin\theta\le\tan\theta\le\sqrt{r/(kJ)}/8$
(from (\ref{eq:lb-rank-k-angle})) and due to our choice of $m\ge8t/(1-r)$,
we have $1-\rho=1-t^{2}/(km^{2})\ge\frac{63}{64}$ and $J=1+\frac{r\rho}{1-r}\le1+\frac{r(1-r)}{64}\le\frac{257}{256}$.

For a row $i$ of the form $be_{j}$ in the next group of rows, we
have a bound that is no larger.
\[
\left|\left(U_{\omega}\right)_{i,d+k}\right|=\frac{b\cos\theta/\sqrt{d}}{\sqrt{\lambda_{-}}}\le\sqrt{\frac{2}{kn_{0}}}
\]
For a row $i$ of the form $a\sqrt{J}f_{h}$, we have $\left(U_{\omega}\right)_{i,d+k}=0$.
This proves the upper bound on the coherence. The lower bound $\|U_{\omega}\|^{2}_{\infty}\ge1/n$
follows because every column of $U_{\omega}$ has norm one.

\textbf{Privacy and subspace error.} For an output $\widehat{X}$,
set $\widehat{P}=\widehat{X}\widehat{X}^{\top}$ and decode 
\[
\widehat{\omega}_{j}(\widehat{X})=\operatorname{sign}(e^{\top}_{0}\widehat{P}e_{j}),\qquad j=1,\ldots,d,
\]
with a fixed convention for the sign of zero. Flipping one bit of
$\omega$ changes exactly $t$ rows; the appended rows are independent
of $\omega$. Thus the same group-privacy parameters $\epsilon_{t}\le1$
and $\delta_{t}\le e/100$ apply. The argument leading to (\ref{eq:lb-hamming})
in the rank-one proof uses only these parameters and a decoder into
$\{-1,1\}^{d}$, so it applies to this matrix-valued output as well.
Consequently, for at least half of the choices of $\omega$, 
\[
\Pr\left[d_{H}(\widehat{\omega}(\widehat{X}),\omega)\ge d/32\right]\ge\frac{1}{4}.
\]
Each incorrectly decoded bit $j$ contributes at least $\cos^{2}\theta\sin^{2}\theta/d$
to the squared Euclidean norm of row $e^{\top}_{0}(\widehat{P}-P_{\omega})$,
since $e^{\top}_{0}P_{\omega}e_{j}=\cos\theta\sin\theta\,\omega_{j}/\sqrt{d}$.
The two projectors have equal rank, so their spectral norm distance
equals the sine of their largest principal angle. Hence 
\begin{align}
\left\Vert (I-\widehat{P})V_{k,\omega}\right\Vert ^{2}_{2} & =\|\widehat{P}-P_{\omega}\|^{2}_{2}\nonumber \\
 & \ge\left\Vert e^{\top}_{0}(\widehat{P}-P_{\omega})\right\Vert ^{2}_{2}\ge\frac{\cos^{2}\theta\sin^{2}\theta}{d}d_{H}(\widehat{\omega}(\widehat{X}),\omega).\label{eq:lb-rank-k-projector}
\end{align}

\textbf{Comparison with the claimed rate.} Since $t\ge1/(2\epsilon)$,
(\ref{eq:lb-rank-k-angle}) gives 

\begin{align*}
\sin^{2}\theta & =\frac{\tan^{2}\theta}{1+\tan^{2}\theta}\ge\frac{da^{2}b^{2}t^{2}}{\Delta^{2}_{k}}\cdot\frac{1}{1+1/64}\ge\frac{16}{65}\frac{da^{2}b^{2}}{\epsilon^{2}\Delta^{2}_{k}}\\
\cos^{2}\theta & =\frac{1}{1+\tan^{2}\theta}\ge\frac{64}{65}\\
\frac{\cos^{2}\theta\sin^{2}\theta}{32} & \ge\frac{32}{4225}\frac{da^{2}b^{2}}{\epsilon^{2}\Delta^{2}_{k}}
\end{align*}
On the other hand, $(d+k)a^{2}J=(d+k)kJ/(k+drJ)\ge k\ge1$, and the
coherence bound $\|U_{\omega}\|^{2}_{\infty}\le2/(kn_{0})$ implies
\[
(d+k)k\lambda_{+}\|U_{\omega}\|^{2}_{\infty}+1\le2(d+k)a^{2}J+1\le3(d+k)a^{2}J.
\]
It follows that 
\begin{align*}
\frac{\cos^{2}\theta\sin^{2}\theta}{32} & \ge\frac{32}{12675J}\cdot\frac{d}{d+k}\cdot b^{2}\frac{(d+k)k\lambda_{+}\|U_{\omega}\|^{2}_{\infty}+1}{\epsilon^{2}\Delta^{2}_{k}}\\
 & \ge\frac{1}{1024}\cdot\frac{d}{d+k}\cdot\frac{dr}{k+dr}\cdot\frac{(d+k)k\lambda_{+}\|U_{\omega}\|^{2}_{\infty}+1}{\epsilon^{2}\Delta^{2}_{k}},
\end{align*}
where we used $J\le257/256$ and $b^{2}=drJ/(k+drJ)\ge dr/(k+dr)$.
Together with the Hamming bound and (\ref{eq:lb-rank-k-projector}),
this proves (\ref{eq:lb-rank-k-general}). Finally, if $d\ge k$ and
$r\ge c_{0}k/(d+k)$, then 
\[
\frac{d}{d+k}\ge\frac{1}{2},\qquad\frac{dr}{k+dr}\ge\frac{c_{0}}{2+c_{0}},
\]
which proves (\ref{eq:lb-rank-k-matching}). 
\end{proof}

\end{document}